\documentclass[a4paper,9.5pt]{article}
\usepackage{etoolbox}
\newtoggle{longversion}
\toggletrue{longversion}

\usepackage{url}
\usepackage{microtype}

\usepackage{hyperref}
\hypersetup{
  colorlinks   = true, %Colours links instead of ugly boxes
  linkcolor    = black!50!brown, %Colour of internal links
  citecolor   = black!50!brown, %Colour of citations
}

\usepackage{amsmath}
\usepackage{amssymb}
\usepackage{amsthm}
\usepackage{verbatim}
\usepackage{bm}
\usepackage{color,graphicx,xcolor}
\usepackage{mdwtab}
\usepackage{mathtools,tikz}
\usepackage{hhline}
\usepackage{multirow}
\usepackage{pdfpages}
\usepackage{enumitem}  
\usepackage{caption}
\usepackage{subcaption}
\DeclareCaptionFormat{myformat}{\fontsize{8}{9}\selectfont#1#2#3}
\usepackage{appendix}
\usepackage{authblk}

\iftoggle{longversion}{
\linespread{0.995}
}
{
\setlength{\abovedisplayskip}{0.3em}
\setlength{\belowdisplayskip}{0.3pt}
\setlength{\abovedisplayshortskip}{0pt}
\setlength{\belowdisplayshortskip}{0pt}
\twocolumn
\newcommand{\longonly}[1]{}
\newcommand{\shortonly}[1]{#1}

}

\iftoggle{longversion}{
\onecolumn
\newcommand{\longonly}[1]{#1} %stuff to put only in long version
\newcommand{\shortonly}[1]{} %stuff to put only in short version
}
{
\setlength{\abovedisplayskip}{0.3em}
\setlength{\belowdisplayskip}{0.3pt}
\setlength{\abovedisplayshortskip}{0pt}
\setlength{\belowdisplayshortskip}{0pt}
\twocolumn
\newcommand{\longonly}[1]{}
\newcommand{\shortonly}[1]{#1}

}

\usepackage{cite}
\usepackage{amsmath,amssymb,amsfonts,amsthm,xspace}
\usepackage{algorithmic}
\usepackage{graphicx}
\usepackage{textcomp}

\usepackage{chngcntr} 
\counterwithin{section}{part}

\usetikzlibrary{arrows,shapes.geometric}
\allowdisplaybreaks
\theoremstyle{definition}
\newtheorem{example}{Example}[part]
\newtheorem{remark}{Remark}[part]
\newtheorem{thm}{Theorem}[part]
\newtheorem{defn}{Definition}[part]

\newtheorem*{defn*}{Definition}
\newtheorem*{lemma*}{Lemma}

\newtheorem{lemma}[thm]{Lemma}

\newtheorem{prop}[thm]{Proposition}
\newtheoremstyle{thmnum}{\topsep}{\topsep}{\itshape}{0pt}{\bfseries}{.}{ }{\thmname{#1}\thmnote{ \bfseries #3}}
\theoremstyle{thmnum}
\newtheorem{duplicatelemma}{Lemma}

\usepackage{pgffor}
\foreach \x in {a,...,z}{%
\expandafter\xdef\csname vec\x \endcsname{\noexpand\ensuremath{\noexpand\bm{\x}}}
}

\foreach \x in {A,...,Z}{%
\expandafter\xdef\csname vec\x \endcsname{\noexpand\ensuremath{\noexpand\bm{\x}}}
}

\foreach \x in {A,...,Z}{%
\expandafter\xdef\csname c\x \endcsname{\noexpand\ensuremath{\noexpand\mathcal{\x}}}
}

\foreach \x in {A,...,Z}{%
\expandafter\xdef\csname bb\x \endcsname{\noexpand\ensuremath{\noexpand\mathbb{\x}}}
}

\newcommand{\Prob}{\ensuremath{{\mathbb P}}}

\newcommand{\defineqq}{\ensuremath{\stackrel{\textup{\tiny def}}{=}}}

\def\ttl{\tilde{\tl}}
\def\tlx{\ensuremath{\tilde{x}}}
\def\tla{\ensuremath{\tilde{a}}}
\def\tlb{\ensuremath{\tilde{b}}}
\def\tlc{\ensuremath{\tilde{c}}}
\newcommand{\norm}[1]{\left\lVert#1\right\rVert_2}
\def\perm{\Pi}
\def\permchoice{\pi}
\def\Noisevec{W^\tl}

\usetikzlibrary{shapes,arrows,positioning,decorations.pathreplacing,calc}

\def\tl{t} %taglength for authentication code
\def\msg{\ensuremath{m}} %message
\def\msgh{\ensuremath{\hat{\msg}}} %reconstruction
\def\Msg{\ensuremath{M}} %message random variable
\def\Msgh{\ensuremath{\hat{\Msg}}} %reconstruction random variable
\def\nummsg{\mbox{$N$}} %number of messages
\def\msgset{\ensuremath{\mathcal{M}}} %set of messages
\def\state{\ensuremath{s}}
\def\State{\ensuremath{S}}
\def\tagchannel{\ensuremath{W_{\rV|\tU,\State}}}

\def\pfa{P_{\mathsf{R}|\state_0}}
\def\pwa{P_{\mathsf{A}}}
\def\peauth{P_{\textrm{e,auth}}}
\def\tagenc{\ensuremath{A}} %tag encoder e.g. \tagenc_1

\def\tagdec{\ensuremath{\alpha}} %tag decoder

\def\idenc{\ensuremath{B}} %identification encoder e.g. 

\def\iddec{\ensuremath{\beta}} %identification decoder

\def\tu{u} %transmitted tag letter
\def\vectu{\ensuremath{\tu^{\tl}}} %transmitted vector for the authentication code
\def\rv{\ensuremath{v}}
\def\vecrv{\ensuremath{\rv^{\tl}}} %received vector for the authentication code
\def\tU{U} %random transmitted symbol
\def\vectU{\ensuremath{\tU^{\tl}}} %random transmitted vector
\def\rV{\ensuremath{V}}
\def\vecrV{\ensuremath{\rV^{\tl}}} %received vector

\def\accept{\ensuremath{\mathsf{A}}}
\def\reject{\ensuremath{\mathsf{R}}}

\def\deterministic{\ensuremath{\mathrm{deterministic}}}
\def\random{\ensuremath{\mathrm{random}}}
\def\strong{\ensuremath{\mathrm{strong}}}
\def\weak{\ensuremath{\mathrm{weak}}}

\def\auth{\ensuremath{\mathrm{auth}}}
\def\authcomm{\ensuremath{\mathrm{auth-comm}}}
\def\nonadv{\ensuremath{\mathrm{non-adv}}}
\def\rand{\ensuremath{\mathrm{rand}}}
\def\det{\ensuremath{\mathrm{det}}}

\newcommand{\mac}{\ensuremath{\text{DM-MAC}}\xspace}
\newcommand{\mch}{\ensuremath{W_{Z|X,Y}}\xspace}
\newcommand{\mach}{\ensuremath{W_{Y|X_1,X_2,X_3}}\xspace}
\newcommand{\machijk}{\ensuremath{W_{Y|X_i,X_j,X_k}}\xspace}

\newcommand{\na}{\ensuremath{\text{no attack}}\xspace}
\newcommand{\malone}{\ensuremath{\text{mal 1}}}
\newcommand{\maltwo}{\ensuremath{\text{mal 2}}}

\newcommand{\set}[1]{\left\{#1\right\}}

\def\BibTeX{{\rm B\kern-.05em{\sc i\kern-.025em b}\kern-.08em
    T\kern-.1667em\lower.7ex\hbox{E}\kern-.125emX}}
\begin{document}

\title{Byzantine Multiple Access
}

\author[1]{Neha Sangwan}
\author[2]{Mayank Bakshi}
\author[3]{Bikash Kumar Dey}
\author[1]{Vinod M. Prabhakaran}
\affil[1]{Tata Institute of Fundamental Research. Mumbai, India}
\affil[2]{Chinese University of Hong Kong, Hong Kong}
\affil[3]{Indian Institute of Technology Bombay, Mumbai, India}

\makeatletter
\patchcmd{\@maketitle}
  {\addvspace{0.5\baselineskip}\egroup}
  {\addvspace{-1.5\baselineskip}\egroup}
  {}
  {}
\makeatother

\maketitle

\begin{abstract}
	We study communication over multiple access channels (MAC) where one of the users is possibly adversarial. When all users behave non-adversarially, we want their messages to be decoded reliably. When an adversary is present, we consider two different decoding guarantees. 

In part I, we require that the honest users' messages be decoded reliably. We study the three-user MAC\footnote{It turns out that the capacity region for the two-user MAC follows from the capacity of the point-to-point arbitrarily varying channel.}. We characterize the capacity region for randomized codes (where each user shares an independent secret key with the receiver). We also study the capacity region for deterministic codes. We obtain necessary conditions including a new non-symmetrizability condition for this capacity region to be non-trivial. We show that when none of the users are symmetrizable, the randomized coding capacity region is also achievable with deterministic codes. This is analogous to the result of Ahlswede and Cai (1991) for arbitrarily varying MAC.

	In part II, we consider the weaker goal of authenticated communication where we only require that an adversarial user must not be able to cause an undetected error on the honest users' messages. Therefore, when an adversarial user is present, it is sufficient to detect its presence without decoding the message of the honest user. For the two-user MAC, we show that the following three-phase scheme is rate-optimal: a standard MAC code is first used to achieve unauthenticated communication; this is followed by two authentication phases where each user authenticates their message treating the other user as a possible adversary. We show that the authentication phases can be very short since this form of authentication itself, when possible, can be achieved for message sets whose size grow doubly exponentially in blocklength. This leads to our result that the authenticated communication capacity region of a discrete memoryless MAC is either zero or the (unauthenticated) MAC capacity region itself. This also, arguably, explains the similar nature of authenticated communication capacity of a discrete memoryless point-to-point adversarial channel recently found by Kosut and Kliewer (ITW, 2018). We also obtain analogous results for additive Gaussian noise channels. 

\end{abstract}

\newpage
\vspace*{\fill}
\begin{center}
\huge\bfseries{ Part I\\ Reliable Communication in the Presence of an Adversarial User\par}
\end{center}
\vspace*{\fill}
\newpage

\part{Reliable Communication in the Presence of an Adversarial User}

\section{Introduction} 
Consider a multiple access channel (MAC) in which at most one of the users may behave adversarially (Fig.~\ref{fig:advMAC}). If all users are non-adversarial, we require that their messages be reliably decoded. However, if one of the users is adversarial, the decoder must correctly recover the message of all the other (honest) users. 

For the two-user MAC, clearly, each user can at least achieve the capacity of the arbitrarily varying channels (AVC)~\cite{BlackwellBTAMS60} that amounts to treating the other user's input as channel state. Further, it is also easy to see that no higher rate is possible as the other user can behave exactly like an adversary in the AVC setup. Thus, the capacity region is the rectangular region defined by the AVC capacities of the two users' channels, \emph{i.e.}, there is no trade-off between the rates\footnote{This observation holds true under deterministic coding, stochastic encoding (where the encoders have private randomness), and randomized coding settings under both maximum and average probabilities of error. A similar observation can be made for any $k$-user MAC where up to $k-1$ users may adversarially collude}. Here we study the three-user MAC and characterize the randomized coding capacity. We also characterize the deterministic coding capacity under average probability of error for most channels.\footnote{Our characterization is incomplete for channels in which some, but not all users are symmetrizable. See remark~\ref{rem:gap}.}

Our problem is closely related to the two-user arbitrarily varying MACs
(AV-MAC), which has an {\em external} adversary (jammer) unlike in our problem.
For this setup, Jahn~\cite{Jahn81} obtained the randomized coding capacity
where each user has an independent source of randomness (secret key) which is
shared with the receiver. He also showed that this region is also the
deterministic coding capacity region under average probability of error
whenever the latter has a non-empty interior, a result along the lines of
Ahlswede's dichotomy for arbitrarily varying channels~\cite{Ahlswede78}.
Gubner~\cite{Gubner90} proved necessary conditions (non-symmetrizability
conditions) for capacity to be non-zero. Gubner's notion
of non-symmetrizability of AV-MACs generalizes the corresponding notion of
AVCs. In particular, he defined an AV-MAC to be non-symmetrizable if it is not
symmetrizable in any of the following three senses -- \emph{(i)} the AV-MAC can
be symmetrized for the first user's input while keeping second user's input
unchanged, \emph{(ii)} the AV-MAC can be symmetrized for the second user's
input while keeping first user's input unchanged, and \emph{(iii)} the AV-MAC
can be simultaneously symmetrized for both users' inputs. Ahlswede and
Cai~\cite{AhlswedeC99} showed that Gubner's necessary conditions are also
sufficient for the deterministic coding capacity region to have a non-empty
interior. Recently, Pereg and Steinberg~\cite{PeregS19} obtained the capacity region for
arbitrarily varying MAC with state constraints. They also closed the gap for
the problem without state constraints by addressing the case where exactly one
of the users has zero capacity.

\begin{figure}\centering
\resizebox{0.99\columnwidth}{!}{\begin{tikzpicture}[scale=0.5]
	\draw (2,0) rectangle ++(3,2) node[pos=.5]{User-3};
	\draw (2,3) rectangle ++(3,2) node[pos=.5]{User-2};
	\draw (2,6) rectangle ++(3,2) node[pos=.5]{User-1};
	\draw (8.5,2.5) rectangle ++(4.5,3) node[pos=.5]{$\mach$};
	\draw (16,3) rectangle ++(3,2) node[pos=.5]{$\phi$};
	\draw[->] (1,1) node[anchor=east]{$\msg_3$} -- ++ (1,0) ;
	\draw[->] (5,1) -- node[midway,below,sloped] {$X_3^n$} ++ (3.5,2);
	
	\draw[->] (1,4) node[anchor=east]{$\msg_2$} -- ++ (1,0) ;
	\draw[->] (5,4) -- node[above] {$X_2^n$} ++ (3.5,0);

	\draw[->] (1,7) node[anchor=east]{$\msg_1$} -- ++ (1,0) ;
	\draw[->] (5,7) -- node[midway,above,sloped] {$X_1^n$} ++ (3.5,-2);

	\draw[->] (13,4) -- node[above] {$Y^n$} ++ (3,0);
	\draw[->] (19,4) -- ++ (1,0) node[anchor=west]{\begin{array}{c}\msgh_1\\ \msgh_2\\ \msgh_3\end{array}};
\end{tikzpicture}}
\caption{Adversarial MAC: At most one user may be adversarial. Reliable decoding of the messages of all honest users is required. Clearly, no decoding guarantees are given for an adversarial user.}\label{fig:advMAC}
\end{figure}
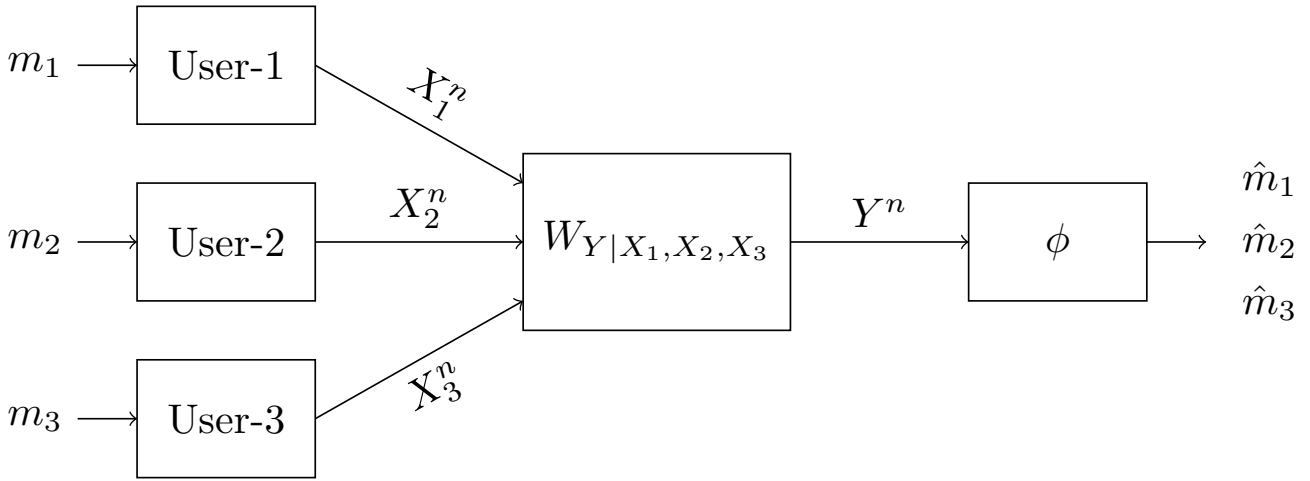

There are several other related works. Wiese and Boche~\cite{WieseB13} considered two-user arbitrarily varying multiple access channels with conferencing encoders. La and Anantharam~\cite{LaADIMACS04} studied MACs with strategic users modeled as a cooperative game. In part II of this paper (also see~\cite{NehaBDP19}), we consider authenticated communication over a two-user
MAC where at most one user may be adversarial. Compared to the one we consider here, authentication is a weaker goal -- ensure reliable decoding when both users are non-adversarial, however, if one of the users behaves adversarially, the decoder (with high probability) must either output the correct message for the honest user or abort, i.e., an adversary should not be able to cause an undetected erroneous output for the honest user. Byzantine attacks on the nodes and edges of networks have been studied under omniscient and weaker adversarial models in~\cite{KosutTTIT14} and~\cite{JaggiLKHKM07,WangSK10}, respectively.

Returning to our problem, the capacity characterizations for AV-MACs~\cite{Jahn81,Gubner90,AhlswedeC99,PeregS19} shed light on the nature of capacity region for our problem. It is clear that, at the very least, if a rate triple $(R_1,R_2,R_3)$ is achievable for our setting, then the rate pair $(R_2,R_3)$ must also be achievable over the two-user AV-MAC $W^{(1)}$ that is formed by treating user-$1$'s input as the adversarially chosen state. Similarly, the rate pairs $(R_1,R_3)$ and $(R_1,R_2)$ must also be achievable over corresponding AV-MACs $W^{(2)}$ and $W^{(3)}$ respectively.\footnote{In fact, a stronger necessary condition follows by noting that the encoder of each user must not depend on the knowledge of which user, if any, is the adversary. Thus, as in compound channels, the same code should work regardless of the $W^{(k)}$ that represents the actual channel. We use this observation in our converse arguments.} 

In particular, it follows that the capacity region for our problem has an empty interior whenever any of the AV-MACs $W^{(1)}$, $W^{(2)}$, and $W^{(3)}$ are symmetrizable in the sense of Gubner's characterization. One of the main contributions of our work is to show that, in addition to the nine symmetrizability conditions inherited from above (three for each $W^{(k)}$), fully characterizing non-symmetrizability of our adversarial MAC relies on also excluding three additional symmetrizability conditions (Eq.~\eqref{eq:symm3}). Roughly speaking, each of these conditions reflect whether or not an adversarial user at a node $k$ can attack in a manner that is  also consistent with an adversarial user at a node $j\neq k$ while resulting in a decoding ambiguity about the remaining user's message (see Figure~\ref{fig:symm3}). Example~\ref{ex:symmetrizable} shows that the new symmetrizability conditions proposed in this paper are not redundant in view of the symmetrizability conditions derived from two-user AV-MACs. 

On the achievability side, we show that as long as our adversarial MAC is non-symmetrizable, \emph{i.e.}, none of the twelve symmetrizability conditions hold,  the capacity region for our problem equals the \emph{randomized} capacity region.\footnote{As in~\cite{Jahn81}, in the randomized setting, each user may share an unlimited number of random bits with the decoder that are unknown to other users.}

\section{Problem Setup and Main Results}
\subsection{System model}\label{section:MAC2_model}
Consider the three-user MAC setup shown in Fig.~\ref{fig:advMAC}. The memoryless channel \mach has input alphabets $\mathcal{X}_1,\mathcal{X}_2,\mathcal{X}_3$, and output alphabet ${\mathcal Y}$. 
\shortonly{An $(\nummsg_1,\nummsg_2,\nummsg_3,n)$ deterministic {\em code} for the multiple access channel \mach consists of: \vspace{-0.25em}
\begin{enumerate}[label=(\roman*)]
\item three message sets, $\mathcal{M}_i = \{1,\ldots,\nummsg_i\}$, $i=1,2,3$,
\item three encoders, $f_{i}:\mathcal{M}_i\rightarrow \mathcal{X}_i^n$, $i=1,2,3$, and
\item a decoder, $\phi:\mathcal{Y}^n\rightarrow\mathcal{M}_1\times\mathcal{M}_2\times\mathcal{M}_3.$ 
\end{enumerate}}

\longonly{\begin{defn}[Deterministic code]
An $(\nummsg_1,\nummsg_2,\nummsg_3,n)$ deterministic {\em code} for the \mac \mach consists of the following: 
\begin{enumerate}[label=(\roman*)]
\item three message sets, $\mathcal{M}_i = \{1,\ldots,\nummsg_i\}$, $i=1,2,3$,
\item three encoders, $f_{i}:\mathcal{M}_i\rightarrow \mathcal{X}_i^n$, $i=1,2,3$, and
\item a decoder, $\phi:\mathcal{Y}^n\rightarrow\mathcal{M}_1\times\mathcal{M}_2\times\mathcal{M}_3.$ 
\end{enumerate}
\end{defn}}
 
We define the average probability error $P_e$ as the maximum of average probabilities under the four scenarios where at most one user is adversarial. Let $(\Msgh_1,\Msgh_2,\Msgh_3)=\phi(Y^n)$. 
\[ P_e(f_1,f_2,f_3,\phi)\defineqq\max \{ P_{e,0}, P_{e,1}, P_{e,2}, P_{e,3} \},\]
where the terms on the right-hand side are defined below. Note that our notation suppresses their dependence on the code. $P_{e,0}$ is the average probability of error when none of the users are adversarial,
\begin{align*}
&P_{e,0}\defineqq \frac{1}{\nummsg_1\nummsg_2\nummsg_3}
\sum_{m_1,m_2,m_3} e_0(m_1,m_2,m_3),\text{ where}\\
&e_0(m_1,m_2,m_3)=
\Prob\Big( (\Msgh_1,\Msgh_2,\Msgh_3) \neq (m_1,m_2,m_3) \Big|
\\ &\qquad\qquad\qquad 
X_1^n=f_1(m_1),  X_2^n=f_2(m_2), X_3^n=f_3(m_3)\Big).
\end{align*}
$P_{e,i}$, $i=1,2,3$ is the average error when user-$i$ is adversarial. $P_{e,1}$ is as below, $P_{e,2},P_{e,3}$ can be defined similarly.
\begin{align*}
&P_{e,1}\defineqq \max_{x_1^n} \frac{1}{\nummsg_2\nummsg_3} \sum_{m_2,m_3} e_1(x_1^n,m_2,m_3),\text{ where}\\
&e_1(x_1^n,m_2,m_3)= \Prob\Big( (\Msgh_2,\Msgh_3) \neq (m_2,m_3) \Big|\\
&\qquad\qquad\qquad\qquad X_1^n=x_1^n, X_2^n=f_2(m_2), X_3^n=f_3(m_3) \Big).
\end{align*}
We emphasize that the decoder is unaware of whether any of the users is adversarial and the identity of the adversarial user (if any). We also note that $P_{e,0}\leq P_{e,1}+P_{e,2}+P_{e,3}$.

We say a rate triple $(R_1,R_2,R_3)$ is {\em achievable} if there is a sequence of $(\lfloor2^{nR_1}\rfloor,\lfloor2^{nR_2}\rfloor,\lfloor2^{nR_3}\rfloor,n)$ codes  $\{f_1^{(n)},f_2^{(n)},f_3^{(n)},\allowbreak\phi^{(n)}\}_{n=1}^\infty$ such that $\lim_{n\rightarrow\infty}P_{e}(f_1^{(n)},f_2^{(n)},f_3^{(n)},\phi^{(n)})\rightarrow0$. The {\em deterministic coding capacity region} $\mathcal{R}_{\deterministic}$ is the closure of the set of all achievable rate triples.

\longonly{
\begin{defn}[Randomized code]
An $(\nummsg_1,\nummsg_2,\nummsg_3,n)$ randomized {\em code} for the \mac \mach consists of the following: 
\begin{enumerate}[label=(\roman*)]
\item three message sets, $\mathcal{M}_i = \{1,\ldots,\nummsg_i\}$, $i=1,2,3$,
\item three independent randomized encoders, $F_{i}:\mathcal{M}_i\rightarrow \mathcal{X}_i^n$ where $F_i$ takes values in $\mathcal{F}_i,\, i = 1,2,3$ and
\item a privately randomized decoder, $\Phi:\mathcal{Y}^n\times\mathcal{F}_1\times\mathcal{F}_2\times\mathcal{F}_3\rightarrow\mathcal{M}_1\times\mathcal{M}_2\times\mathcal{M}_3$ where $\Phi(y^n,F_1,F_2,F_3) = (\Psi_1(y^n,F_1,F_2,F_3),\Psi_2(y^n,F_1,F_2,F_3),\Psi_3(y^n,F_1,F_2,F_3))$ for some randomized functions $\Psi_i:\mathcal{Y}^n\times\mathcal{F}_1\times\mathcal{F}_2\times\mathcal{F}_3\rightarrow\mathcal{M}_1\times\mathcal{M}_2\times\mathcal{M}_3, \, i = 1,2,3$.
\end{enumerate}
\end{defn}}

\shortonly{A randomized code consists of {\em independent} random encoding maps
$F_1,F_2,F_3$ and a decoder $\Phi$ (which may depend on $F_1,F_2,F_3$), i.e.,
the encoders randomize independently of each other and their randomization is
available to the decoder. }This is analogous to the randomized code of
Jahn~\cite{Jahn81} for 2-user AV-MACs. Probability of error and capacity region
for randomized codes can be defined similarly\footnote{Clearly, for randomized
codes, the capacity region will remain unchanged for maximum and average
probabilities of error criteria. Hence we only consider the average error
criterion here.}. \shortonly{Please see the extended draft for the
definitions~\cite{ExtendedDraft}.} We will consider the natural setting where
the adversarial user has access to its random encoding map and denote its {\em
randomized coding capacity region} by $\mathcal{R}_{\random}$.  We also
consider two slightly artificial settings -- (i) adversarial user-$i$ is a {\em
strong} adversary if it may choose its encoding map $f_i$ in addition to the
input vector $x_i^n$, and (ii) an adversary is a {\em weak} adversary if it
does not have access to its random encoding map when choosing its input vector.
We denote the corresponding randomized coding capacity region by
$\mathcal{R}_{\random}^{\strong}$ and $\mathcal{R}_{\random}^{\weak}$
respectively. Clearly,
$\mathcal{R}_{\random}^{\strong} \subseteq \mathcal{R}_{\random} \subseteq \mathcal{R}_{\random}^{\weak}$. 
\longonly{ 
Analogous to the deterministic case, the average probability error $P^{\text{strong}}_{e}$ under the strong adversary is defined as 
\[ P^{\text{strong}}_{e}(F_1,F_2,F_3,\Phi)\defineqq\max \{ P^{\text{rand}}_{e,0}, P^{\text{strong}}_{e,1}, P^{\text{strong}}_{e,2}, P^{\text{strong}}_{e,3} \},\]
where 
\begin{align*}
&P^{\text{rand}}_{e,0}\defineqq \frac{1}{\nummsg_1\nummsg_2\nummsg_3}
\sum_{m_1,m_2,m_3} e^{\text{rand}}_0(m_1,m_2,m_3),\text{ where}\\
&e^{\text{rand}}_0(m_1,m_2,m_3)=
\Prob\Big( \Phi(Y^n,F_1,F_2,F_3) \neq (m_1,m_2,m_3) \Big|X_1^n=F_1(m_1),  X_2^n=F_2(m_2), X_3^n=F_3(m_3)\Big).
\end{align*}
 $P^{\text{strong}}_{e,1}$ is as below, $P^{\text{strong}}_{e,2},P^{\text{strong}}_{e,3}$ can be defined similarly.
\begin{align*}
&P^{\text{strong}}_{e,1}\defineqq \max_{x_1^n\in\mathcal{X}^n, f_1\in\mathcal{F}_1} \frac{1}{\nummsg_2\nummsg_3} \sum_{m_2,m_3} e^{\text{strong}}_{f_1}(x_1^n,m_2,m_3),\text{ where}\\
&e^{\text{strong}}_{f_1}(x_1^n,m_2,m_3)= \Prob\Big( (\Psi_2(Y^n,f_1,F_2,F_3),\Psi_3(Y^n,f_1,F_2,F_3)) \neq (m_2,m_3) \Big|\shortonly{\\
&\qquad\qquad\qquad\qquad }X_1^n=x_1^n, X_2^n=F_2(m_2), X_3^n=F_3(m_3) \Big).
\end{align*}
We say a rate triple $(R_1,R_2,R_3)$ is {\em achievable} against the strong adversary, if there is a sequence of $(\lfloor2^{nR_1}\rfloor,\lfloor2^{nR_2}\rfloor,\lfloor2^{nR_3}\rfloor,n)$ codes  $\{F_1^{(n)},F_2^{(n)},F_3^{(n)},\Phi^{(n)}\}_{n=1}^\infty$ such that $\lim_{n\rightarrow\infty}P^{\text{strong}}_{e}(F_1^{(n)},F_2^{(n)},F_3^{(n)},\Phi^{(n)})\rightarrow0$. The {\em strong randomized coding capacity region} $\mathcal{R^{\text{strong}}_{\text{random}}}$ is the closure of the set of all achievable rate triples. 

Probability of error and capacity region for randomized codes with weak adversary can be defined by replacing  $P^{\text{strong}}_{e,i}$ with $P^{\text{weak}}_{e,i},\, i = 1,2,3$ in the above definition,  where
\begin{align}
&P^{\text{weak}}_{e,1}\defineqq \max_{x_1^n\in \mathcal{X}^n} \frac{1}{\nummsg_2\nummsg_3} \sum_{m_2,m_3} e^{\text{weak}}_{1}(x_1^n,m_2,m_3),\label{eq:pe1_weak}\\
&\text{ where}\quad e^{\text{weak}}_{1}(x_1^n,m_2,m_3)= \Prob\Big( (\Psi_2(Y^n,F_1,F_2,F_3),\Psi_3(Y^n,F_1,F_2,F_3)) \neq (m_2,m_3) \Big|\nonumber\\
&\qquad\qquad\qquad\qquad X_1^n=x_1^n, X_2^n=F_2(m_2), X_3^n=F_3(m_3) \Big).\nonumber
\end{align}
$P^{\text{weak}}_{e,2}$ and $P^{\text{weak}}_{e,3}$ can be defined similarly.\\
For the standard adversary who has access to the random encoding map but is not allowed to choose it, the probability of error and capacity region can be defined in a similar fashion, by replacing $P^{\text{strong}}_{e,i}$ with $P^{\text{random}}_{e,i},\, i = 1,2,3$. We define $P^{\text{random}}_{e,1}$ below. $P^{\text{rand}}_{e,2}$ and $P^{\text{rand}}_{e,3}$ can be defined similarly. 
\begin{align*}
&P^{\text{rand}}_{e,1}\defineqq \max_{x_1^n:\mathcal{F}_1\rightarrow \mathcal{X}^n} \frac{1}{\nummsg_2\nummsg_3} \sum_{m_2,m_3} e^{\text{rand}}_{1}(x_1^n(\cdot),m_2,m_3),\text{ where}\\
&e^{\text{rand}}_{1}(x_1^n(\cdot),m_2,m_3)= \Prob\Big( (\Psi_2(Y^n,F_1,F_2,F_3),\Psi_3(Y^n,F_1,F_2,F_3)) \neq (m_2,m_3) \Big|\shortonly{\\
&\qquad\qquad\qquad\qquad }X_1^n=x_1^n(F_1), X_2^n=F_2(m_2), X_3^n=F_3(m_3) \Big).
\end{align*}

The notation $x_1^n(F_1)$ signifies that the adversary can choose the sequence $x_1^n$ with the knowledge of the encoding map $F_1$. The capacity region in this setting is denoted by $\mathcal{R}_{\random}$.
}

\subsection{Randomized coding capacity region}
Let $\mathcal{R}$ be the closure of the set of all rate triples $(R_1,R_2,R_3)$ such that for some $p(u)p(x_1|u)p(x_2|u)p(x_3|u)$ the following conditions hold for all permutations $(i,j,k)$ of $(1,2,3)$:
\begin{align}
R_i &\leq \min_{q(x_k|u)} I(X_i;Y|X_j,U),\quad\text{and} \label{eq:rateconstraint1}\\
R_i+R_j &\leq \min_{q(x_k|u)} I(X_i,X_j;Y|U), \label{eq:rateconstraint2}
\end{align}
where the mutual information terms above are evaluated using the joint distribution
$p(u) p(x_i|u)p(x_j|u)q(x_k|u) \allowbreak \mach(y|x_1,x_2,x_3)$. Here $|\mathcal{U}|\leq 3$.
\vspace{-0.25em}
\begin{thm}\label{thm:random}
\[ \mathcal{R}_{\random}^{\strong} = \mathcal{R}_{\random} = \mathcal{R}_{\random}^{\weak} = \mathcal{R}.\]
\end{thm} \vspace{-0.25em}\vspace{-0.25em}
We prove this by showing an achievability for the strong adversary and a converse for the weak adversary.

\subsection{Deterministic coding capacity region}

\begin{figure}
\centering
\begin{subfigure}[t]{.49\columnwidth}
\centering
\begin{tikzpicture}[scale=0.5]
	\draw (1.7,6) rectangle ++(1,1) node[pos=.5]{$q$};
	\draw (4.1,4) rectangle ++(1.3,1.5) node[pos=.5]{$W$};
	\draw[->] (1,4.75) node[anchor=east]{$x_2$} -- ++ (3.1,0);

	\draw[->] (1,6.8) node[anchor=east]{$\tilde{x}_2$} -- ++ (0.7,0) ;
	\draw[->] (1,6.2) node[anchor=east]{$\tilde{x}_3$} -- ++ (0.7,0) ;
	\draw[->] (3.4,4.25) -- ++(0.7,0);
	\draw[->] (3.4,5.25) -- ++(0.7,0);
	\draw[-] (3.4,5.25) -- ++ (0,1.25);
	\draw[-] (3.4,4.25) -- ++ (0,-1);
	\draw[-] (2.7, 6.5) -- ++(0.7,0);
	\draw[-] (1,3.25) node[anchor=east]{$x_3$} -- ++ (2.4,0);

	\draw[->] (5.4,4.75) --  ++ (0.5,0)node[anchor= west] {$y$};
\end{tikzpicture}
\end{subfigure}
\begin{subfigure}[t]{.49\columnwidth}
\centering
\begin{tikzpicture}[scale=0.5]
	\draw (1.7,6) rectangle ++(1,1) node[pos=.5]{$q$};
	\draw (4.1,4) rectangle ++(1.3,1.5) node[pos=.5]{$W$};
	\draw[->] (1,4.75) node[anchor=east]{$\tilde{x}_2$} -- ++ (3.1,0);

	\draw[->] (1,6.8) node[anchor=east]{$x_2$} -- ++ (0.7,0) ;
	\draw[->] (1,6.2) node[anchor=east]{$x_3$} -- ++ (0.7,0) ;
	\draw[->] (3.4,4.25) -- ++(0.7,0);
	\draw[->] (3.4,5.25) -- ++(0.7,0);
	\draw[-] (3.4,5.25) -- ++ (0,1.25);
	\draw[-] (3.4,4.25) -- ++ (0,-1);
	\draw[-] (2.7, 6.5) -- ++(0.7,0);
	\draw[-] (1,3.25) node[anchor=east]{$\tilde{x}_3$} -- ++ (2.4,0);

	\draw[->] (5.4,4.75) --  ++ (0.5,0)node[anchor= west] {$y$};
\end{tikzpicture}
\end{subfigure}
\caption{We say $\mach$ is $\cX_2\times\cX_3$-{symmetrizable by}~$\cX_1$ if, for each $(x_2,\tlx_2,x_3,\tlx_3)$, the conditional output distributions in the two cases above are the same. Thus, the receiver is unable to tell whether users~2 and~3 are sending $(x_2,x_3)$ or $(\tlx_2,\tlx_3)$.}
\label{fig:symm1}
\end{figure}
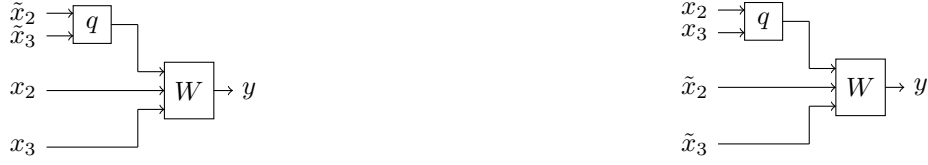

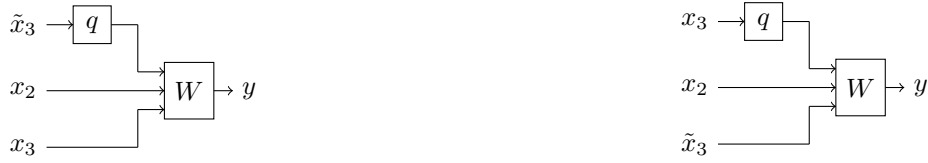
\begin{figure}
\centering
\begin{subfigure}[b]{.49\columnwidth}
\centering
\begin{tikzpicture}[scale=0.5]
	\draw (1.7,6) rectangle ++(1,1) node[pos=.5]{$q$};
	\draw (4.1,4) rectangle ++(1.3,1.5) node[pos=.5]{$W$};
	\draw[->] (1,4.75) node[anchor=east]{$x_2$} -- ++ (3.1,0);
	
	\draw[->] (1,6.5) node[anchor=east]{$\tilde{x}_3$} -- ++ (0.7,0) ;
	\draw[->] (3.4,4.25) -- ++(0.7,0);
	\draw[->] (3.4,5.25) -- ++(0.7,0);
	\draw[-] (3.4,5.25) -- ++ (0,1.25);
	\draw[-] (3.4,4.25) -- ++ (0,-1);
	\draw[-] (2.7, 6.5) -- ++(0.7,0);
	\draw[-] (1,3.25) node[anchor=east]{$x_3$} -- ++ (2.4,0);

	\draw[->] (5.4,4.75) --  ++ (0.5,0)node[anchor= west] {$y$};
\end{tikzpicture}
\end{subfigure}
\begin{subfigure}[b]{.49\columnwidth}
\centering
\begin{tikzpicture}[scale=0.5]
	\draw (1.7,6) rectangle ++(1,1) node[pos=.5]{$q$};
	\draw (4.1,4) rectangle ++(1.3,1.5) node[pos=.5]{$W$};
	\draw[->] (1,4.75) node[anchor=east]{$x_2$} -- ++ (3.1,0);

	\draw[->] (1,6.5) node[anchor=east]{$x_3$} -- ++ (0.7,0) ;
	\draw[->] (3.4,4.25) -- ++(0.7,0);
	\draw[->] (3.4,5.25) -- ++(0.7,0);
	\draw[-] (3.4,5.25) -- ++ (0,1.25);
	\draw[-] (3.4,4.25) -- ++ (0,-1);
	\draw[-] (2.7, 6.5) -- ++(0.7,0);
	\draw[-] (1,3.25) node[anchor=east]{$\tilde{x}_3$} -- ++ (2.4,0);

	\draw[->] (5.4,4.75) --  ++ (0.5,0)node[anchor= west] {$y$};
\end{tikzpicture}
\end{subfigure}
\caption{We say $\mach$ is $\cX_3|\cX_2$-{symmetrizable by}~$\cX_1$ if, for each $(x_2,x_3,\tlx_3)$, the conditional output distributions in the two cases above are the same. The receiver is unable to tell whether user~3 is sending $x_3$ or $\tlx_3$.}
\label{fig:symm2}
\end{figure}

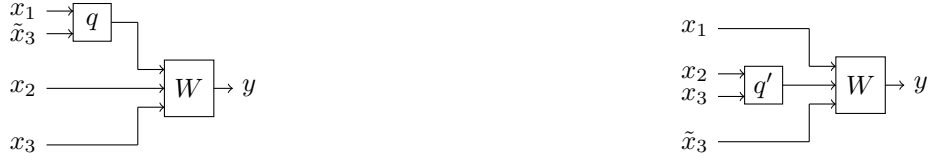
\begin{figure}[t]
\centering
\begin{subfigure}[t]{.49\columnwidth}
\centering
\begin{tikzpicture}[scale=0.5]
	\draw (1.7,6) rectangle ++(1,1) node[pos=.5]{$q$};
	\draw (4.1,4) rectangle ++(1.3,1.5) node[pos=.5]{$W$};
	\draw[->] (1,4.75) node[anchor=east]{$x_2$} -- ++ (3.1,0);
	
	\draw[->] (1,6.8) node[anchor=east]{$x_1$} -- ++ (0.7,0) ;
	\draw[->] (1,6.2) node[anchor=east]{$\tilde{x}_3$} -- ++ (0.7,0) ;
	\draw[->] (3.4,4.25) -- ++(0.7,0);
	\draw[->] (3.4,5.25) -- ++(0.7,0);
	\draw[-] (3.4,5.25) -- ++ (0,1.25);
	\draw[-] (3.4,4.25) -- ++ (0,-1);
	\draw[-] (2.7, 6.5) -- ++(0.7,0);
	\draw[-] (1,3.25) node[anchor=east]{$x_3$} -- ++ (2.4,0);

	\draw[->] (5.4,4.75) --  ++ (0.5,0)node[anchor= west] {$y$};
\end{tikzpicture}
\end{subfigure}
\begin{subfigure}[t]{.49\columnwidth}
\centering
\begin{tikzpicture}[scale=0.5]
	\draw (1.7,4.25) rectangle ++(1,1) node[pos=.5]{$q'$};
	\draw (4.1,4) rectangle ++(1.3,1.5) node[pos=.5]{$W$};
	\draw[->] (1,4.45) node[anchor=east]{$x_3$} -- ++ (0.7,0);
	\draw[->] (1,5.05) node[anchor=east]{$x_2$} -- ++ (0.7,0);
	\draw[->] (2.7,4.75) -- ++ (1.4,0);

	\draw[-] (1,6.25) node[anchor=east]{$x_1$} -- ++ (2.4,0) ;
	\draw[->] (3.4,4.25) -- ++(0.7,0);
	\draw[->] (3.4,5.25) -- ++(0.7,0);
	\draw[-] (3.4,5.25) -- ++ (0,1);
	\draw[-] (3.4,4.25) -- ++ (0,-1);
	\draw[-] (1,3.25) node[anchor=east]{$\tilde{x}_3$} -- ++ (2.4,0);

	\draw[->] (5.4,4.75) --  ++ (0.5,0)node[anchor= west] {$y$};
\end{tikzpicture}
\end{subfigure}
\caption{We say $\mach$ is $\cX_3$-{symmetrizable by}~$\cX_1/\cX_2$ if, for each $(x_1,x_2,x_3,\tlx_3)$, the conditional output distributions in the two cases above are the same. The receiver is unable to tell whether user~3 is sending $x_3$ (and  user~1 being malicious) or user~3 is sending $\tlx_3$ (and user~2 being malicious).}
\label{fig:symm3}
\end{figure}

\vspace{-0.25em}
\begin{thm}\label{thm:deterministic}
\[ \mathcal{R}_{\deterministic} = \mathcal{R},\qquad \text{if } \textup{int}(\mathcal{R}_{\deterministic})\neq\varnothing.\]
\end{thm} \vspace{-0.25em}\vspace{-0.25em}
The proof uses an extension of the elimination 
technique~\cite{Ahlswede78,Jahn81} to first show that
$n^2$-valued randomness at each encoder is sufficient to achieve
any rate-triple in $\mathcal{R}_{\random}^{\weak}$. A deterministic
code of small rate can be used to send $2\log_2 n$ bits out of each
message. These message bits are then used as the encoder randomness
in the next phase to communicate the rest of the message bits using
a randomized code.\\

\begin{defn}
Let $(i,j,k)$ be some permutation of $(1,2,3)$. We define three symmetrizability conditions for $\mach$.
\begin{enumerate}
\item We say that $\mach$ is $\cX_j\times\cX_k$-{\em symmetrizable by}~$\cX_i$ if for some distribution $q(x_i|x_j,x_k)$
\begin{align}
&\sum_{x_i} q(x_i|\tlx_j,\tlx_k)\machijk(y|x_i,x_j,x_k)\notag\\
&\qquad = \sum_{x_i} q(x_i|x_j,x_k)\machijk(y|x_i,\tlx_j,\tlx_k),\notag\\
&\qquad\qquad\quad \forall\; x_j,\tlx_j\in\mathcal{X}_j,\; x_k,\tlx_k\in\mathcal{X}_k,\; y\in\mathcal{Y}. \label{eq:symm1}
\end{align}
\item  We say that $\mach$ is $\cX_k|\cX_j$-{\em symmetrizable by}~$\cX_i$ if for some distribution $q(x_i|x_k)$
\begin{align}
&\sum_{x_i} q(x_i|\tlx_k)\machijk(y|x_i,x_j,x_k)\notag\\
&\qquad = \sum_{x_i} q(x_i|x_k)\machijk(y|x_i,x_j,\tlx_k),\notag\\
&\qquad\qquad\quad \forall\; x_j\in\mathcal{X}_j,\; x_k,\tlx_k\in\mathcal{X}_k,\; y\in\mathcal{Y}. \label{eq:symm2}
\end{align}
\item We say that $\mach$ is $\cX_k$-{\em symmetrizable by}~$\cX_i/\cX_j$ if for some pair of distributions $q(\tlx_i|x_i,x_k)$ and $q'(\tlx_j|x_j,x_k)$
\begin{align}
&\sum_{\tlx_i} q(\tlx_i|x_i,\tlx_k)\machijk(y|\tlx_i,x_j,x_k)\notag\\
&\qquad = \sum_{\tlx_j} q'(\tlx_j|x_j,x_k)\machijk(y|x_i,\tlx_j,\tlx_k),\notag\\
&\qquad\quad\;\forall\; x_i\in\mathcal{X}_i,\;x_j\in\mathcal{X}_j,\; x_k,\tlx_k\in\mathcal{X}_k,\; y\in\mathcal{Y}. \label{eq:symm3}
\end{align}
\end{enumerate}
\end{defn}
We will say that {\em user-$k$ is symmetrizable} if any of the above three symmetrizability conditions \eqref{eq:symm1}-\eqref{eq:symm3} holds for some distinct $i,j \in \{1,2,3\}\setminus\{k\}$. And, we will say that the channel $\mach$ is {\em not symmetrizable} if none of the users are symmetrizable, i.e., the channel is not  $\cX_j\times\cX_k$-{\em symmetrizable by}~$\cX_i$, $\cX_k|\cX_j$-{\em symmetrizable by}~$\cX_i$, and $\cX_k$-{\em symmetrizable by}~$\cX_i/\cX_j$ for all permutations $(i,j,k)$ of $(1,2,3)$.

\vspace{-0.25em}
\begin{thm}\label{thm:symmetrizability}
$\mathcal{R}_{\deterministic} = \mathcal{R}$ if $\mach$ is not symmetrizable. Furthermore, $\textup{int}(\mathcal{R}_{\deterministic})\neq\varnothing$ only if $\mach$ is not symmetrizable.
\end{thm} \vspace{-0.25em}\vspace{-0.25em}

\begin{remark}\label{rem:gap} We prove the converse part of Theorem~\ref{thm:symmetrizability} by showing that if user-$k$ is symmetrizable, then any achievable rate triple $(R_1,R_2,R_3)$ must be such that $R_k=0$. Our capacity region characterization does not cover the case where some (but not all) users are symmetrizable. In this case, by Theorem~\ref{thm:random}, $\mathcal{R}$ restricted to rates of non-symmetrizable users is clearly an outer bound on $\mathcal{R}_\deterministic$. It is tempting to conjecture that these regions are equal. A similar conjecture for the two-user AV-MAC was recently proved by Pereg and Steinberg~\cite{PeregS19}. \end{remark}

Clearly, symmetrizability conditions for the two-user AV-MAC with ${X}_i$ as the state and ${X}_j,{X}_k$ as the inputs are also symmetrizability conditions for our problem (Figures~\ref{fig:symm1} and~\ref{fig:symm2}). Thus, for the converse statement, the first two conditions \eqref{eq:symm1}-\eqref{eq:symm2} follow from Gubner~\cite{Gubner90}. The third condition \eqref{eq:symm3} is new and arises from the Byzantine nature of the users in this problem. Consider Figure~\ref{fig:symm3}. Here the receiver is unable to tell apart the two possibilities shown, i.e., whether user~1 is malicious with user~3 sending $x_3$ or user~2 is malicious with user~3 sending $\tlx_3$.

The following example shows that the third symmetrizability condition  does not imply  the others. The channel below is neither $\cX_j\times\cX_k$-{symmetrizable by}~$\cX_i$ nor $\cX_k|\cX_j$-{symmetrizable by}~$\cX_i$ for any permutation $(i,j,k)$ of $(1,2,3)$. However, it is $\cX_3$-symmetrizable by~$\cX_1/\cX_2$.
\begin{example}\label{ex:symmetrizable}
Let $\cX_1=\cX_2=\cY=\{0,1\}^3$ and $\cX_3=\{0,1\}$. We denote $x_1=(a_1,b_1,c_1)$, $x_2=(a_2,b_2,c_2)$, and $y=(y_1,y_2,y_3)$. Consider the channel $\mach$ defined by
\begin{align*}
(Y_1,Y_2)&=(C_1,C_2),\\
Y_3&=\left\{\begin{array}{ll} 
B_1\oplus(A_1\odot X_3) &\text{w.p. } 1/2\\
B_2\oplus(A_2 \odot X_3) &\text{w.p. } 1/2
\end{array}\right.
\end{align*}
where $\odot$ denotes multiplication and $\oplus$ denotes addition modulo~2.

To see that this channel is $\cX_3$-symmetrizable by~$\cX_1/\cX_2$, consider the ``deterministic'' $q((\tla_1,\tlb_1,\tlc_1)|(a_1,b_1,c_1),\tlx_3)$ and $q'((\tla_2,\tlb_2,\tlc_2)|(a_2,b_2,c_2),x_3)$ defined as follows: let $g,g':\{0,1\}^4\rightarrow\{0,1\}^2$ be defined as
\begin{align*}
g((a_1,b_1,c_1),\tlx_3) &= (0,b_1 \oplus (a_1\odot \tlx_3),c_1),\\
g'((a_2,b_2,c_2),x_3) &= (0,b_2 \oplus (a_2\odot x_3),c_2).
\end{align*}
Then 
\begin{align*}
q((\tla_1,\tlb_1,\tlc_1)|(a_1,b_1,c_1),\tlx_3)
&= 1_{(\tla_1,\tlb_1,\tlc_1)=g((a_1,b_1,c_1),\tlx_3)},\\
q'((\tla_2,\tlb_2,\tlc_2)|(a_2,b_2,c_2),x_3)
&= 1_{(\tla_2,\tlb_2,\tlc_2)=g'((a_2,b_2,c_2),x_3)}.
\end{align*}
Consider the two cases shown in Figure~\ref{fig:symm3} with $x_1=(a_1,b_1,c_1)$, $x_2=(a_2,b_2,c_2)$, and $q$ and $q'$ defined as above.
It follows that, in both the cases, the channel output $Y$ has the same conditional distribution given each input. In particular,
\begin{align*}
(Y_1,Y_2)&=(c_1,c_2),\\
Y_3&=\left\{\begin{array}{ll} 
b_1\oplus(a_1\odot \tlx_3) &\text{w.p. } 1/2\\
b_2\oplus(a_2 \odot x_3) &\text{w.p. } 1/2.
\end{array}\right.
\end{align*}
\noindent This shows that the symmetrizability condition~\eqref{eq:symm3} holds for $(i,j,k)=(1,2,3)$. 

Since $(Y_1,Y_2)=(C_1,C_2)$, it is clear that neither user-1 nor user-2 is symmetrizable. It only remains to show that the channel is neither $\cX_3|\cX_2$-{symmetrizable by}~$\cX_1$ nor $\cX_3|\cX_1$-{symmetrizable by}~$\cX_2$. Suppose the channel is $\cX_3|\cX_2$-{symmetrizable by}~$\cX_1$. Then, to satisfy \eqref{eq:symm2} for $x_2=(0,0,c_2)$ and $(x_3,\tlx_3)=(0,1)$, it is easy to verify that we must have
\begin{align*}
&q(0,0,0|1)+q(0,0,1|1)+q(1,0,0|1)+q(1,0,1|1)\\
&\qquad= q(1,1,0|0)+q(1,1,1|0)+q(0,0,0|0)+q(0,0,1|0).
\end{align*}
However, to satisfy \eqref{eq:symm2} for $x_2=(1,0,c_2)$ and $(x_3,\tlx_3)=(0,1)$, we can show that we must have
\begin{align*}
&1+q(0,0,0|1)+q(0,0,1|1)+q(1,0,0|1)+q(1,0,1|1)\\
&\qquad= q(1,1,0|0)+q(1,1,1|0)+q(0,0,0|0)+q(0,0,1|0).
\end{align*}
Hence, the channel is not $\cX_3|\cX_2$-{symmetrizable by}~$\cX_1$. By symmetry, it is not  $\cX_3|\cX_1$-{symmetrizable by}~$\cX_2$.

\end{example}

\section{Proof Sketches}
\subsection{Randomized coding capacity region}
\usetikzlibrary{arrows, shapes,positioning,
                chains,% <--- new
                decorations.markings,
                shadows, shapes.arrows}
  \definecolor{lightgreen}{rgb}{0.4,0.4,0.1}
\definecolor{lightblue}{rgb}{0.7372549019607844,0.8313725490196079,0.9019607843137255}
\definecolor{darkblue}{rgb}{0.08235294117647059,0.396078431372549,0.7529411764705882}
\definecolor{orangered}{rgb}{0.6,0.3,0.1}
\definecolor{gray}{rgb}{0.7529411764705882,0.7529411764705882,0.7529411764705882}

\begin{figure}\longonly{\centering} \scalebox{0.85}{\begin{tikzpicture}[line cap=round,line join=round,>=triangle 45,x=1.1cm,y=1.1cm]
\tikzstyle{myarrows}=[line width=0.5mm,draw = black!20!green!40!red,-triangle 45,postaction={draw, line width=2mm, shorten >=4mm, -}]

\fill[line width=0.4pt,color=lightgreen,fill=lightgreen,fill opacity=0.07] (-9.2,2.7) -- (-9.2,-.7) -- (-7.8,-.7) -- (-7.8,2.7) -- cycle;
\draw [rotate around={90:(-11.5,1)},line width=0.4pt,dotted,color=lightblue,fill=lightblue,fill opacity=0.46] (-11.5,1) ellipse (0.73cm and 0.367cm);
\draw [rotate around={90:(-10.25,1.5)},line width=0.4pt,dotted,color=lightblue,fill=lightblue,fill opacity=0.46] (-10.25,1.5) ellipse (0.73cm and 0.367cm);
\draw [rotate around={90:(-10.25,1)},line width=0.4pt,color=gray,fill=gray,fill opacity=0.2] (-10.25,1) ellipse (1.83cm and 0.55cm);
\draw [line width=0.4pt,color=lightgreen] (-9.2,2.7)-- (-9.2,-.7);
\draw [line width=0.4pt,color=lightgreen] (-9.2,-.7)-- (-7.8,-.7);
\draw [line width=0.4pt,color=lightgreen] (-7.8,-.7)-- (-7.8,2.7);
\draw [line width=0.4pt,color=lightgreen] (-7.8,2.7)-- (-9.2,2.7);

\draw [color = blue!60!red,->,line width=0.1pt] (-11.5,0.5) -- (-10.25,1.7);
\draw [color = blue!60!red,->,line width=0.1pt] (-11.5,1.5) -- (-10.25,1);
\draw [color = blue!60!red,->,line width=0.7pt] (-11.5,1.2) -- (-10.25,2);

\draw [color = black!20!green!40!red,->,line width=0.7pt] (-10.25,2) -- (-8.3,1.1);
\draw [color = black!20!green!40!red,->,line width=0.1pt] (-10.25,1) -- (-8.6,2);
\draw [color = black!20!green!40!red,->,line width=0.1pt] (-10.25,1.7) -- (-8.8,0.7);
\draw [color = black!20!green!40!red,->,line width=0.1pt] (-10.25,0.5) -- (-8.3,1.6);
\draw [color = black!20!green!40!red,->,line width=0.1pt] (-10.25,0.2) -- (-8.6,0.3);
\draw [color = black!20!green!40!red,->,line width=0.1pt] (-10.25,-0.5) -- (-8.3,-0.1);
\draw [color = black!20!green!40!red,->,line width=0.1pt] (-10.25,2.5) -- (-8.4,2.4);

\draw [rotate around={-90:(-8.5,1)},line width=0.4pt,color=gray,fill=gray,fill opacity=0.26] (-8.5,1) ellipse (1.83cm and 0.55cm);

\draw [fill=darkblue] (-11.5,1.2) circle (2.5pt);
\draw [fill=darkblue,opacity=0.8] (-11.5,1.5) circle (2.5pt);
\draw [fill=darkblue,opacity=0.8] (-11.5,0.5) circle (2.5pt);
\draw[color=darkblue] (-11.5,0) node {$\cM_i=\left[\frac{\epsilon 2^{nR_i}}{3}\right]$};
\draw[color=black] (-11.8,1.15) node {{\scriptsize $m_i$}};
\draw [fill=white] (-10.25,0.5) circle (2.5pt);
\draw [fill=white] (-10.25,0.2) circle (2.5pt);
\draw [fill=white] (-10.25,-0.5) circle (2.5pt);
\draw [fill=darkblue,opacity=0.8] (-10.25,1.7) circle (2.5pt);%
\draw [fill=darkblue] (-10.25,2) circle (2.5pt);%
\draw [fill=darkblue,opacity=0.8] (-10.25,1) circle (2.5pt);%
\draw [fill=white] (-10.25,2.5) circle (2.5pt);
\draw (-10.25,-1) node {$\tilde{\cM}_i=\left[2^{nR_i}\right]$};
\draw[color = lightgreen] (-8.5,-1) node {$\cX_i^n$};
\draw[color = blue!60!red] (-11,1.75) node {$L_i$};
\draw[color = blue!60!red] (-9.8,2.2) node {{\scriptsize ${\color{black}\tilde{m}_i=}{\color{blue!60!red}L_i}{\color{black}(m_i)}$}};
\draw[color = black!20!green!40!red] (-9.6,2.7) node {$G_i$};
\draw[color = black] (-7.8,0.95) node {{\scriptsize $x_i^n={\color{black!20!green!40!red}G_i}({\color{blue!60!red}L_i}{\color{black}(m_i)})$}};

\draw [fill=darkblue,opacity=0.8] (-8.6,2) circle (2.5pt);%
\draw [fill=darkblue] (-8.3,1.1) circle (2.5pt);%
\draw [fill=darkblue,opacity=0.8] (-8.8,0.7) circle (2.5pt);%
\draw [fill=white] (-8.3,1.6) circle (2.5pt);
\draw [fill=white] (-8.6,0.3) circle (2.5pt);
\draw [fill=white] (-8.3,-0.1) circle (2.5pt);
\draw [fill=white] (-8.4,2.4) circle (2.5pt);

%Y^n block
\draw [line width=0.4pt,color=orangered] (-7.2,2.7)-- (-7.2,-.7);
\draw [line width=0.4pt,color=orangered] (-7.2,-.7)-- (-5.8,-.7);
\draw [line width=0.4pt,color=orangered] (-5.8,-.7)-- (-5.8,2.7);
\draw [line width=0.4pt,color=orangered] (-5.8,2.7)-- (-7.2,2.7);
\fill[line width=0.4pt,color=orangered,fill=orangered,fill opacity=0.07] (-7.2,2.7) -- (-7.2,-.7) -- (-5.8,-.7) -- (-5.8,2.7) -- cycle;
\draw[color = orangered] (-6.5,-1) node {$\cY^n$};
\draw [fill=black] (-6.25,1.4) circle (2.5pt);
\draw[color = black] (-6.25,1.65) node {$y^n$};

% decoder step 1
\draw [rotate around={90:(-4.75,1)},line width=0.4pt,color=gray,fill=gray,fill opacity=0.2] (-4.75,1) ellipse (1.83cm and 0.55cm);
\draw [rotate around={90:(-4.75,1.5)},line width=0.4pt,dotted,color=lightblue,fill=lightblue,fill opacity=0.46] (-4.75,1.5) ellipse (0.73cm and 0.367cm);
\draw [fill=white] (-4.75,0.5) circle (2.5pt);
\draw [fill=white] (-4.75,0.2) circle (2.5pt);
\draw [fill=white] (-4.75,-0.5) circle (2.5pt);
\draw [fill=darkblue] (-4.75,1.7) circle (2.5pt);%
\draw [fill=darkblue] (-4.75,2) circle (2.5pt);%
\draw [fill=darkblue] (-4.75,1) circle (2.5pt);%
\draw [fill=white] (-4.75,2.5) circle (2.5pt);
\draw (-4.75,-1) node {$\tilde{\cM}_i$};
\draw[color = black!20!green!40!red] (-5.45,1.5) node {$\Gamma^{(k)}_i$};
\draw[color = blue!60!red] (-4,2) node {$\Lambda_i$};

\draw [->,color=darkblue] (-5.4,3) to [out=250,in=150] (-4.9,2);
\draw[color = black] (-5.3,3.2) node {valid inner messages for $L_i$};

\draw [myarrows](-5.9,1)--(-5.1,1);

%decoder step 2
\draw [rotate around={90:(-3.5,1.5)},line width=0.4pt,dotted,color=lightblue,fill=lightblue,fill opacity=0.46] (-3.5,1.5) ellipse (0.73cm and 0.367cm);
\draw [fill=darkblue] (-3.5,1.7) circle (2.5pt);
\draw [fill=darkblue] (-3.5,2) circle (2.5pt);
\draw [fill=darkblue] (-3.5,1) circle (2.5pt);
\draw [color=darkblue] (-3.3,.6) node {$\cM_i$};
\draw [color=red!50!blue,fill=red!70!blue!20!white] (-3.5,0) circle (2.5pt);
\draw [color=red!50!blue] (-3.3,-0.1) node {$\bot$};

\draw [color = blue!60!red,->,line width=0.2pt] (-4.75,1.7) -- (-3.5,1);
\draw [color = blue!60!red,->,line width=0.2pt] (-4.75,1) -- (-3.5,2);
\draw [color = blue!60!red,->,line width=0.2pt] (-4.75,2) -- (-3.5,1.7);

\draw [color = blue!60!red,->,line width=0.2pt] (-4.75,0.5) -- (-3.55,0);
\draw [color = blue!60!red,->,line width=0.2pt] (-4.75,0.2) -- (-3.55,0);
\draw [color = blue!60!red,->,line width=0.2pt] (-4.75,-0.5) -- (-3.55,0);
\draw [color = blue!60!red,->,line width=0.2pt] (-4.75,2.5) -- (-3.55,0);

\draw (-9.75,-1.7) node {\small {(a)} Encoder $F_i:L_i\circ G_i$};
\draw (-5.25,-1.7) node {\small {(b)} Decoder $\Phi^{(k)}_i:\Gamma_i^{(k)}\circ\Lambda_i$};

\end{tikzpicture}}
\caption{The encoders and decoders for Theorem~\ref{thm:random}.}\label{fig:random_proof}

\end{figure}
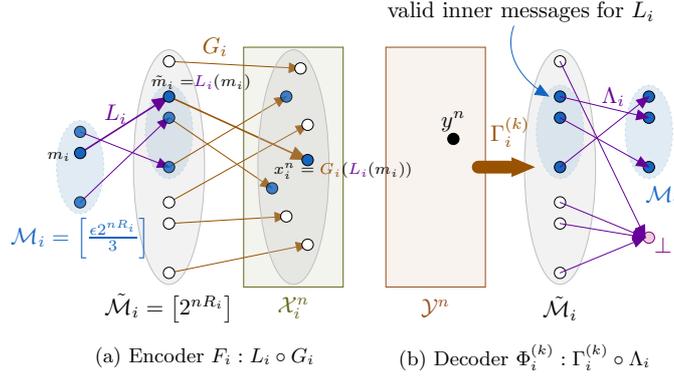

Below we sketch the proof of achievability for the strong adversary setting. \longonly{Full details and a converse proof for the weak adversary case are available in Appendix~\ref{app:A}}. \shortonly{Full details and a converse proof for the weak adversary case are available in the extended draft~\cite{ExtendedDraft}.} 
\vspace{-0.25em}
\begin{proof}[Proof sketch (achievability of Theorem~\ref{thm:random})] Our achievability relies on the following encoder and decoder constructions (Figure~\ref{fig:random_proof}).\\
\noindent\underline{\em(a) Encoder $F_i=L_i\circ G_i$:} First, for each user $i$, a random injective map $L_i$ embeds the given message $m_i\in\cM_i$ into the set $\tilde{\cM}_i$ whose size is $3/\epsilon$ times that of $\cM_i$ to obtain the inner message $\tilde{m}_i$. Next, the codeword $x_i^n\in\cX_i^n$ is generated using a random code $G_i$ drawn for the message set $\tilde{\cM}_i$ using an input distribution $p_i$. The encoder maps $L_i$ and $G_i$ are made available to the decoder as the shared secret between user-$i$ and the decoder.\\ \noindent\underline{\em (b) Decoders $\Phi_i^{(k)}=\Gamma_i^{(k)}\circ \Lambda_i$, $k\neq i$:} For each $k\neq i$, we first use a decoding rule $\Gamma_i^{(k)}$ for a $2$-user AV-MAC that is formed by treating user $k$'s channel input as the adversarial state. If the reconstructed inner message $\hat{\tilde{m}}_i^{(k)}=\Gamma_i^{(k)}(y^n)$ is valid (under the embedding $L_i$), the map $\Lambda_i$ outputs message reconstruction $\hat{m}_i^{(k)}=L_i^{-1}(\hat{\tilde{m}}_i^{(k)})$. Else, it outputs a failure, denoted by $\bot$. 

The achievability follows by showing that, as long as the rate triple $(R_1,R_2,R_3)$ satisfy the constraints~\eqref{eq:rateconstraint1} and~\eqref{eq:rateconstraint2} with $U$ being a constant value -- thus, each pair of rates lie in the corresponding AV-MAC capacity region -- the following hold simultaneously: \emph{(i)} when there is no adversary or when node $k$ is the only adversarial node, $\hat{m}_i$ equals $m_i$ w.h.p., and \emph{(ii)} when a node other that node $k$ is the adversarial node, w.h.p., either $\hat{m}_i$ equals $m_i$ or the decoder $\Phi_i^{(k)}$ outputs $\bot$. Thus, with high probability, for each non adversarial user $i$, \emph{(i)} at least one of the $\Phi_i^{(k)}$'s outputs the correct message, and \emph{(ii)} none of the $\Phi_i^{(k)}$'s output an incorrect message (though they may output $\bot$). Finally, the achievability of any rate triple satisfying~\eqref{eq:rateconstraint1} and~\eqref{eq:rateconstraint2} (\emph{i.e.}, with an arbitrary auxiliary random variable $U$ in the bounds) follows from a time sharing argument. 
\end{proof} \vspace{-0.25em}\vspace{-0.25em}

\subsection{Deterministic coding capacity region (Theorem~\ref{thm:symmetrizability})}
\begin{proof}[Proof (Converse of Theorem~\ref{thm:symmetrizability})]

\shortonly{
Clearly, symmetrizability conditions for the two-user AV-MAC with $\mathcal{X}_i$ as the state alphabet and $\mathcal{X}_j,\mathcal{X}_k$ as the input alphabets are also symmetrizability conditions for our problem. Conditions~1 and 2 follow from Gubner~\cite{Gubner90}.

To show condition~3, consider $(i,j,k)=(1,2,3)$, the other cases follow similarly. Suppose $q(\tlx_1|x_1,x_3)$ and $q'(\tlx_j|x_j,x_k)$ satisfy \eqref{eq:symm3}. i.e.,
\begin{align}
&\sum_{\tlx_1} q(\tlx_1|x_1,\tlx_3)\mach(y|\tlx_1,x_2,x_3)\notag\\
&\qquad = \sum_{\tlx_2} q'(\tlx_2|x_2,x_3)\mach(y|x_1,\tlx_2,\tlx_3),\notag\\
&\qquad\quad\;\forall\; x_1\in\mathcal{X}_1,\;x_2\in\mathcal{X}_2,\; x_3,\tlx_3\in\mathcal{X}_3,\; y\in\mathcal{Y}. \label{eq:conversesymm}
\end{align}
Let $m_3,\tilde{m}_3\in\mathcal{M}_3$ be distinct, and let $x_3^n=f_3(m_3)$ and $\tlx_3^n=f_3(\tilde{m}_3)$. We consider two different settings in which user-3 sends $x_3^n$ and $\tlx_3^n$ respectively:
\begin{enumerate}
\item[(i)] In the first setting, user-1 is adversarial. It chooses an $M_1\sim\textup{Unif}(\mathcal{M}_1)$. Let $X_1^n=f_1(M_1)$. To produce its input $\tilde{X}_1^n$ to the channel, it passes $(X_1^n,\tlx_3^n)$ through $q^n$, the $n$-fold product of the channel $q(\tlx_1|x_1,x_3)$. User-2, being non-adversarial, sends as its input to the channel $X_2^n=f_2(M_2)$, where $M_2\sim\textup{Unif}(\mathcal{M}_2)$. User-3 sends $x_3^n$ corresponding to message ${m}_3$. The distribution of the received vector in this case is
\begin{align*}
\frac{1}{\nummsg_1\nummsg_2}\sum_{m_1,m_2}
\prod_{t=1}^n &\sum_{\tlx_{1,t}}q(\tlx_{1,t}|f_{1,t}(m_1),\tlx_{3,t})\\&  \mach(y_t|\tlx_{1,t},f_{2,t}(m_2),x_{3,t}).
\end{align*}
\item[(ii)] In the second setting, user-2 is adversarial. It chooses an $M_2\sim\textup{Unif}(\mathcal{M}_2)$. Let $X_2^n=f_2(M_2)$. To produce its input $\tilde{X}_2^n$ to the channel, it passes $(X_2^n,x_3^n)$ through $q'^n$, the $n$-fold product of the channel $q'(\tlx_2|x_2,x_3)$. User-1, being non-adversarial now, sends as its input to the channel $X_1^n=f_1(M_1)$, where $M_1\sim\textup{Unif}(\mathcal{M}_1)$. User-3 here sends $\tlx_3^n$ corresponding to message $\tilde{m}_3$. Here, the distribution of the received vector is
\begin{align*}
\frac{1}{\nummsg_1\nummsg_2}\sum_{m_1,m_2}
\prod_{t=1}^n &\sum_{\tlx_{2,t}}q'(\tlx_{2,t}|f_{2,t}(m_2),x_{3,t})\\&  \mach(y_t|f_{1,t}(m_1),\tlx_{2,t},\tlx_{3,t}).
\end{align*}
\end{enumerate}
By \eqref{eq:conversesymm}, the above two distributions are identical. Hence, for any decoder, the sum of probabilities of decoding error for messages $m_3$ and $\tilde{m}_3$ must be at least 1, i.e.,
if we define $e^{3}_1(m_3,\tlx_1^n)\defineqq \frac{1}{\nummsg_2}\sum_{m_2'} e_1(\tlx_1^n,m_2',m_3)$ and similarly $e^{3}_2(\tilde{m}_3,\tlx_2^n)\defineqq\frac{1}{\nummsg_1}\sum_{m_1'} e_2(m_1',\tlx_2^n,\tilde{m}_3)$, then
\[ E_{\tilde{X}_1^n}[e^{3}_1(m_3,\tilde{X}_1^n)] + E_{\tilde{X}_2^n}[e^{3}_2(\tilde{m}_3,\tilde{X}_2^n)] \geq 1.\]
	Note that the distribution of $\tilde{X}_1^n$ (resp. $\tilde{X}_2^n$) does not depend on $m_3$ (resp. $\tilde{m}_3$). Arguing along the lines of~\cite[(3.29) in page~187]{CsiszarN88}, we can show that the average probability of error is at least 1/8. \shortonly{See the extended version~\cite{ExtendedDraft} for more details.}
\end{proof}

}

\longonly{ Clearly, symmetrizability conditions for the two-user AV-MAC with $\mathcal{X}_i$ as the state alphabet and $\mathcal{X}_j,\mathcal{X}_k$ as the input alphabets are also symmetrizability conditions for our problem. Conditions~1 and 2 follow from Gubner~\cite{Gubner90}.

To show condition~3, consider $(i,j,k)=(1,2,3)$, the other cases follow similarly. Suppose $q(\tlx_1|x_1,x_3)$ and $q'(\tlx_j|x_j,x_k)$ satisfy \eqref{eq:symm3}, i.e.,
\begin{align}
&\sum_{\tlx_1} q(\tlx_1|x_1,\tlx_3)\mach(y|\tlx_1,x_2,x_3)\notag\\
&\qquad = \sum_{\tlx_2} q'(\tlx_2|x_2,x_3)\mach(y|x_1,\tlx_2,\tlx_3),\notag\\
&\qquad\quad\;\forall\; x_1\in\mathcal{X}_1,\;x_2\in\mathcal{X}_2,\; x_3,\tlx_3\in\mathcal{X}_3,\; y\in\mathcal{Y}. \label{eq:conversesymm}
\end{align}
Let $m_3,\tilde{m}_3\in\mathcal{M}_3$ be distinct, and let $x_3^n=f_3(m_3)$ and $\tlx_3^n=f_3(\tilde{m}_3)$. We consider two different settings in which user-3 sends $x_3^n$ and $\tlx_3^n$ respectively:
\begin{enumerate}
\item[(i)] In the first setting, user-1 is adversarial. It chooses an $M_1\sim\textup{Unif}(\mathcal{M}_1)$. Let $X_1^n=f_1(M_1)$. To produce its input $\tilde{X}_{1,\tilde{m}_3}^n$ to the channel, it passes $(X_1^n,\tlx_3^n)$ through $q^n$, the $n$-fold product of the channel $q(\tlx_1|x_1,x_3)$. User-2, being non-adversarial, sends as its input to the channel $X_2^n=f_2(M_2)$, where $M_2\sim\textup{Unif}(\mathcal{M}_2)$. User-3 sends $x_3^n$ corresponding to message ${m}_3$. The distribution of the received vector in this case is
\begin{align*}
\frac{1}{\nummsg_1\nummsg_2}\sum_{m_1,m_2}
\prod_{t=1}^n &\sum_{\tlx_{1,\tilde{m}_3,t}}q(\tlx_{1,\tilde{m}_3,t}|f_{1,t}(m_1),\tlx_{3,t}) \mach(y_t|\tlx_{1,\tilde{m}_3,t},f_{2,t}(m_2),x_{3,t}).
\end{align*}
\item[(ii)] In the second setting, user-2 is adversarial. It chooses an $M_2\sim\textup{Unif}(\mathcal{M}_2)$. Let $X_2^n=f_2(M_2)$. To produce its input $\tilde{X}_{2,m_3}^n$ to the channel, it passes $(X_2^n,x_3^n)$ through $q'^n$, the $n$-fold product of the channel $q'(\tlx_2|x_2,x_3)$. User-1, being non-adversarial now, sends as its input to the channel $X_1^n=f_1(M_1)$, where $M_1\sim\textup{Unif}(\mathcal{M}_1)$. User-3 here sends $\tlx_3^n$ corresponding to message $\tilde{m}_3$. Here, the distribution of the received vector is
\begin{align*}
\frac{1}{\nummsg_1\nummsg_2}\sum_{m_1,m_2}
\prod_{t=1}^n &\sum_{\tlx_{2,m_3,t}}q'(\tlx_{2,m_3,t}|f_{2,t}(m_2),x_{3,t})\mach(y_t|f_{1,t}(m_1),\tlx_{2,m_3,t},\tlx_{3,t}).
\end{align*}
\end{enumerate}
By \eqref{eq:conversesymm}, the above two distributions are identical. Hence, for any decoder, the sum of probabilities of decoding error for messages $m_3$ and $\tilde{m}_3$ must be at least 1, i.e.,
if we define $e^{3}_1(m_3,\tlx_1^n)\defineqq \frac{1}{\nummsg_2}\sum_{m_2'} e_1(\tlx_1^n,m_2',m_3)$ and similarly $e^{3}_2(\tilde{m}_3,\tlx_2^n)\defineqq\frac{1}{\nummsg_1}\sum_{m_1'} e_2(m_1',\tlx_2^n,\tilde{m}_3)$, then
\shortonly{
\[ E_{\tilde{X}_1^n}[e^{3}_1(m_3,\tilde{X}_1^n)] + E_{\tilde{X}_2^n}[e^{3}_2(\tilde{m}_3,\tilde{X}_2^n)] \geq 1.\]}
\longonly{
\begin{align*}
E_{\tilde{X}_{1,\tilde{m}_3}^n}[e^{3}_1(m_3,\tilde{X}_{1,\tilde{m}_3}^n)]& + E_{\tilde{X}_{2,m_3}^n}[e^{3}_2(\tilde{m}_3,\tilde{X}_{2,m_3}^n)] = \\
&\quad \sum_{y^n:\phi(y^n)\neq m_3}\left(\frac{1}{\nummsg_1\nummsg_2}\sum_{m_1,m_2}\prod_{t=1}^n \sum_{\tlx_{1,t}}q(\tlx_{1,t}|f_{1,t}(m_1),\tlx_{3,t}) \mach(y_t|\tlx_{1,t},f_{2,t}(m_2),x_{3,t})\right)\\
&\quad+\sum_{y^n:\phi(y^n)\neq \tilde{m}_3}\left(\frac{1}{\nummsg_1\nummsg_2}\sum_{m_1,m_2}
\prod_{t=1}^n \sum_{\tlx_{2,t}}q'(\tlx_{2,t}|f_{2,t}(m_2),x_{3,t})\mach(y_t|f_{1,t}(m_1),\tlx_{2,t},\tlx_{3,t})\right)\\
&\quad\stackrel{\text{(a)}}{\geq} 1,
\end{align*}
where (a) follows from \eqref{eq:conversesymm}.

}
	Note that the distribution of $\tilde{X}_1^n$ (resp. $\tilde{X}_2^n$) does not depend on $m_3$ (resp. $\tilde{m}_3$). Arguing along the lines of~\cite[(3.29) in page~187]{CsiszarN88},\shortonly{ we can show that the average probability of error is at least 1/8. See the extended version~\cite{ExtendedDraft} for more details.}

\longonly{
\begin{align*}
2P_{e}(f_1,f_2,f_3,\phi)&\geq P_{e,1}+P_{e,2}\\
&\geq \frac{1}{N_3}\sum_{m_3}E_{\tilde{X}_{1}^n}[e^{3}_1(m_3,\tilde{X}_{1}^n)] + \frac{1}{N_3}\sum_{m_3}E_{\tilde{X}_{2}^n}[e^{3}_1(m_3,\tilde{X}_{2}^n)]
\end{align*}
for any attack vectors $\tilde{X}_1^n$ and $\tilde{X}_2^n$. In particular, for the attack vectors $\frac{1}{N_3}\sum_{\tilde{m}_3}\tilde{X}_{1,\tilde{m_3}}$ and $\frac{1}{N_3}\sum_{m_3}\tilde{X}_{2,m_3}$,
\begin{align*}
2P_{e}(f_1,f_2,f_3,\phi)\geq \frac{1}{N^2_3}\sum_{\tilde{m}_3}\sum_{m_3}\left(E_{\tilde{X}_{1,\tilde{m}_3}^n}[e^{3}_1(m_3,\tilde{X}_{1,\tilde{m}_3}^n)] + E_{\tilde{X}_{2,m_3}^n}[e^{3}_2(\tilde{m}_3,\tilde{X}_{2,m_3}^n)]\right).
\end{align*}
For $m_3\neq \tilde{m}_3$, the term in brackets on the right is upper bounded by 1, otherwise it is upper bounded by zero. Thus, 
\begin{align*}
P_{e}(f_1,f_2,f_3,\phi)&\geq \frac{N_3(N_3-1)/2}{2N_3^2}\\
&\geq \frac{1}{8}.
\end{align*}}}

\end{proof}
Our achievability proof uses ideas from~\cite{CsiszarN88} and is along the lines of the achievability proofs in~\cite{AhlswedeC99,PeregS19}. \longonly{It is presented in Appendix~\ref{app:3}.} \shortonly{In the interest of space, we only present our decoder  here.}
For a random variable $X$, let $\mathcal{P}_{X}$ denote the set of empirical distributions for blocklength $n$. For a distribution $P_X \in \mathcal{P}_X$, let $T^n_X$ be the set of typical sequences with relative frequencies specified by $P_X$.

\begin{defn}[Decoder]\label{def:decoder}
For $\eta>0, \,(m_1,m_2,m_3)\in \mathcal{M}_1\times\mathcal{M}_2\times\mathcal{M}_3$, and encoding maps, $f_1:\mathcal{M}_1\rightarrow \mathcal{X}_1^n, \, f_2:\mathcal{M}_2\rightarrow \mathcal{X}_2^n$ and $f_3:\mathcal{M}_3\rightarrow \mathcal{X}_3^n$, the decoding set $D_{m_1,m_2,m_3}  \subseteq \mathcal{Y}^n$ is defined as the intersection of three decoding sets $D_{m_1},D_{m_2},D_{m_3}\subseteq \mathcal{Y}^n$ defined below.\\
 For $m_1\in \mathcal{M}_1$, a sequence $\mathbf{y} \in D_{m_1}$ if there exists  some permutation $(j, k)$ of $(2,3)$, $m_j \in \mathcal{M}_j,\,\mathbf{x}_k\in \mathcal{X}^n_k$, and random variables $X_1, X_j, X_k$ with $(f_1(m_1), f_j(m_j), \mathbf{x}_{k}, \mathbf{y}) \in T^{n}_{X_1,X_j,X_k,Y}$ and $D(P_{X_1,X_j,X_k, Y}||P_{X_1}\times P_{X_j}\times P_{X_k}\times W)< \eta$ such that the following hold:
	\begin{description}			
			 \item[(a)] If there exists $m'_1 \in \mathcal{M}_1$, $m'_1 \neq m_1$, $m'_j \in \mathcal{M}_j$, $m'_j\neq m_j$, $\mathbf{x}'_k\in \mathcal{X}^n_k$, and random variables $X_1',X_j', X_k'$ such that $(f_1(m_1), f_1(m'_1),f_j(m_j), f_j(m'_j), \mathbf{x}_k, \mathbf{x}'_k, \mathbf{y}) \in T^{n}_{X_1 ,X_1', X_j,X'_j, X_k ,X_k',Y}$ and  $D(P_{X_1',X'_j,X'_k, Y}||P_{X'_1}\times P_{X'_j}\times P_{X'_k}\times W)< \eta$, then $I(X_1,X_j,Y;X'_1,X'_j|X_k) < \eta$.

			\item[(b)] If there exists $m'_1 \in \mathcal{M}_1$, $m'_1 \neq m_1$, $\mathbf{x}'_k\in \mathcal{X}^n_k$, and random variables $X_1',X_k'$ such that $(f_1(m_1),f_1(m'_1),f_j(m_j), \allowbreak \mathbf{x}_k, \mathbf{x}'_k, \mathbf{y}) \in T^{n}_{X_1, X_1',X_j, X_k, X_k', Y}$ and  $D(P_{X_1',X_j,X_k', Y}||P_{X_1'}\times P_{X_j}\times P_{X_k'}\times W)< \eta$, then $I(X_1,X_j,Y;X_1'|X_k) < \eta$.

			\item[(c)] If there exists $m'_1 \in \mathcal{M}_1$, $m'_1 \neq m_1$, $\mathbf{x}_j\in \mathcal{X}^n_j$, $m_k\in \mathcal{M}_k$, and random variables $X_1',X_j', X_k'$ such that $(f_1(m_1),f_1(m'_1),f_j(m_j), \mathbf{x}_j, \mathbf{x}_k, f_k(m_k), \mathbf{y}) \in T^{n}_{X_1, X_1', X_j, X_j', X_k, X_k',Y}$ and  $D(P_{X_1',X_j',X_k', Y}||P_{X_1'}\times P_{X_j}\times P_{X_k'}\times W)< \eta$, then $I(X_1,X_j,Y;X_1',X_k'|X_k) < \eta$.

	\end{description}
The decoding sets $D_{m_2}$ and $D_{m_3}$ are defined similarly. 
\end{defn}
While we decode each user's message separately, as in joint decoding, the structure of the other users' codebooks is made use of. Conditions {\bf(a)} and {\bf(b)} above are similar to the decoding conditions in~\cite{AhlswedeC99}. Condition~{\bf(c)} is associated with our new non-symmetrizability criterion. It handles the situation in which an adversarial user tries to make another user appear adversarial while pretending to act honestly.  

The following lemma guarantees that the decoder above is well-defined. This is analogous to~\cite[Lemma~4]{CsiszarN88}.
\begin{lemma}[Disambiguity of decoding]\label{lemma:dec}
Suppose the channel \mach is not symmetrizable. Let $P_{X_1}\in \mathcal{P}_{X_1}, P_{X_2} \in \mathcal{P}_{X_2}$ and $P_{X_3}\in \mathcal{P}_{X_3}$ be distributions such that for some $\alpha >0,$ $\min_{x_1}P_{X_1}(x_1),\min_{x_2}P_{X_2}(x_2),\allowbreak \min_{x_3}P_{X_3}(x_3)\geq \alpha$. Let $f_1:\mathcal{M}_1\rightarrow T^{n}_{X_1}, \, f_2:\mathcal{M}_2\rightarrow T^{n}_{X_2}$ and $f_3:\mathcal{M}_3\rightarrow T^{n}_{X_3}$ be any encoding maps. There exists a choice of $\eta>0$  such that if $(\tilde{m}_1,\tilde{m}_2,\tilde{m}_3)\neq(m_1,m_2,m_3),\, D_{\tilde{m}_1,\tilde{m}_2,\tilde{m}_3}\cap D_{m_1,m_2,m_3} = \emptyset$.
\end{lemma}

\shortonly{In the extended draft~\cite{ExtendedDraft} we show the existence of codes which when used with the above decoder allow us to achieve all rates in $\textup{int}(\mathcal{R})$ whenever the MAC is non-symmetrizable.}
\longonly{In Appendix~\ref{app:3}, we show the existence of codes which when used with the above decoder allow us to achieve all rates in $\textup{int}(\mathcal{R})$ whenever the MAC is non-symmetrizable.}

\longonly{
\begin{appendices}

\section{Proof of Theorem~\ref{thm:random}}\label{app:A}
\begin{proof}[Proof (Achievability of Theorem 1)] For each $k=1,2,3$, let $W^{(k)}$ be the $2$-user AVMAC formed by channel inputs from node $k$ as the state and the remaining channel inputs as legitimate inputs.  Let $(R_1,R_2,R_3)$ be a rate triple such that, for some $p_1(x_1)p_2(x_2)p_3(x_3)$, the following conditions hold for all permutations $(i,j,k)$ for $(1,2,3)$: \begin{align}
R_i &\leq \min_{q(x_k)} I(X_i;Y|X_j),\quad\text{and}\label{eq:rateconstraint1c}\\
R_i+R_j &\leq \min_{q(x_k)} I(X_i,X_j;Y),\label{eq:rateconstraint2c}
\end{align} 
with the mutual information terms evaluated using the joint distribution
$p_i(x_i)p_j(x_j)q(x_k) \mach(y|x_1,x_2,x_3)$. Let $\epsilon>0$ be arbitrary and let $n$ be large enough. Note that, by~\cite{Jahn81}, the rate pair $(R_i,R_j)$ is an  achievable rate pair for the AVMAC $W^{(k)}$. For each $i\in[3]$, let $\tilde{\cM}_i=[2^{nR_i}]$ and $\cM_i=[\epsilon 2^{nR_i}/3]$. In the following, we show the existence of a randomized $(\epsilon 2^{nR_1}/3,\epsilon 2^{nR_2}/3,\epsilon 2^{nR_3}/3,n)$ code  $(F_1,F_2,F_3,\Phi)$ with $P_e$ no larger than $\epsilon$.

\paragraph{Code design} Before describing our code, we describe the following maps. First, for each user $i$, let $L_i:\cM_i\to\tilde{\cM}_i$ and $\Lambda_i:\tilde{\cM}_i\to\cM_i$ be the forward and reverse maps for a uniformly chosen injection from $\cM_i$ to $\tilde{\cM}_i$. Next, let $G_i:\tilde{\cM}_i\to\cX_i^n$ be the encoder map for a randomly chosen code whose each letter is drawn independently from the probability distribution $p_{i}$. Further, for each $k$ and $i\neq k$, let $\Gamma^{(k)}_i:\cY^n\to\tilde{\cM}_i$ be the reconstruction map for the message from user $i$ (as described in~\cite{Jahn81}) when the (randomly drawn) encoder maps $(G_i:i\in\set{1,2,3}\setminus\set{k})$ are used to generate the codewords sent over the AVMAC $W^{(k)}$. Let $n$ be large enough such that the code $((G_i:i\in\set{1,2,3}\setminus\set{k}),(\Gamma^{(k)}_i:i\in\set{1,2,3}\setminus\set{k}))$ has error probability no larger than $\epsilon/3$.
 
For each $i\in\set{1,2,3}$, the encoder map $F_i:\cM_i\to\cX^n_i$ is defined as $F_i(m_i)= G_i(L_i(m_i))$ for every $m_i\in\cM_i$. For each $k\in\set{1,2,3}\setminus\set{i}$, let 
$$\Phi_i^{(k)}(y^n)=\begin{cases}\Lambda_i(\Gamma^{(k)}_i(y^n)) &\mbox{if }\Gamma^{(k)}_i(y^n)\in L_i(\cM_i),\\ \bot &\mbox{otherwise}.\end{cases}$$ 
The decoder $\Phi:\cY^n\to\cM_1\times\cM_2\times\cM_3$ outputs $\Phi(y^n)=(\hat{m}_1,\hat{m}_2,\hat{m}_3)$, where, for each $i\in\set{1,2,3}$ and $(j,k)$ a permutation of $\set{1,2,3}\setminus\set{i}$, 
$$\hat{m}_i=\begin{cases}\Phi_i^{(j)}(y^n)& \mbox{ if }  \Phi_i^{(j)}(y^n)=\Phi_i^{(k)}(y^n)\neq\bot \\
\Phi_i^{(j)}(y^n)& \mbox{ if } \Phi_i^{(j)}(y^n)\neq\bot\mbox{ and } \Phi_i^{(k)}(y^n)=\bot\\
\Phi_i^{(k)}(y^n)& \mbox{ if } \Phi_i^{(k)}(y^n)\neq\bot\mbox{ and } \Phi_i^{(j)}(y^n)=\bot\\
1&\mbox{otherwise.}
\end{cases}$$
\paragraph{Error Analysis}
We first show that as long as the rate triple $(R_1,R_2,R_3)$ satisfy the constraints~\eqref{eq:rateconstraint1c} and~\eqref{eq:rateconstraint2c} with $U$ being a constant value -- thus, each pair of rates lie in the corresponding AV-MAC capacity region -- the following hold simultaneously for every user $i$ and potentially adversarial user $k\neq i$: \emph{(i)} $\Phi_i^{(k)}$ equals $m_i$ w.h.p. if user $k$ is indeed adversarial and \emph{(ii)} $\Phi_i^{(k)}$ either equals $\bot$ or $m_i$ if user $k$ is not adversarial. To this end, consider any permutation $(i,j,k)$ of $(1,2,3)$ and assume that the  adversarial user (if any) is user $k$. Note that the rate pair $(R_i,R_j)$ is an achievable rate for the AV-MAC $W^{(k)}$. 
\begin{description}
\item[$(i)$]First, consider  $\Phi_i^{(k)}$. Applying the achievability proof for AV-MACs~\cite{Jahn81} to the AV-MAC $W^{(k)}$, with the random code with encoder mappings $G_i$ and $G_j$ and the decoder mapping $(\Gamma_i^{(k)},\Gamma_j^{(k)})$,  $\Gamma_i^{(k)}(Y^n)$ equals $L_i(m_i)$ with probability at least $1-\epsilon/3$.  Thus, with probability at least $1-\epsilon/3$, $\Phi_i^{(k)}(Y^n)$ equals $m_i$. This holds true both when $k$ is an adversarial node or a non-adversarial node. 
\item[$(ii)$] Next, consider $\Phi_i^{(j)}$ and let $X_k^n$ be the random variable denoting user $k$'s  potentially adversarial channel input. In this case, $\Gamma_i^j(Y^n)$ may not equal $L_i(m_i)$ (with a high probability) as the vector transmitted from node $k$ may not be a valid codeword from the codebook for $F_k$. Let $\Prob_{F_i,M_i,F_j,X_k^n,Y^n}$ be the joint probability distribution of the quintuplet $(F_i,M_i,F_j,X_k^n,Y^n)$ when $X_k^n$ is generated according to a possibly adversarial distribution $\Prob_{X_k^n}$, $F_i$ and $F_j$ are generated according to our random code constructions, $M_i$ and $M_j$ are uniformly chosen from $\cM_i$ and $\cM_j$ respectively, and $Y^n$ is the channel output random variable when the channel inputs are $F_i(M_i)$, $F_j(M_j)$, and $X_k^n$.  Further, let $L_i$, $G_i$, and $\Gamma_i^{(j)}$  be the random variables denoting the maps described in our code design (such that $F_i=L_i\circ G_i$). In the following chain of inequalities, we let $f_i$, $l_i$, $g_i$, $\gamma_i^{(j)}$, $m_i$, $f_j$, $x_k^n$, and $y^n$  denote instantiations of the  random variables $F_i$, $L_i$, $G_i$, $\Gamma_i^{(j)}$, $M_i$, $F_j$, $X_k^n$, and $Y^n$ respectively. We have, 
\begin{align*}
&\Prob_{F_i,M_i,F_j,M_j,X_k^n,Y^n}(\Phi_i^{(j)}(Y^n)\notin \set{M_i,\bot})	\\
&= \Prob_{F_i,M_i,F_j,M_j,X_k^n,Y^n}(\Gamma_i^{(j)}(Y^n)\in L_i(\cM_i)\setminus\set{L_i(M_i)})	\\
&= \sum_{\substack{l_i,g_i,m_i,f_j,m_j,x_k^n\\ y^n: \gamma_i^{(j)}(y^n)\in l_i(\cM_i)\setminus\set{l_i(m_i)}}}\Prob_{L_i}(l_i)\Prob_{G_i}(g_i)\Prob_{M_i}(m_i)\Prob_{F_j}(f_j)\Prob_{M_j}(m_j)\Prob_{X_k^n}(x_k^n)W_{Y^n|X_i^n,X_j^n,X_k^n}(y^n|f_i(m_i),f_j(m_j),x_k^n) \\&=\sum_{g_i,m_i,f_j,m_j,x_k^n,y^n}\Prob_{G_i}(g_i)\Prob_{M_i}(m_i)\Prob_{F_j}(f_j)\Prob_{M_j}(m_j)\Prob_{X_k^n}(x_k^n)\\
&\qquad\qquad\times\sum_{l_i}\Prob_{L_i}(l_i) W_{Y^n|X_i^n,X_j^n,X_k^n}(y^n|g_i(l_i(m_i)),f_j(m_j),x_k^n) \mathbf{1}_{\set{\gamma_i^{(j)}(y^n)\in l_i(\cM_i)\setminus\set{l_i(m_i)}}}\\
&=\sum_{g_i,m_i,f_j,m_j,x_k^n,y^n}\Prob_{G_i}(g_i)\Prob_{M_i}(m_i)\Prob_{F_j}(f_j)\Prob_{M_j}(m_j)\Prob_{X_k^n}(x_k^n)\\
&\qquad\qquad\times\sum_{l_i}\sum_{\tilde{m}_i\in\tilde{\cM}_i}\sum_{\tilde{m}_i'\in\tilde{\cM}_i}\Prob_{L_i}(l_i) \mathbf{1}_{\set{l_i(m_i)=\tilde{m}_i}} W_{Y^n|X_i^n,X_j^n,X_k^n}(y^n|g_i(\tilde{m}_i),f_j(m_j),x_k^n) \mathbf{1}_{\set{\gamma_i^{(j)}(y^n)=\tilde{m}_i'}}\mathbf{1}_{\set{\tilde{m}_i'\in l_i(\cM_i)\setminus\set{\tilde{m}_i}}}\\
&=\sum_{g_i,m_i,f_j,m_j,x_k^n,y^n}\Prob_{G_i}(g_i)\Prob_{M_i}(m_i)\Prob_{F_j}(f_j)\Prob_{M_j}(m_j)\Prob_{X_k^n}(x_k^n)\\
&\qquad\qquad\times\sum_{\tilde{m}_i\in\tilde{\cM}_i}\sum_{\tilde{m}_i'\in\tilde{\cM}_i} W_{Y^n|X_i^n,X_j^n,X_k^n}(y^n|g_i(\tilde{m}_i),f_j(m_j),x_k^n) \mathbf{1}_{\set{\gamma_i^{(j)}(y^n)=\tilde{m}_i'}}\sum_{l_i}\Prob_{L_i}(l_i) \mathbf{1}_{\set{l_i(m_i)=\tilde{m}_i}}\mathbf{1}_{\set{\tilde{m}_i'\in l_i(\cM_i)\setminus\set{\tilde{m}_i}}}\\
&=\sum_{g_i,m_i,f_j,m_j,x_k^n,y^n}\Prob_{G_i}(g_i)\Prob_{M_i}(m_i)\Prob_{F_j}(f_j)\Prob_{M_j}(m_j)\Prob_{X_k^n}(x_k^n)\\
&\qquad\qquad\times\sum_{\tilde{m}_i\in\tilde{\cM}_i}\sum_{\tilde{m}_i'\in\tilde{\cM}_i} W_{Y^n|X_i^n,X_j^n,X_k^n}(y^n|g_i(\tilde{m}_i),f_j(m_j),x_k^n) \mathbf{1}_{\set{\gamma_i^{(j)}(y^n)=\tilde{m}_i'}}\Prob_{L_i}(L_i(m_i)=\tilde{m}_i,\tilde{m}_i'\in L_i(\cM_i\setminus\set{m_i}))\\
&=\sum_{g_i,m_i,f_j,m_j,x_k^n,y^n}\Prob_{G_i}(g_i)\Prob_{M_i}(m_i)\Prob_{F_j}(f_j)\Prob_{M_j}(m_j)\Prob_{X_k^n}(x_k^n)\\
&\qquad\qquad\times\sum_{\tilde{m}_i\in\tilde{\cM}_i}\sum_{\tilde{m}_i'\in\tilde{\cM}_i} W_{Y^n|X_i^n,X_j^n,X_k^n}(y^n|g_i(\tilde{m}_i),f_j(m_j),x_k^n) \mathbf{1}_{\set{\gamma_i^{(j)}(y^n)=\tilde{m}_i'}}\frac{1}{\lvert\tilde{\cM}_i\rvert}\frac{\lvert\cM_i\rvert-1}{\lvert\tilde{\cM}_i\rvert-1}\\
&=\sum_{\tilde{m}_i\in\tilde{\cM}_i}  \frac{1}{\lvert\tilde{\cM}_i\rvert}\frac{\lvert\cM_i\rvert-1}{\lvert\tilde{\cM}_i\rvert-1}\\
&\leq \frac{\lvert\cM_i\rvert}{\lvert\tilde{\cM}_i\rvert}= \epsilon/3.
\end{align*}
Thus, with high probability, for each non-adversarial user $i$, at least one of the decoders $\Phi_i^{(j)}$ or $\Phi_i^{(k)}$ outputs the true message while the other decoder outputs either the true message or $\bot$.
\end{description}
\end{proof}

\begin{proof}[{Proof (Converse of Theorem~\ref{thm:random})}]
We show converse for the weak adversary. Suppose $(F_1,F_2,F_3,\Phi)$ is a $(2^{nR_1},2^{nR_2},2^{nR_3},n)$ randomized code such that $P_e\leq \epsilon$ for some $\epsilon>0$. Recall that $F_1,F_2,F_3$ are independent. 
Let $M_i\sim\textup{Unif}(\mathcal{M}_i)$, $i=1,2,3$ be independent. Let $\hat{M}_i= \Psi_i(Y^n,F_1,F_2,F_3), \, i = 1,2,3$. Then, $\epsilon$ is an upperbound on~\eqref{eq:pe1_weak} which is given by
\begin{align*}
 P_{e,1}^{\text{weak}}&= \max_{x_1^n} \;\Prob_{F_2,F_3,\Phi}\Big( (\Msgh_2,\Msgh_3)\neq(\Msg_2,\Msg_3)\Big|X_1^n=x_1^n,X_2^n=F_2(\Msg_2), X_3^n=F_3(\Msg_3)\Big)\\
 &= \max_{p_{X_1^n}} \Prob_{F_2,F_3,\Phi}\Big( (\Msgh_2,\Msgh_3)\neq(\Msg_2,\Msg_3)\Big| X_2^n=F_2(\Msg_2), X_3^n=F_3(\Msg_3)\Big).
\end{align*}
We consider the following $p_{X_1^n}$.
\[ p_{X_1^n}(x_1^n) = \prod_{i=1}^n q_{X_{1,i}} (x_{1,i}).\]
By Fano's inequality, under this $p_{X_1^n}$ and when $X_i^n=F_i(W_i)$, $i=2,3$, 
\begin{align*}
H(\Msg_2,\Msg_3|Y^n,\Phi) \leq 1+n\epsilon(R_2+R_3).
\end{align*}
Ignoring small terms, we have
\begin{align*}
n(R_2+R_3) &\leq H(M_2,M_3)\\
 &\leq H(M_2,M_3|\Phi,F_2,F_3)\\
 &\stackrel{\text{(a)}}{\approx} I(M_2,M_3;Y^n|\Phi,F_2,F_3)\\
 &= \sum_{i=1}^n I(M_2,M_3;Y_i|Y^{i-1},\Phi,F_2,F_3)\\
 &\leq \sum_{i=1}^n I(M_2,M_3,\Phi,F_2,F_3,Y^{i-1};Y_i)\\
 &= \sum_{i=1}^n I(M_2,M_3,\Phi,F_2,F_3,Y^{i-1},X_{2,i},X_{3,i};Y_i)\\
 &\stackrel{\text{(b)}}{=} \sum_{i=1}^n I(X_{2,i},X_{3,i};Y_i),
\end{align*}
where (a) follows from Fano's inequality (ignoring an $O(n\epsilon)$ term), (b) follows from the memorylessness of the channel and the independence of $X_{1,i}$ over $i=1,\ldots,n$ for the particular $p_{X_1^n}$ under consideration.

Let $U\sim\textup{Unif}\{1,2,\ldots,n\}$ independent of $(M_1,M_2,M_3,F_1,F_2,F_3,\Phi,Y^n)$. We have (where we ignore an additive $O(\epsilon)$ term) 
\begin{align*}
R_2+R_3\leq I(X_{2,U},X_{3,U};Y_U|U).
\end{align*}
Since, the above bound holds for all $p_{X_1^n}(x_1^n)=\prod_{i=1}^n q_{X_{i,i}}(x_{1,i})$, and noticing that conditioned on $X_{1,U},X_{2,U},X_{3,U}$ the channel law $\mach$ gives the conditional probability of $Y_U$, we may write
\begin{align}
R_2+R_3 \leq  \min_{q(x_1|u)} I(X_2,X_3;Y|U). \label{eq:converseR2+R3}
\end{align}
We note that the distribution of $U,X_1,X_2,X_3,Y$ is $p(u)q(x_1|u)p(x_2|u)p(x_3|u)\mach(y|x_1,x_2,x_3)$ where $p(x_2|u)$ is determined by the distribution of $F_2$ and $p(x_3|u)$ is determined by the distribution of $F_3$.
 
Proceeding similarly, for $p_{X_1^n}(x_1^n)=\prod_{i=1}^n q_{X_{1,i}}(x_{1,i})$,
\begin{align*}
nR_2 &\leq H(M_2)\\
 &\leq H(M_2|M_3,\Phi,F_2,F_3)\\
 &\approx I(M_2;Y^n|M_3,\Phi,F_2,F_3)\\
 &= \sum_{i=1}^n I(M_2;Y_i|Y^{i-1},M_3,\Phi,F_2,F_3)\\
 &= \sum_{i=1}^n I(M_2,X_{2,i};Y_i|X_{3,i},Y^{i-1},M_3,\Phi,F_2,F_3)\\
 &\leq \sum_{i=1}^n I(X_{2,i}, Y^{i-1},M_2,M_3,\Phi,F_2,F_3; Y_i|X_{3,i})\\
 &= \sum_{i=1}^n I(X_{2,i};Y_i|X_{3,i}).
\end{align*}
Hence, we have
\begin{align}
R_2 \leq \min_{q_(x_1|u)} I(X_2;Y|X_3,U), \label{eq:converseR2}
\end{align}
where the joint distribution  of the random variables is $p(u)q(x_1|u)p(x_2|u)p(x_3|u)\mach(y|x_1,x_2,x_3)$. We note that $p(u)p(x_2|u)p(x_3|u)$ are the same as in \eqref{eq:converseR2+R3}. Similarly,
\begin{align}
R_3\leq \min_{q(x_1|u)} I(X_3;Y|X_2,U). \label{eq:converseR3}
\end{align}

Similarly, considering $P_{e,2}^{\text{weak}}$ with $p_{X_2^n}(x_2^n)=\prod_{i=1}^n q_{X_{2,i}}(x_{2,i})$ (and $X_i^n=F_i(W_i)$, $i=1,3$), we get
\begin{align}
R_3 &\leq \min_{q(x_2|u)} I(X_3;Y|X_1,U), \label{eq:converseR3B}\\
R_1 &\leq \min_{q(x_2|u)} I(X_1;Y|X_3,U), \label{eq:converseR1B}\\
R_3+R_1 &\leq \min_{q(x_2|u)} I(X_3,X_1;Y|U), \label{eq:converseR3+R1}
\end{align}
where the joint distribution of the random variables is $p(u)p(x_1|u)q(x_2|u)p(x_3|u)\mach(y|x_1,x_2,x_3)$. We note that $p(u)$ and $p(x_3|u)$ here is the same as in \eqref{eq:converseR2+R3}-\eqref{eq:converseR3}. Considering $P_{e,3}^{\text{weak}}$ with $p_{X_3^n}(x_3^n)=\prod_{i=1}^n q_{X_{3,i}}(x_{3,i})$ (and $X_i^n=F_i(W_i)$, $i=1,2$), we similarly arrive at
\begin{align}
R_1 &\leq \min_{q(x_3|u)} I(X_1;Y|X_3,U), \label{eq:converseR1C}\\
R_2 &\leq \min_{q(x_3|u)} I(X_2;Y|X_2,U), \label{eq:converseR2C}\\
R_1+R_2 &\leq \min_{q(x_3|u)} I(X_1,X_2;Y|U), \label{eq:converseR1+R2}
\end{align}
where the joint distribution of the random variables is $p(u)p(x_1|u)p(x_2|u)q(x_3|u)\mach(y|x_1,x_2,x_3)$. The $p(u)$, $p(x_1|u)$, and $p(x_2|u)$ are the same as in \eqref{eq:converseR2+R3}-\eqref{eq:converseR1+R2}. By Caratheodory's theorem, it is enough to choose $U$ such that $|\mathcal{U}|\leq 3.$ This completes the proof of converse.
\end{proof}

\section{Proof of Theorem~\ref{thm:deterministic}}
\begin{proof}[{Proof of Theorem~\ref{thm:deterministic}}]
We argue that under the given condition,
$$\mathcal{R}_{\random}^{\weak} = \mathcal{R}_{\deterministic}.$$

The proof is similar to that of the corresponding result 
\cite[Theorem 1]{Jahn81} for AVMAC. Jahn uses an extension of Ahlswede's
elimination technique to reduce the randomness of each encoder (shared
with the decoder) to take only $n^2$ values. A further extension to three
users in a similar manner leads to the same result for our setup
under the week adversary model - the rate-triples in
$\mathcal{R}_{\random}^{\weak}$ can be achieved by codes where each encoder
randomness is limited to take only $n^2$ values. A deterministic code for 
any rate triple in
$\mathcal{R}_{\random}^{\weak}$ can then be constructed as the concatenation of
two codes. A $o(n)$-length deterministic code with
$n^2$ codewords  can be used to communicate $2\log_2 n$ bits out of each
message. An arbitrarily small rate is required for this code. 
Then a randomized code, which uses the $2\log_2 n$ bits messages of the
deterministic code as encoder randomness', can be used to transmit the rest of
the message bits.
\end{proof}

\section{Proof of Theorem~\ref{thm:symmetrizability}}\label{app:3}
For a random variable $X$, let $\mathcal{P}_{X}$ denote the set of empirical distributions for blocklength $n$. For a distribution $P_X \in \mathcal{P}_X$, let $T^n_X$ be the set of typical sequences with relative frequencies specified by $P_X$.

\begin{lemma}\label{codebook}
For any  $\epsilon>0,\,  n\geq n_0(\epsilon), \, N_1, N_2, N_3\geq\exp(n\epsilon)$ and types $P_1, P_2, P_3$ over $\mathcal{X}_1, \mathcal{X}_2, \mathcal{X}_3$ respectively, there exists codebooks $\left\{\mathbf{x}_{11}, \ldots,  \mathbf{x}_{1N_1}\in \mathcal{X}_1^n\right\}, \left\{\mathbf{x}_{21}, \ldots,  \mathbf{x}_{2N_2}\in \mathcal{X}_2^n\right\}$ and $ \left\{\mathbf{x}_{31}, \ldots,  \mathbf{x}_{3N_3}\in \mathcal{X}_3^n\right\}$ each of type $P_1, P_2$ and $P_3$ respectively such that for every permutation $(i,j,k)$ of $(1,2,3)$, for every $(\mathbf{x}_i,\mathbf{x}_j, \mathbf{x}_k)  \in \mathcal{X}^n_i \times \mathcal{X}^n_j \times \mathcal{X}^n_k$, for every joint type $P_{X_i,X'_i,X_j,X'_j,X_k,X'_k}$ and for $R_i \defineqq (1/ n )\log_2{N_i},R_j \defineqq (1/ n )\log_2{N_j}$ and $R_k \defineqq (1/ n )\log_2{N_k}$, the following holds.
\begin{align}
&\frac{1}{N_iN_j}|\{(r,s): (\mathbf{x}_{ir},\mathbf{x}_{js})\in T^{n}_{X_i, X_j} \}| \leq \exp\left(-n\epsilon\right)\text{ if }I(X_i;X_j)>\epsilon,   \label{code_eq1}\\
&\frac{1}{N_iN_j}|\{(r,s): (\mathbf{x}_{ir},\mathbf{x}_{js}, \mathbf{x}_k)\in T^{n}_{X_i, X_j, X_k} \}| \leq \exp\left(-n\epsilon/2\right) \text{ if } I(X_i,X_j;X_k)>\epsilon,  \label{code_eq2}\\
&|\{u:(\mathbf{x}_i, \mathbf{x}_{iu}, \mathbf{x}_j, \mathbf{x}_k)\in T^{n}_{X_i, X_i^{'}, X_j, X_k}\}|\leq \exp\left\{n\left(|R_i-I(X_i';X_i,X_j,X_k)|^+ +\epsilon\right)\right\}, \label{code_eq3}\\
&|\{u, v:(\mathbf{x}_i, \mathbf{x}_{iu}, \mathbf{x}_j, \mathbf{x}_{jv}, \mathbf{x}_k)\in T^{n}_{X_i, X_i^{'}, X_j, X'_j, X_k}\}|\leq \exp\left\{n\left(|R_i+R_j-I(X_i',X_j';X_i,X_j,X_k)|^+ +\epsilon\right)\right\}, \label{code_eq4}\\
&\frac{1}{N_iN_j}|\{(r,s): (\mathbf{x}_{ir},\mathbf{x}_{iu}, \mathbf{x}_{js},  \mathbf{x}_k)\in T^{n}_{X_i,X'_i,  X_j, X_k} \text{ for some }u\neq r \}| \leq \exp\left(-n\epsilon/2\right) \nonumber\\
&\qquad \qquad\text{ if } I(X_i,X_j;X'_i,X_k)-|R_i-I(X'_i;X_k)|^{+}>\epsilon, \label{code_eq5} \\
&\frac{1}{N_iN_j}|\{(r,s): (\mathbf{x}_{ir},\mathbf{x}_{iu}, \mathbf{x}_{js},  \mathbf{x}_k, \mathbf{x}_{kt}\in T^{n}_{X_i,X'_i,X_j,X_k,X_k'} \text{ for some }u\neq r \text{ and some } t \in [1:N_k]\}| \leq \exp\left(-n\epsilon/2\right)  \nonumber\\
&\qquad \qquad\text{ if } I(X_i,X_j;X'_i, X'_k, X_k)-|R_i+R_k-I(X'_i,X'_k;X_k)|^{+}>\epsilon, \text{ and} \label{code_eq6} \\
&\frac{1}{N_iN_j}|\{(r,s): (\mathbf{x}_{ir},\mathbf{x}_{iu}, \mathbf{x}_{js}, \mathbf{x}_{jv}, \mathbf{x}_k)\in T^{n}_{X_i,X'_i,X_j,X_j',X_k} \text{ for some }u\neq r \text{ and }	v\neq s\}| \leq \exp\left(-n\epsilon/2\right) \nonumber \\
&\qquad \qquad\text{ if } I(X_i,X_j;X'_i ,X'_j, X_k)-|R_i+R_j-I(X'_i,X'_j;X_k)|^{+}>\epsilon. \label{code_eq7}
\end{align}
\end{lemma}
\begin{proof}[Proof sketch]
For each $i=1,2,3$, we draw $N_i$ codewords independently and uniformly at random from the type class $P_i$. The inequalities follow by first using standard counting arguments based on types to calculate the expected number of codewords lying in the sets specified by~\eqref{code_eq1}-~\eqref{code_eq7} and then applying concentration inequalities to bound probability that the actual number of codewords deviates from the expectation. We skip the details here as the proof follows similarly to~\cite[Lemma~3]{CsiszarN88}.
	
\end{proof}

\noindent We recall the definition of the decoder.
\begin{defn*}[Decoder]
For $\eta>0, \,(m_1,m_2,m_3)\in \mathcal{M}_1\times\mathcal{M}_2\times\mathcal{M}_3$, and encoding maps, $f_1:\mathcal{M}_1\rightarrow \mathcal{X}_1^n, \, f_2:\mathcal{M}_2\rightarrow \mathcal{X}_2^n$ and $f_3:\mathcal{M}_3\rightarrow \mathcal{X}_3^n$, the decoding set $D_{m_1,m_2,m_3}  \subseteq \mathcal{Y}^n$ is defined as the intersection of three decoding sets $D_{m_1},D_{m_2},D_{m_3}\subseteq \mathcal{Y}^n$ defined below.\\
 For $m_1\in \mathcal{M}_1$, a sequence $\mathbf{y} \in D_{m_1}$ if there exists  some permutation $(j, k)$ of $(2,3)$, $m_j \in \mathcal{M}_j,\,\mathbf{x}_k\in \mathcal{X}^n_k$, and random variables $X_1, X_j, X_k$ with $(f_1(m_1), f_j(m_j), \mathbf{x}_{k}, \mathbf{y}) \in T^{n}_{X_1,X_j,X_k,Y}$ and $D(P_{X_1,X_j,X_k, Y}||P_{X_1}\times P_{X_j}\times P_{X_k}\times W)< \eta$ such that the following hold:
	\begin{description}			
			 \item[(a)] If there exists $m'_1 \in \mathcal{M}_1$, $m'_1 \neq m_1$, $m'_j \in \mathcal{M}_j$, $m'_j\neq m_j$, $\mathbf{x}'_k\in \mathcal{X}^n_k$, and random variables $X_1',X_j', X_k'$ such that $(f_1(m_1), f_1(m'_1),f_j(m_j), f_j(m'_j), \mathbf{x}_k, \mathbf{x}'_k, \mathbf{y}) \in T^{n}_{X_1 ,X_1', X_j,X'_j, X_k ,X_k',Y}$ and  $D(P_{X_1',X'_j,X'_k, Y}||P_{X'_1}\times P_{X'_j}\times P_{X'_k}\times W)< \eta$, then $I(X_1,X_j,Y;X'_1,X'_j|X_k) < \eta$.

			\item[(b)] If there exists $m'_1 \in \mathcal{M}_1$, $m'_1 \neq m_1$, $\mathbf{x}'_k\in \mathcal{X}^n_k$, and random variables $X_1',X_k'$ such that $(f_1(m_1),f_1(m'_1),f_j(m_j),\allowbreak  \mathbf{x}_k, \mathbf{x}'_k, \mathbf{y}) \in T^{n}_{X_1, X_1',X_j, X_k, X_k', Y}$ and  $D(P_{X_1',X_j,X_k', Y}||P_{X_1'}\times P_{X_j}\times P_{X_k'}\times W)< \eta$, then $I(X_1,X_j,Y;X_1'|X_k) < \eta$.

			\item[(c)] If there exists $m'_1 \in \mathcal{M}_1$, $m'_1 \neq m_1$, $\mathbf{x}_j\in \mathcal{X}^n_j$, $m_k\in \mathcal{M}_k$, and random variables $X_1',X_j', X_k'$ such that $(f_1(m_1),f_1(m'_1),f_j(m_j), \mathbf{x}_j, \mathbf{x}_k, f_k(m_k), \mathbf{y}) \in T^{n}_{X_1, X_1', X_j, X_j', X_k, X_k',Y}$ and  $D(P_{X_1',X_j',X_k', Y}||P_{X_1'}\times P_{X_j}\times P_{X_k'}\times W)< \eta$, then $I(X_1,X_j,Y;X_1',X_k'|X_k) < \eta$.

	\end{description}
The decoding sets $D_{m_2}$ and $D_{m_3}$ are defined similarly. 
\end{defn*}
The following Lemma implies that the decoder above is well defined.
\begin{duplicatelemma}[\ref{lemma:dec}]
Suppose the channel \mach is not symmetrizable. Let $P_{X_1}\in \mathcal{P}_{X_1}, P_{X_2} \in \mathcal{P}_{X_2}$ and $P_{X_3}\in \mathcal{P}_{X_3}$ be distributions such that for some $\alpha >0,$ $\min_{x_1}P_{X_1}(x_1),\min_{x_2}P_{X_2}(x_2),\min_{x_3}P_{X_3}(x_3)\geq \alpha$. Let $f_1:\mathcal{M}_1\rightarrow T^{n}_{X_1}, \, f_2:\mathcal{M}_2\rightarrow T^{n}_{X_2}$ and $f_3:\mathcal{M}_3\rightarrow T^{n}_{X_3}$ be any encoding maps. There exists a choice of $\eta>0$  such that if $(\tilde{m}_1,\tilde{m}_2,\tilde{m}_3)\neq(m_1,m_2,m_3),\, D_{\tilde{m}_1,\tilde{m}_2,\tilde{m}_3}\cap D_{m_1,m_2,m_3} = \emptyset$.
\end{duplicatelemma}
\begin{proof}
To prove the lemma, it is sufficient to show that the decoding regions of the decoders for each user are disjoint. Let us consider the decoder for user 1. Suppose $\mathbf{y}\in\mathcal{Y}^n$ is such that $\mathbf{y}$ lies in the decoding regions, $D_{m_1},D_{\tilde{m}_1}$ of $m_1,\tilde{m}_1\in \mathcal{M}_1$ where $\tilde{m}_1\neq m_1$. Then there exist permutations $(i,j)$ and $(\tilde{i},\tilde{j})$ of $(2,3)$ such that one of the following cases holds.\\
{\bf Case 1}: $(\tilde{i}, \tilde{j}) = (i,j)$\\
There exist $m_j,\tilde{m}_j\in \mathcal{M}_j$, sequences $\mathbf{x}_i,\tilde{\mathbf{x}}_i\in \mathcal{X}_i^n$, and random variables $X_1,\tilde{X}_1,X_j,\tilde{X}_j,X_i,\tilde{X}_i$ with $(f_1(m_1),f_1(\tilde{m}_1),\allowbreak f_j(m_j),f_j(\tilde{m}_j),\mathbf{x}_i,\tilde{\mathbf{x}}_i)\in T^{n}_{X_1\tilde{X}_1X_j\tilde{X}_jX_i\tilde{X}_i}$ such that $D(P_{X_1X_jX_i Y}||P_{X_1}\times P_{X_j}\times P_{X_i}\times W),\, D(P_{\tilde{X}_1\tilde{X}_j\tilde{X}_i Y}||P_{\tilde{X}_1}\times P_{\tilde{X}_j}\times P_{\tilde{X}_i}\times W)< \eta$ and
\begin{description}
	\item [Case 1(a)] if $\tilde{m}_j\neq m_j$, then $I(X_1X_jY;\tilde{X}_1\tilde{X}_j|X_i), I(\tilde{X}_1\tilde{X}_jY;X_1X_j|\tilde{X}_i) < \eta $.
	\item [Case 1(b)] if $\tilde{m}_j =  m_j$, then $\tilde{X}_j = X_j$ and $I(X_1X_jY;\tilde{X}_1|X_i), I(\tilde{X}_1X_jY;X_1|\tilde{X}_i) < \eta$. 
\end{description}
{\bf Case 2}: $(\tilde{i}, \tilde{j}) = (j,i)$\\
There exist $m_j\in \mathcal{M}_j$, $\tilde{m}_i\in \mathcal{M}_i$, sequences $\tilde{\mathbf{x}}_{j}\in \mathcal{X}_{j}^n$, $\mathbf{x}_i\in \mathcal{X}_i^n$ and random variables $X_1$, $\tilde{X}_1$, $X_j$, $\tilde{X}_{j}$, $X_i$, $\tilde{X}_i$ with $(f_1(m_1)$, $f_1(\tilde{m}_1)$, $f_j(m_j)$, $\tilde{\mathbf{x}}_{j}$, $\mathbf{x}_i$, $f_i(\tilde{m}_i))\in T^{n}_{X_1\tilde{X}_1X_j\tilde{X}_{j}X_i\tilde{X}_i}$ such that $D(P_{X_1X_jX_i Y}||P_{X_1}\times P_{X_j}\times P_{X_i}\times W)$, $ D(P_{\tilde{X}_1\tilde{X}_{j}\tilde{X}_i Y}||P_{\tilde{X}_1}\times P_{\tilde{X}_{j}}\times P_{\tilde{X}_i}\times W)< \eta$ and
 $I(X_1  X_j Y;\tilde{X}_1 \tilde{X}_i|X_i)$,$ I(\tilde{X}_1\tilde{X}_iY;X_1X_j|\tilde{X}_{j})<\eta$.\\

\noindent We first analyze {\bf Case 1(a)}. Let $W_{Y|X_1X_jX_i}$ be denoted by $W$.
\begin{align*}
&D(P_{X_1X_jX_i Y}||P_{X_1}\times P_{X_j}\times P_{X_i}\times W) + D(P_{\tilde{X}_1,\tilde{X}_j}||P_{\tilde{X}_1}\times P_{\tilde{X}_j}) + I(X_1X_jY;\tilde{X}_1\tilde{X}_j|X_i) \stackrel{(a)}{=} \\
&\quad\sum_{x_1,x_j,x_i,y}P_{X_1X_jX_iY}(x_1,x_j,x_i,y)\log{\frac{P_{X_1X_jX_iY}(x_1,x_j,x_i,y)}{P_{X_1}(x_1)P_{X_j}(x_j)P_{X_i}(x_i)W(y|x_1,x_j,x_i)}} + \sum_{\tilde{x}_1,\tilde{x}_j}P_{\tilde{X}_1\tilde{X}_j}(\tilde{x}_1,\tilde{x}_j)\log{\frac{P_{\tilde{X}_1\tilde{X}_j}(\tilde{x}_1,\tilde{x}_j)}{P_{\tilde{X}_1}(\tilde{x}_1)P_{\tilde{X}_j}(\tilde{x}_j)}}\\
&\quad + \sum_{x_1,\tilde{x}_1,x_j,\tilde{x}_j,x_i,y}P_{X_1\tilde{X}_1X_j\tilde{X}_jX_iY}(x_1,\tilde{x}_1,x_j,\tilde{x}_j,x_i,y)\log{\frac{P_{X_1\tilde{X}_1X_j\tilde{X}_jY|X_i}(x_1,\tilde{x}_1,x_j,\tilde{x}_j,y|x_i)}{P_{X_1X_jY|X_i}(x_1,x_j,y|x_i)P_{\tilde{X}_1\tilde{X}_j|X_i}(\tilde{x}_1,\tilde{x}_j|x_i)}}\\
&\quad =\sum_{x_1,\tilde{x}_1,x_j,\tilde{x}_j,x_i,y}P_{X_1\tilde{X}_1X_j\tilde{X}_jX_iY}(x_1,\tilde{x}_1,x_j,\tilde{x}_j,x_i,y)\log{\frac{P_{X_1\tilde{X}_1X_j\tilde{X}_jX_iY}(x_1,\tilde{x}_1,x_j,\tilde{x}_j,x_i,y)}{P_{X_1}(x_1)P_{\tilde{X}_1}(\tilde{x}_1)P_{X_j}(x_j)P_{\tilde{X}_j}(\tilde{x}_j)P_{X_i|\tilde{X}_1\tilde{X}_j}(x_i|\tilde{x}_1,\tilde{x}_j)W(y|x_1,x_j,x_i)}}\\
&\quad = D(P_{X_1\tilde{X}_1X_j\tilde{X}_jX_iY}||P_{X_1}P_{\tilde{X}_1}P_{X_j}P_{\tilde{X}_j}P_{X_i|\tilde{X}_1\tilde{X}_j}W)\\
&\quad \stackrel{\text{(b)}}{\geq} D(P_{X_1\tilde{X}_1X_j\tilde{X}_jY}||P_{X_1}P_{\tilde{X}_1}P_{X_j}P_{\tilde{X}_j}\tilde{V}_1)\text{ where }\tilde{V}_1(y|x_1,\tilde{x}_1,x_j,\tilde{x}_j) = \sum_{x_i}P_{X_i|\tilde{X}_1\tilde{X}_j}(x_i|\tilde{x}_1,\tilde{x}_j)W(y|x_1,x_j,x_i),
\end{align*}
where (b) follows from the log sum inequality. From the given conditions, we know that the term on the LHS of (a) is no greater than $3\eta$. Thus, $D(P_{X_1\tilde{X}_1X_j\tilde{X}_jY}||P_{X_1}P_{\tilde{X}_1}P_{X_j}P_{\tilde{X}_j}\tilde{V}_1) \leq 3\eta$. Using Pinsker's inequality, it follows that 
\begin{equation}\label{eq:a5}
\sum_{x_1,\tilde{x}_1,x_j,\tilde{x}_j,y}\Big|P_{X_1\tilde{X}_1X_j\tilde{X}_jY}(x_1,\tilde{x}_1,x_j,\tilde{x}_j,y)-P_{X_1}(x_1)P_{\tilde{X}_1}(\tilde{x}_1)P_{X_j}(x_j)P_{\tilde{X}_j}(\tilde{x}_j)\tilde{V}_1(y|x_1,\tilde{x}_1,x_j,\tilde{x}_j)\Big| \leq c\sqrt{3\eta},
\end{equation}
where $c$ is some positive constant. Following a similar line of argument, we can show that 
\begin{align*}
3\eta &\geq D(P_{\tilde{X}_1\tilde{X}_j\tilde{X}_i Y}||P_{\tilde{X}_1}\times P_{\tilde{X}_j}\times P_{\tilde{X}_i}\times W) + D(P_{X_1X_j}||P_{X_1}\times P_{X_j}) + I(\tilde{X}_1\tilde{X}_jY;X_1 X_j|\tilde{X}_i)\\
& \geq D(P_{X_1\tilde{X}_1X_j\tilde{X}_jY}||P_{X_1}P_{\tilde{X}_1}P_{X_j}P_{\tilde{X}_j}V_1)\text{ where }V_1(y|x_1,\tilde{x}_1,x_j,\tilde{x}_j) = \sum_{\tilde{x}_i}P_{\tilde{X}_i|X_1X_j}(\tilde{x}_i|x_1,x_j)W(y|\tilde{x}_1,\tilde{x}_j,\tilde{x}_i)
\end{align*}
\noindent Using Pinsker's inequality, it follows that 
\begin{equation}\label{eq:a6}
\sum_{x_1,\tilde{x}_1,x_j,\tilde{x}_j,y}\Big|P_{X_1\tilde{X}_1X_j\tilde{X}_jY}(x_1,\tilde{x}_1,x_j,\tilde{x}_j,y)-P_{X_1}(x_1)P_{\tilde{X}_1}(\tilde{x}_1)P_{X_j}(x_j)P_{\tilde{X}_j}(\tilde{x}_j)V_1(y|x_1,\tilde{x}_1,x_j,\tilde{x}_j)\Big| \leq c\sqrt{3\eta}.
\end{equation}
From \eqref{eq:a5} and \eqref{eq:a6}, 
\begin{equation*}
\sum_{x_1,\tilde{x}_1,x_j,\tilde{x}_j,y}P_{X_1}(x_1)P_{\tilde{X}_1}(\tilde{x}_1)P_{X_j}(x_j)P_{\tilde{X}_j}(\tilde{x}_j)\Big|\tilde{V}_1(y|x_1,\tilde{x}_1,x_j,\tilde{x}_j)-V_1(y|x_1,\tilde{x}_1,x_j,\tilde{x}_j)\Big| \leq 2c\sqrt{3\eta}.
\end{equation*}
This implies that 
\begin{equation}\label{eq:1a}
\max_{x_1,\tilde{x}_1,x_j,\tilde{x}_j,y}\Big|\tilde{V}_1(y|x_1,\tilde{x}_1,x_j,\tilde{x}_j)-V_1(y|x_1,\tilde{x}_1,x_j,\tilde{x}_j)\Big| \leq \frac{2c\sqrt{3\eta}}{\alpha^4}.
\end{equation}
Similar to~\cite[(A.15) on page~748]{AhlswedeC99}, since $\mach$ is not $\cX_1\times\cX_j$-{\em symmetrizable by}~$\cX_i$ (i.e.,~\eqref{eq:symm1} does not hold for $(i,j,k) = (i,j,1)$), we can show that for any pair of channels $P_{\tilde{X}_i|X_1X_j}$ and $P_{X_i|\tilde{X}_1\tilde{X}_j}$, there exists $\zeta_1>0$ such that
\begin{equation*}
\max_{x_1,\tilde{x}_1,x_j,\tilde{x}_j,y}\Big|\tilde{V}_1(y|x_1,\tilde{x}_1,x_j,\tilde{x}_j)-V_1(y|x_1,\tilde{x}_1,x_j,\tilde{x}_j)\Big| \geq \zeta_1.
\end{equation*}
This contradicts \eqref{eq:1a} if $\eta < \frac{\zeta_1^2\alpha^8}{12c^2}$.\\

\noindent We now analyze {\bf Case 1(b)}. 
\begin{align*}
&D(P_{X_1X_jX_i Y}||P_{X_1}\times P_{X_j}\times P_{X_i}\times W)  + I(X_1X_jY;\tilde{X}_1|X_i) \stackrel{(a)}{=} \\
&\quad\sum_{x_1,x_j,x_i,y}P_{X_1X_jX_iY}(x_1,x_j,x_i,y)\log{\frac{P_{X_1X_jX_iY}(x_1,x_j,x_i,y)}{P_{X_1}(x_1)P_{X_j}(x_j)P_{X_i}(x_i)W(y|x_1,x_j,x_i)}} \\
&\quad\quad + \sum_{x_1,\tilde{x}_1,x_j,x_i,y}P_{X_1\tilde{X}_1X_jX_iY}(x_1,\tilde{x}_1,x_j,x_i,y)\log{\frac{P_{X_1\tilde{X}_1X_jY|X_i}(x_1,\tilde{x}_1,x_j,y|x_i)}{P_{X_1X_jY|X_i}(x_1,x_j,y|x_i)P_{\tilde{X}_1|X_i}(\tilde{x}_1|x_i)}}\\
&\quad\quad =\sum_{x_1,\tilde{x}_1,x_j,x_i,y}P_{X_1\tilde{X}_1X_jX_iY}(x_1,\tilde{x}_1,x_j,x_i,y)\log{\frac{P_{X_1\tilde{X}_1X_jX_iY}(x_1,\tilde{x}_1,x_j,x_i,y)}{P_{X_1}(x_1)P_{\tilde{X}_1}(\tilde{x}_1)P_{X_j}(x_j)P_{X_i|\tilde{X}_1}(x_i|\tilde{x}_1)W(y|x_1,x_j,x_i)}}\\
&\quad\quad = D(P_{X_1\tilde{X}_1X_jX_iY}||P_{X_1}P_{\tilde{X}_1}P_{X_j}P_{X_i|\tilde{X}_1}W)\\
&\qquad \stackrel{\text{(b)}}{\geq} D(P_{X_1\tilde{X}_1X_jY}||P_{X_1}P_{\tilde{X}_1}P_{X_j}\tilde{V}_2)\text{ where }\tilde{V}_2(y|x_1,\tilde{x}_1,x_j) = \sum_{x_i}P_{X_i|\tilde{X}_1}(x_i|\tilde{x}_1)W(y|x_1,x_j,x_i),
\end{align*}
where (b) follows from the log sum inequality. From the given conditions, we know that the term on the LHS of (a) is no greater than $2\eta$. Thus, $D(P_{X_1\tilde{X}_1X_jY}||P_{X_1}P_{\tilde{X}_1}P_{X_j}\tilde{V}_2) \leq 2\eta$. Using Pinsker's inequality, it follows that 
\begin{equation}\label{eq:a52}
\sum_{x_1,\tilde{x}_1,x_j,y}\Big|P_{X_1\tilde{X}_1X_jY}(x_1,\tilde{x}_1,x_j,y)-P_{X_1}(x_1)P_{\tilde{X}_1}(\tilde{x}_1)P_{X_j}(x_j)\tilde{V}_2(y|x_1,\tilde{x}_1,x_j)\Big| \leq c\sqrt{2\eta},
\end{equation}
where $c$ is some positive constant. Following a similar line of argument, we can show that 
\begin{align*}
2\eta& \geq D(P_{\tilde{X}_1 X_j \tilde{X}_i Y}||P_{\tilde{X}_1}\times P_{X_j}\times P_{\tilde{X}_i}\times W)  + I(\tilde{X}_1X_jY;X_1|\tilde{X}_i)\\
\qquad & \geq D(P_{X_1\tilde{X}_1X_jY}||P_{X_1}P_{\tilde{X}_1}P_{X_j}V_2)\text{ where }V_2(y|x_1,\tilde{x}_1,x_j) = \sum_{\tilde{x}_i}P_{\tilde{X}_i|X_1}(\tilde{x}_i|x_1)W(y|\tilde{x}_1,x_j,\tilde{x}_i). 
\end{align*}
Using Pinsker's inequality, it follows that 
\begin{equation}\label{eq:a62}
\sum_{x_1,\tilde{x}_1,x_j,y}\Big|P_{X_1\tilde{X}_1X_jY}(x_1,\tilde{x}_1,x_j,y)-P_{X_1}(x_1)P_{\tilde{X}_1}(\tilde{x}_1)P_{X_j}(x_j)V_2(y|x_1,\tilde{x}_1,x_j)\Big| \leq c\sqrt{3\eta}.
\end{equation}
From \eqref{eq:a52} and \eqref{eq:a62}, 
\begin{equation*}
\sum_{x_1,\tilde{x}_1,x_j,y}P_{X_1}(x_1)P_{\tilde{X}_1}(\tilde{x}_1)P_{X_j}(x_j)\Big|\tilde{V}_2(y|x_1,\tilde{x}_1,x_j)-V_2(y|x_1,\tilde{x}_1,x_j)\Big| \leq 2c\sqrt{3\eta}.
\end{equation*}
This implies that 
\begin{equation}\label{eq:1b}
\max_{x_1,\tilde{x}_1,x_j,y}\Big|\tilde{V}_2(y|x_1,\tilde{x}_1,x_j)-V_2(y|x_1,\tilde{x}_1,x_j)\Big| \leq \frac{2c\sqrt{2\eta}}{\alpha^4}.
\end{equation}
Similar to~\cite[(A.5) on page~747]{AhlswedeC99}, since $\mach$ is not $\cX_1|\cX_j$-{\em symmetrizable by}~$\cX_i$ (i.e.,~\eqref{eq:symm2} does not hold for $(i,j,k) = (i,j,1)$), we can show that for any pair for channels $P_{\tilde{X}_i|X_1}$ and $P_{X_i|\tilde{X}_1}$, there exists $\zeta_2>0$ such that
\begin{equation*}
\max_{x_1,\tilde{x}_1,x_j,y}\Big|\tilde{V}_2(y|x_1,\tilde{x}_1,x_j)-V_2(y|x_1,\tilde{x}_1,x_j)\Big| \geq \zeta_2.
\end{equation*}
This contradicts \eqref{eq:1b} if $\eta < \frac{\zeta_2^2\alpha^8}{8c^2}$.\\

 \noindent We now analyse {\bf Case 2}. 
\begin{align*}
&D(P_{X_1X_jX_i Y}||P_{X_1}\times P_{X_j}\times P_{X_i}\times W) + D(P_{\tilde{X}_1,\tilde{X}_i}||P_{\tilde{X}_1}\times P_{\tilde{X}_i}) + I(X_1  X_j Y;\tilde{X}_1 \tilde{X}_i|X_i) \stackrel{(a)}{=} \\
&\quad\sum_{x_1,x_j,x_i,y}P_{X_1X_jX_iY}(x_1,x_j,x_i,y)\log{\frac{P_{X_1X_jX_iY}(x_1,x_j,x_i,y)}{P_{X_1}(x_1)P_{X_j}(x_j)P_{X_i}(x_i)W(y|x_1,x_j,x_i)}} + \sum_{\tilde{x}_1,\tilde{x}_i}P_{\tilde{X}_1\tilde{X}_i}(\tilde{x}_1,\tilde{x}_i)\log{\frac{P_{\tilde{X}_1\tilde{X}_i}(\tilde{x}_1,\tilde{x}_i)}{P_{\tilde{X}_1}(\tilde{x}_1)P_{\tilde{X}_i}(\tilde{x}_i)}}\\
&\quad + \sum_{x_1,\tilde{x}_1,x_j,x_i,\tilde{x}_i,y}P_{X_1\tilde{X}_1X_jX_i\tilde{X}_iY}(x_1,\tilde{x}_1,x_j,x_i,\tilde{x}_i,y)\log{\frac{P_{X_1\tilde{X}_1X_j\tilde{X}_iY|X_i}(x_1,\tilde{x}_1,x_j,\tilde{x}_i,y|x_i)}{P_{\tilde{X}_1\tilde{X}_i|X_i}(\tilde{x}_1,\tilde{x}_i|x_i)P_{X_1X_jY|X_i}(x_1,x_j,y|x_i)}}\\
&\quad =\sum_{x_1,\tilde{x}_1,x_j,x_i,\tilde{x}_i,y}P_{X_1\tilde{X}_1X_jX_i\tilde{X}_iY}(x_1,\tilde{x}_1,x_j,x_i,\tilde{x}_i,y)\log{\frac{P_{X_1\tilde{X}_1X_jX_i\tilde{X}_iY}(x_1,\tilde{x}_1,x_j,x_i,\tilde{x}_i,y)}{P_{X_1}(x_1)P_{\tilde{X}_1}(\tilde{x}_1)P_{X_j}(x_j)P_{\tilde{X}_i}(\tilde{x}_i)P_{X_i|\tilde{X}_1\tilde{X}_i}(x_i|\tilde{x}_1,\tilde{x}_i)W(y|x_1,x_j,x_i)}}\\
&\quad = D(P_{X_1\tilde{X}_1X_jX_i\tilde{X}_iY}||P_{X_1}P_{\tilde{X}_1}P_{X_j}P_{\tilde{X}_i}P_{X_i|\tilde{X}_1\tilde{X}_i}W)\\
&\quad \stackrel{\text{(b)}}{\geq} D(P_{X_1\tilde{X}_1X_j\tilde{X}_iY}||P_{X_1}P_{\tilde{X}_1}P_{X_j}P_{\tilde{X}_i}\tilde{V}_3)\text{ where }\tilde{V}_3(y|x_1,\tilde{x}_1,x_j,\tilde{x}_i) = \sum_{x_i}P_{X_i|\tilde{X}_1\tilde{X}_i}(x_i|\tilde{x}_1,\tilde{x}_i)W(y|x_1,x_j,x_i),
\end{align*}
where (b) follows from the log sum inequality. From the given conditions, we know that the term on the LHS of (a) is no greater than $3\eta$. Thus, $D(P_{X_1\tilde{X}_1X_j\tilde{X}_iY}||P_{X_1}P_{\tilde{X}_1}P_{X_j}P_{\tilde{X}_i}\tilde{V}_3) \leq 3\eta$. Using Pinsker's inequality, it follows that 
\begin{equation}\label{eq:a1}
\sum_{x_1,\tilde{x}_1,x_j,\tilde{x}_i,y}\Big|P_{X_1\tilde{X}_1X_j\tilde{X}_iY}(x_1,\tilde{x}_1,x_j,\tilde{x}_i,y)-P_{X_1}(x_1)P_{\tilde{X}_1}(\tilde{x}_1)P_{X_j}(x_j)P_{\tilde{X}_i}(\tilde{x}_i)\tilde{V}_3(y|x_1,\tilde{x}_1,x_j,\tilde{x}_i)\Big| \leq c\sqrt{3\eta}
\end{equation}
for some constant $c>0$.
\noindent Following a similar line of argument, 
\begin{align*}
&D(P_{\tilde{X}_1\tilde{X}_{j}\tilde{X}_i Y}||P_{\tilde{X}_1}\times P_{\tilde{X}_{j}}\times P_{\tilde{X}_i}\times W) + D(P_{X_1X_j}||P_{X_1}\times P_{X_{j}}) + I(\tilde{X}_1 \tilde{X}_iY ;X_1  X_j|\tilde{X}_{j}) \stackrel{(a)}{=} \\
&\quad\sum_{\tilde{x}_1,\tilde{x}_{j},\tilde{x}_i,y}P_{\tilde{X}_1\tilde{X}_{j}\tilde{X}_iY}(\tilde{x}_1,\tilde{x}_{j},\tilde{x}_i,y)\log{\frac{P_{\tilde{X}_1\tilde{X}_{j}\tilde{X}_iY}(\tilde{x}_1,\tilde{x}_{j},\tilde{x}_i,y)}{P_{\tilde{X}_1}(\tilde{x}_1)P_{\tilde{X}_{j}}(\tilde{x}_{j})P_{\tilde{X}_i}(\tilde{x}_i)W(y|\tilde{x}_1,\tilde{x}_{j},\tilde{x}_i)}} + \sum_{x_1,x_j}P_{X_1X_{j}}(x_1,x_{j})\log{\frac{P_{X_1X_{j}}(x_1,x_{j})}{P_{X_1}(x_1)P_{X_{j}}(x_{j})}}\\
&\quad + \sum_{x_1,\tilde{x}_1,x_j,\tilde{x}_{j},\tilde{x}_i,y}P_{X_1\tilde{X}_1X_j\tilde{X}_{j}\tilde{X}_iY}(x_1,\tilde{x}_1,x_j,\tilde{x}_{j},\tilde{x}_i,y)\log{\frac{P_{X_1\tilde{X}_1X_j\tilde{X}_iY|\tilde{X}_{j}}(x_1,\tilde{x}_1,x_j,\tilde{x}_i,y|\tilde{x}_{j})}{P_{\tilde{X}_1\tilde{X}_iY|\tilde{X}_{j}}(\tilde{x}_1,\tilde{x}_i,y|\tilde{x}_{j})P_{X_1X_j|\tilde{X}_{j}}(x_1,x_j|\tilde{x}_{j})}}\\
&\quad =\sum_{x_1,\tilde{x}_1,x_j,\tilde{x}_{j},\tilde{x}_i,y}P_{X_1\tilde{X}_1X_j\tilde{X}_{j}\tilde{X}_iY}(x_1,\tilde{x}_1,x_j,\tilde{x}_{j},\tilde{x}_i,y)\log{\frac{P_{X_1\tilde{X}_1X_j\tilde{X}_{j}\tilde{X}_iY}(x_1,\tilde{x}_1,x_j,\tilde{x}_{j},\tilde{x}_i,y)}{P_{X_1}(x_1)P_{\tilde{X}_1}(\tilde{x}_1)P_{X_j}(x_j)P_{\tilde{X}_i}(\tilde{x}_i)P_{\tilde{X}_{j}|X_1X_j}(\tilde{x}_{j}|x_1,x_j)W(y|\tilde{x}_1,\tilde{x}_{j},\tilde{x}_i)}}\\
&\quad = D(P_{X_1\tilde{X}_1X_j\tilde{X}_{j}\tilde{X}_iY}||P_{X_1}P_{\tilde{X}_1}P_{X_j}P_{\tilde{X}_{j}P_{\tilde{X}_{j}}|X_1X_{j}}W)\\
&\quad \geq D(P_{X_1\tilde{X}_1X_j\tilde{X}_iY}||P_{X_1}P_{\tilde{X}_1}P_{X_j}P_{\tilde{X}_i}V_3)\text{ where }V_3(y|x_1,\tilde{x}_1,x_j,\tilde{x}_i) = \sum_{\tilde{x}_{j}}P_{\tilde{X}_{j}|X_1X_j}(\tilde{x}_{j}|x_1,x_j)W(y|\tilde{x}_1,\tilde{x}_{j},\tilde{x}_i).
\end{align*}

\noindent From the given conditions, the term on the left of (a) is no larger than $3\eta$. Thus, $D(P_{X_1\tilde{X}_1X_j\tilde{X}_iY}||P_{X_1}P_{\tilde{X}_1}P_{X_j}P_{\tilde{X}_i}V_3) \leq 3\eta$.

\noindent 
Using Pinsker's inequality, it follows that 
\begin{equation}\label{eq:a2}
\sum_{x_1,\tilde{x}_1,x_j,\tilde{x}_i,y}\Big|P_{X_1\tilde{X}_1X_j\tilde{X}_iY}(x_1,\tilde{x}_1,x_j,\tilde{x}_i,y)-P_{X_1}(x_1)P_{\tilde{X}_1}(\tilde{x}_1)P_{X_j}(x_j)P_{\tilde{X}_i}(\tilde{x}_i)V_3(y|x_1,\tilde{x}_1,x_j,\tilde{x}_i)\Big| \leq c\sqrt{3\eta}.
\end{equation}
From \eqref{eq:a1} and \eqref{eq:a2}, 
\begin{equation}
\sum_{x_1,\tilde{x}_1,x_j,\tilde{x}_j,y}P_{X_1}(x_1)P_{\tilde{X}_1}(\tilde{x}_1)P_{X_j}(x_j)P_{\tilde{X}_i}(\tilde{x}_i)\Big|\tilde{V}_3(y|x_1,\tilde{x}_1,x_j,\tilde{x}_i)-V_3(y|x_1,\tilde{x}_1,x_j,\tilde{x}_i)\Big| \leq 2c\sqrt{3\eta}.
\end{equation}

\noindent This implies that 
\begin{equation}\label{eq:1c}
\max_{x_1,\tilde{x}_1,x_j,\tilde{x}_j,y}\Big|\tilde{V}_3(y|x_1,\tilde{x}_1,x_j,\tilde{x}_i)-V_3(y|x_1,\tilde{x}_1,x_j,\tilde{x}_i)\Big| \leq \frac{2c\sqrt{3\eta}}{\alpha^4}.
\end{equation}
Since $\mach$ is not $\cX_1$-{\em symmetrizable by}~$\cX_j/\cX_i$ (i.e.,~\eqref{eq:symm3} does not hold for $(i,j,k) = (i,j,1)$), for any pair of channels $P_{X_i|\tilde{X}_1\tilde{X}_i}$ and $P_{\tilde{X}_{\tilde{j}}|X_1X_j}$, there exists $\zeta_3>0$, such that
\begin{equation*}
\max_{x_1,\tilde{x}_1,x_j,\tilde{x}_j,y}\Big|\tilde{V}_3(y|x_1,\tilde{x}_1,x_j,\tilde{x}_i)-V_3(y|x_1,\tilde{x}_1,x_j,\tilde{x}_i)\Big| \geq \zeta_3.
\end{equation*}
This contradicts \eqref{eq:1c} if $\eta < \frac{\zeta_3^2\alpha^8}{12c^2}$.
Let $\zeta\defineqq \min{\{\zeta_1,\zeta_2,\zeta_3\}}$, any $\eta$ satisfying $0<\eta<\frac{\zeta^2\alpha^8}{12c^2}$ ensures disjoint decoding regions.

\end{proof}

\begin{proof}[{Proof (Achievability of Theorem~\ref{thm:symmetrizability})}]
Our achievability proof uses ideas from~\cite{CsiszarN88} and is along the lines of the achievability proofs in~\cite{AhlswedeC99,PeregS19}.
Fix distributions $P_{1}\in\mathcal{P}_{X_1},\,P_{2}\in\mathcal{P}_{X_2}$ and $P_3\in\mathcal{P}_{X_3}$. For these distributions, consider the codebook given by Lemma~\ref{codebook} and the decoder given by Definition~\ref{def:decoder} for some $0<\eta<\frac{\zeta^2\alpha^8}{12c^2}$, which will be specified later. For this code, the analysis of probability of error is as below.
\noindent Assume user 3 is adversarial. For $\mathbf{y}\in\mathcal{Y}^n$, let $\phi(\mathbf{y}) = (\psi_1(\mathbf{y}),\psi_2(\mathbf{y}),\psi_3(\mathbf{y}))$. Then, the probability of error as defined earlier is given by 
\begin{align*}
P_{e,3}&\defineqq \max_{\mathbf{x_3}} \frac{1}{N_1N_2} \sum_{r\in \mathcal{M}_1,s\in\mathcal{M}_2} \Prob\Big( (\psi_1(\mathbf{y}),\psi_2(\mathbf{y})) \neq (r, s) \, \Big| X_1^n=\mathbf{x}_{1r},\, X_2^n=\mathbf{x}_{2s},\, X_3^n=\mathbf{x}_3 \Big). 
\end{align*}
Using union bound, we can upper bound $P_{e,3}$ as follows.
\begin{align}
P_{e,3}&\leq \max_{\mathbf{x_3}} \left\{\frac{1}{N_1N_2} \sum_{r\in \mathcal{M}_1,s\in\mathcal{M}_2} \Prob\Big( \psi_1(\mathbf{y}) \neq r\, \Big| X_1^n=\mathbf{x}_{1r},\, X_2^n=\mathbf{x}_{2s},\, X_3^n=\mathbf{x}_3 \Big)\right. \nonumber\\
&\qquad \qquad \left. + \frac{1}{N_1N_2} \sum_{r\in \mathcal{M}_1,s\in\mathcal{M}_2} \Prob\Big( \psi_2(\mathbf{y})	 \neq  s \,\Big| X_1^n=\mathbf{x}_{1r},\, X_2^n=\mathbf{x}_{2s},\, X_3^n=\mathbf{x}_3 \Big)\right\}.\label{eq:second_term}
\end{align}
For a fixed sequence $\mathbf{x}_3\in \mathcal{X}^n_3$ and a received vector $\mathbf{y}\in \mathcal{Y}^n$, Let 
\begin{align*}
\mathcal{Q} &\defineqq \{P_{X_1X_2X_3Y} : I(X_1;X_2)\leq\epsilon, I(X_1X_2;X_3)\leq\epsilon, D(P_{X_1X_2X_3Y}||P_{X_1}P_{X_2}P_{X_3}W)\geq\eta\},\\
A &\defineqq \{(r,s):(\mathbf{x}_{1r},\mathbf{x}_{2s}) \in \cup_{P_{X_1,X_2}:I(X_1;X_2)>\epsilon}\,T^{n}_{X_1X_2}\},\\
B &\defineqq \{(r,s):(\mathbf{x}_{1r},\mathbf{x}_{2s}, \mathbf{x}_3) \in \cup_{P_{X_1X_2X_3}:I(X_1X_2;X_3)>\epsilon}\,T^{n}_{X_1X_2}\},\\
C &\defineqq \{(r,s): (\mathbf{x}_{1r},\mathbf{x}_{2s},\mathbf{x}_3, \mathbf{y})\in \cup _{P_{X_1X_2X_3Y} \in \mathcal{Q}}\,T^{n}_{X_1X_2X_3Y} \},\\
D &\defineqq \{(r,s):(\mathbf{x}_{1r},\mathbf{x}_{2s},\mathbf{x}_3, \mathbf{y})\in \cup_{P_{X_1X_2X_3Y} :  D(P_{X_1X_2X_3Y}||P_{X_1}P_{X_2}P_{X_3}W)<\eta}\,T^{n}_{X_1X_2X_3Y} \text{ and } \psi_1(\mathbf{y})\neq r\}.
\end{align*}
We obtain an upper bound on the the first term of  \eqref{eq:second_term} in the braces in terms of these sets.
\begin{align*}
\frac{1}{N_1N_2} \sum_{r,s} \Prob\Big( \psi_1(\mathbf{y}) \neq r\, \big| X_1^n=\mathbf{x}_{1r},\, X_2^n=\mathbf{x}_{2s},\, X_3^n=\mathbf{x}_3\Big)&\leq \frac{|A|}{N_1N_2}+\frac{|B|}{N_1N_2}+ \frac{1}{N_1N_2}\sum_{(r,s)\in C}W^n(\mathbf{y}|\mathbf{x}_{1r}, \mathbf{x}_{2s}, \mathbf{x}_3)\\+ \frac{1}{N_1N_2}\sum_{(r,s)\in D}W^n(\mathbf{y}|\mathbf{x}_{1r}, \mathbf{x}_{2s}, \mathbf{x}_3).
\end{align*}
We analyse each term on the RHS seperately.
\begin{align*}
\frac{|A|}{N_1N_2} &\leq |\mathcal{P}_{X_1X_2}|\frac{|\{(r,s): (\mathbf{x}_{1r},\mathbf{x}_{2s})\in T^{n}_{X_1, X_2} \}|}{N_1N_2} \\
&\stackrel{\text{(a)}}{\leq} 2^{-n\epsilon/2} \text{ for large enough }n. 
\end{align*}
Here, (a) follows from \eqref{code_eq1} and by noting that the number of joint types is polynomial in $n$. Similarly, using \eqref{code_eq2}
\begin{align*}
\frac{|B|}{N_1N_2} &\leq |\mathcal{P}_{X_1X_2X_3}|\frac{|\{(r,s): (\mathbf{x}_{1r},\mathbf{x}_{2s}, \mathbf{x}_3)\in T^{n}_{X_1, X_2, X_3} \}|}{N_1N_2} \\
&\leq 2^{-n\epsilon/3}\text{ for large enough }n.
\end{align*}
We now analyse the third term.
\begin{align*}
\sum_{(r,s)\in C}W^n(\mathbf{y}|\mathbf{x}_{1r}, \mathbf{x}_{2s}, \mathbf{x}_3) &\leq \sum_{P_{X_1X_2X_3Y} \in \mathcal{P}_{X_1X_2X_3Y}\cap\mathcal{Q}}\,\,\sum_{\mathbf{y}\in T^{n}_{Y|X_1,X_2,X_3}(\mathbf{x}_{1r},\mathbf{x}_{2s},\mathbf{x}_3)}W^n(\mathbf{y}|\mathbf{x}_{1r}, \mathbf{x}_{2s}, \mathbf{x}_3)\\
&\leq |\mathcal{P}_{X_1X_2X_3Y}|\exp{(-nD(P_{X_1X_2X_3Y}||P_{X_1X_2X_3}W)}).
\end{align*}
Note that for any $P_{X_1X_2X_3Y}, \, D(P_{X_1X_2X_3Y}||P_{X_1X_2X_3}W) = D(P_{X_1X_2X_3Y}||P_{X_1}P_{X_2}P_{X_3}W)-I(X_1X_2;X_3)-I(X_1;X_2)$. If $P_{X_1X_2X_3Y} \in \mathcal{Q}$, then $ D(P_{X_1X_2X_3Y}||P_{X_1X_2X_3}W) > \eta -2\epsilon$ and therefore, 
\begin{align*}
\sum_{(r,s)\in C}W^n(\mathbf{y}|\mathbf{x}_{1r}, \mathbf{x}_{2s}, \mathbf{x}_3)&\leq |\mathcal{P}_{X_1X_2X_3Y}|\exp{(-n(\eta-2\epsilon))}\\
&\leq \exp{(-n(\eta-3\epsilon))}\text{ for large enough }n,\\
&\rightarrow0 \text{ as }n\rightarrow\infty \text{ when }\eta>3\epsilon.
\end{align*}
We are left to analyze the last term. For $(r,s) \in D$, error happens when one of the following holds.
\begin{description}
	 \item[(1)] There exists $u \in \mathcal{M}_1$, $u \neq r$, $v \in \mathcal{M}_2$, $v\neq s$, $\mathbf{x}'_3\in \mathcal{X}^n_3$ and random variables $X_1',X_2', X_3'$ such that $(\mathbf{x}_{1r}, \mathbf{x}_{1u},\mathbf{x}_{2s}, \mathbf{x}_{2v}, \mathbf{x}_3, \mathbf{x}'_3, \mathbf{y}) \in T^{n}_{X_1 X_1' X_2X'_2 X_3 X_3'Y}$ and  $D(P_{X_1'X'_2X'_3 Y}||P_{X'_1}\times P_{X'_2}\times P_{X'_3}\times W)< \eta$ and \linebreak $I(X_1X_2Y;X'_1X'_2|X_3) \geq \eta$.

	\item[(2)] There exists $u \in \mathcal{M}_1, u \neq r, \mathbf{x}'_3\in \mathcal{X}^n_3$ and random variables $X_1',X_3'$ such that $(\mathbf{x}_{1r},\mathbf{x}_{1u},\mathbf{x}_{2s}, \mathbf{x}_3, \mathbf{x}'_3, \mathbf{y}) \in T^{n}_{X_1 X_1'X_2 X_3 X_3' Y}$ and  $D(P_{X_1'X_2X_3' Y}||P_{X_1'}\times P_{X_2}\times P_{X_3'}\times W)< \eta$ and $I(X_1X_2Y;X_1'|X_3) \geq \eta$.

	\item[(3)] There exists $u \in \mathcal{M}_1$, $u \neq r$, $\mathbf{x}_2\in \mathcal{X}^n_2$, $t\in \mathcal{M}_3$ and random variables $X_1',X_2', X_3'$ such that $(\mathbf{x}_{1r},\mathbf{x}_{1u},\mathbf{x}_{2s}, \mathbf{x}_2, \mathbf{x}_3, \mathbf{x}_{3t}, \mathbf{y}) \in T^{n}_{X_1 X_1' X_2 X_2' X_3 X_3'Y}$ and  $D(P_{X_1'X_2'X_3' Y}||P_{X_1'}\times P_{X'_2}\times P_{X_3'}\times W)< \eta$ and $I(X_1X_2Y;X_1'X_3'|X_3) \geq \eta$.
\end{description}
Let $\mathcal{Q}_{1}$ be the set of distributions $P_{X_1X'_1X_2X'_2X_3Y}\in \mathcal{P}_{X_1X'_1X_2X'_2X_3Y}$ satisfying  $D(P_{X_1X_2X_3Y}||P_{X_1}P_{X_2}P_{X_3}W)<\eta,\, D(P_{X_1'X'_2X'_3 Y}||P_{X'_1}\times P_{X'_2}\times P_{X'_3}\times W)< \eta$ and $I(X_1X_2Y;X'_1X'_2|X_3) \geq \eta$. When condition {\bf(1)} holds,
\begin{align}
&\frac{1}{N_1N_2}\sum_{(r,s)\in D}W^n(\mathbf{y}|\mathbf{x}_{1r}, \mathbf{x}_{2s}, \mathbf{x}_3) \nonumber\\
&\quad \leq \frac{1}{N_1N_2}\sum_{r,s}\sum_{P_{X_1X'_1X_2X'_2X_3Y}\in \mathcal{Q}_{1}}\sum_{\substack{u,v:(\mathbf{x}_{1r},\mathbf{x}_{1u},\mathbf{x}_{2s},\mathbf{x}_{2v}, \mathbf{x}_3)\\ \in T^{n}_{X_1 X_1'X_2 X'_2 X_3}, \text{ for some }u\neq r,v\neq s}}\sum_{\mathbf{y}\in T^{n}_{Y|X_1X'_1X_2X'_2X_3(\mathbf{x}_{1r},\mathbf{x}_{1u},\mathbf{x}_{2s},\mathbf{x}_{2v}, \mathbf{x}_3)}}W^n(\mathbf{y}|\mathbf{x}_{1r}, \mathbf{x}_{2s}, \mathbf{x}_3).\label{eq_error}
\end{align}
From \eqref{code_eq7}, we observe that it is sufficient to evaluate the right hand term only when
\begin{equation}\label{eq:analysis1}
I(X_1X_2;X'_1 X'_2 X_3)-|R_1+R_2-I(X'_1X'_2;X_3)|^{+}\leq\epsilon.
\end{equation}
For a distribution $P_{X_1X'_1X_2X'_2X_3Y}\in \mathcal{Q}_1$ satisfying the above condition, let
\begin{align*}
P_{e,X_1X'_1X_2X'_2X_3Y} &= \sum_{\substack{u,v:(\mathbf{x}_{1r},\mathbf{x}_{1u},\mathbf{x}_{2s},\mathbf{x}_{2v}, \mathbf{x}_3)\\ \in T^{n}_{X_1 X_1'X_2 X'_2 X_3}, \text{ for some }u\neq r,v\neq s}}\sum_{\mathbf{y}\in T^{n}_{Y|X_1,X'_1,X_2,X'_2,X_3(\mathbf{x}_{1r},\mathbf{x}_{1u},\mathbf{x}_{2s},\mathbf{x}_{2v}, \mathbf{x}_3)}}W^n(\mathbf{y}|\mathbf{x}_{1r}, \mathbf{x}_{2s}, \mathbf{x}_3)\\
& \stackrel{\text{(a)}}{\leq} \exp{\{n(|R_1+R_2-I(X_1'X_2';X_1X_2X_3)|^+ +\epsilon) +nH(Y|X_1X'_1X_2X'_2X_3) - n(H(Y|X_1X_2X_3)-\epsilon)\}}\\
&=\exp{\{n\left(|R_1+R_2-I(X_1'X_2';X_1X_2X_3)|^+ - I(Y;X'_1X'_2|X_1X_2X_3)+2\epsilon\right)\}}.
\end{align*} 
Here, (a) follows from \eqref{code_eq4} and properties of types. Let $R_1+R_2<I(X'_1X'_2;X_3)$,
\begin{align*}
I(Y;X'_1X'_2|X_1X_2X_3) &=  I(X_1X_2Y;X'_1X'_2|X_3)-I(X_1X_2;X'_1X'_2|X_3) \\
&\geq I(X_1X_2Y;X'_1X'_2|X_3) -I(X_1X_2;X'_1X'_2X_3)\\
& \geq \eta-\epsilon,
\end{align*}
where the last inequality follows from the definition of $\mathcal{Q}_1$ and \eqref{eq:analysis1}. Also, $R_1+R_2<I(X'_1X'_2;X_3)\leq I(X'_1X'_2;X_1X_2X_3)$. Thus, 
\begin{align*}
P_{e,X_1X'_1X_2X'_2X_3Y} &\leq \exp\{-n(\eta - 3\epsilon)\}\\
&\rightarrow 0 \text{ as } n\rightarrow0 \text{ if }\eta>3\epsilon.
\end{align*}  
When $R_1+R_2\geq I(X'_1X'_2;X_3)$, \eqref{eq:analysis1} implies 
\begin{align*}
R_1+R_2 &> I(X'_1X'_2;X_3) + I(X_1X_2;X'_1 X'_2 X_3) - \epsilon\\
&= I(X'_1X'_2;X_3) + I(X_1X_2;X_3) + I(X'_1X'_2;X_1X_2|X_3) - \epsilon\\
&\geq I(X'_1X'_2;X_1X_2X_3) - \epsilon.
\end{align*}
This implies that $|R_1+R_2-I(X_1'X_2';X_1X_2X_3)|^+ \leq R_1+R_2-I(X_1'X_2';X_1X_2X_3) + \epsilon$. Thus, 
\begin{align*}
P_{e,X_1X'_1X_2X'_2X_3Y} &\leq \exp\{n(R_1+R_2-I(X_1'X_2';X_1X_2X_3) + \epsilon -I(Y;X'_1X'_2|X_1X_2X_3)+2\epsilon\}\\
&= \exp\{n(R_1+R_2-I(X'_1X'_2;X_1X_2X_3Y)+3\epsilon)\}\\
&\leq \exp\{n(R_1+R_2-I(X'_1X'_2;Y)+3\epsilon)\}.
\end{align*}
Since $P_{X'_1X'_2X'_3Y}$ is such that $D(P_{X_1'X'_2X'_3 Y}||P_{X'_1}\times P_{X'_2}\times P_{X'_3}\times W)< \eta$ where $\eta$ can be chosen arbitrarily small. Thus, $P_{X_1'X'_2X'_3 Y}$ is arbitrarily close to $P_{\tilde{X}_1\tilde{X}_2\tilde{X}_3\tilde{Y}} \defineqq P_{X'_1}\times P_{X'_2}\times P_{X'_3}\times W$. So, for small positive number $\gamma_1$,  $I(X_1'X_2';Y)\geq I(\tilde{X}_1\tilde{X}_2;\tilde{Y})-\gamma_1 \geq \min_{P_{X'_3}} I(\tilde{X}_1\tilde{X}_2;\tilde{Y})-\gamma_1$. Thus, if 
\begin{align}
R_1+R_2&<\min_{P_{X'_3}}I(\tilde{X}_1\tilde{X}_2;Y)-3\epsilon-\gamma_1, \label{eq:cond1}\\
\text{then, }R_1+R_2&\leq \min_{P_{X'_3}}I(X'_1X'_2;Y)-3\epsilon, \nonumber
\end{align}
and therefore, $P_{e,X_1X'_1X_2X'_2X_3Y} \rightarrow 0$ as $n\rightarrow 0$. Since, there are only polynomially many types, this implies that~\eqref{eq_error} tends to zero as $n$ tends to infinity.

To analyse condition {\bf(2)}, let $\mathcal{Q}_2$ be the set of distributions $P_{X_1X'_1X_2X_3Y}\in \mathcal{P}_{X_1X'_1X_2X_3Y}$ satisfying \linebreak $D(P_{X_1X_2X_3Y}||P_{X_1}P_{X_2}P_{X_3}W)<\eta,\, D(P_{X_1'X_2X'_3 Y}||P_{X'_1}\times P_{X_2}\times P_{X'_3}\times W)< \eta$ and $I(X_1X_2Y;X'_1|X_3) \geq \eta$.
\begin{align}
&\frac{1}{N_1N_2}\sum_{(r,s)\in D}W^n(\mathbf{y}|\mathbf{x}_{1r}, \mathbf{x}_{2s}, \mathbf{x}_3) \nonumber\\
&\quad \leq \frac{1}{N_1N_2}\sum_{r,s}\sum_{P_{X_1X'_1X_2X_3Y}\in \mathcal{Q}_2}\sum_{\substack{u:(\mathbf{x}_{1r},\mathbf{x}_{1u},\mathbf{x}_{2s}, \mathbf{x}_3)\\ \in T^{n}_{X_1 X_1'X_2 X_3}, \text{ for some }u\neq r}}\sum_{\mathbf{y}\in T^{n}_{Y|X_1X'_1X_2X_3(\mathbf{x}_{1r},\mathbf{x}_{1u},\mathbf{x}_{2s}, \mathbf{x}_3)}}W^n(\mathbf{y}|\mathbf{x}_{1r}, \mathbf{x}_{2s}, \mathbf{x}_3). \label{eq_error_2}
\end{align}
From \eqref{code_eq5}, we observe that it is sufficient to evaluate the right hand term only when
\begin{equation}\label{eq:analysis2}
I(X_1X_2;X'_1 X_3)-|R_1-I(X'_1;X_3)|^{+}\leq\epsilon.
\end{equation}
For a distribution $P_{X_1,X'_1,X_2X_3Y}\in \mathcal{Q}_2$ satisfying the above condition, let
\begin{align*}
P_{e,_{X_1,X'_1,X_2X_3Y}} &= \sum_{\substack{u:(\mathbf{x}_{1r},\mathbf{x}_{1u},\mathbf{x}_{2s},\mathbf{x}_3)\\ \in T^{n}_{X_1 X_1'X_2 X_3}, \text{ for some }u\neq r}}\sum_{\mathbf{y}\in T^{n}_{Y|X_1X'_1X_2X_3(\mathbf{x}_{1r},\mathbf{x}_{1u},\mathbf{x}_{2s}, \mathbf{x}_3)}}W^n(\mathbf{y}|\mathbf{x}_{1r}, \mathbf{x}_{2s}, \mathbf{x}_3)\\
& \leq \exp{\{n(|R_1-I(X_1';X_1X_2X_3)|^+ +\epsilon) +nH(Y|X_1X'_1X_2X_3) - n(H(Y|X_1X_2X_3)-\epsilon)\}}\\
&=\exp{\{n\left(|R_1-I(X_1';X_1X_2X_3)|^+ - I(Y;X'_1|X_1X_2X_3)+2\epsilon\right)\}}.
\end{align*} 
Let $R_1<I(X'_1;X_3)$, then
\begin{align*}
I(Y;X'_1|X_1X_2X_3) &=  I(X_1X_2Y;X'_1|X_3)-I(X_1X_2;X'_1|X_3) \\
&\geq I(X_1X_2Y;X'_1|X_3) -I(X_1X_2;X'_1X_3)\\
& \geq \eta-\epsilon,
\end{align*}
where the last inequality follows from the definition of $\mathcal{Q}_2$ and \eqref{eq:analysis2}. Also, $R_1<I(X'_1;X_3)\leq I(X'_1;X_1X_2X_3)$. Thus, 
\begin{align*}
P_{e,_{X_1,X'_1,X_2X_3Y}} &\leq \exp\{-n(\eta - 3\epsilon)\}\\
&\rightarrow 0 \text{ as } n\rightarrow0 \text{ if }\eta>3\epsilon.
\end{align*}  
When $R_1\geq I(X'_1;X_3)$, \eqref{eq:analysis2} implies 
\begin{align*}
R_1 &\geq I(X'_1;X_3) + I(X_1X_2;X'_1  X_3) - \epsilon\\
&= I(X'_1;X_3) + I(X_1X_2;X_3) + I(X_1X_2;X'_1|X_3) - \epsilon\\
&\geq I(X'_1;X_1X_2X_3) - \epsilon.
\end{align*}
This implies that $|R_1-I(X_1';X_1X_2X_3)|^+ \leq R_1-I(X_1';X_1X_2X_3) + \epsilon$. Thus, 
\begin{align*}
P_{e,_{X_1,X'_1,X_2X_3Y}} &\leq \exp\{n(R_1-I(X_1';X_1X_2X_3) + \epsilon -I(Y;X'_1|X_1X_2X_3)+2\epsilon)\}\\
&= \exp\{n(R_1-I(X'_1;X_1X_2X_3Y)+3\epsilon)\}\\
&\leq \exp\{n(R_1-I(X'_1;X_2Y)+3\epsilon)\}.
\end{align*}
Since $P_{X'_1X_2X'_3Y}$ is such that $D(P_{X_1'X_2X'_3 Y}||P_{X'_1}\times P_{X_2}\times P_{X'_3}\times W)< \eta$ where $\eta$ can be chosen arbitrarily small, $P_{X_1'X_2X'_3 Y}$ is arbitrarily close to $P_{\tilde{X}_1\tilde{X}_2\tilde{X}_3\tilde{Y}} \defineqq P_{X'_1}\times P_{X_2}\times P_{X'_3}\times W$. So, for small positive number $\gamma_2$,  $I(X_1';X_2Y) \geq I(\tilde{X}_1;\tilde{Y}|\tilde{X}_2)-\gamma_2 \geq \min_{P_{X'_3}} I(\tilde{X}_1;\tilde{Y}|\tilde{X}_2)-\gamma_2$. Thus, if 
\begin{align}
R_1&<\min_{P_{X'_3}}I(\tilde{X}_1;\tilde{Y}|\tilde{X}_2)-3\epsilon-\gamma_2, \label{eq:cond2}\\
\text{then, }R_1 &\leq \min_{P_{X'_3}}I(X'_1;Y|X_2)-3\epsilon,\nonumber
\end{align}
and therefore, $P_{e,_{X_1,X'_1,X_2X_3Y}} \rightarrow 0$ as $n\rightarrow 0$. Since, there are only polynomially many types, this implies that~\eqref{eq_error_2} tends to zero as $n$ tends to infinity.

When condition {\bf(3)} holds, let $\mathcal{Q}_{3}$ be the set of distributions $P_{X_1X'_1X_2X_3X'_3Y}\in \mathcal{P}_{X_1X'_1X_2X_3X'_3Y}$ satisfying \linebreak $D(P_{X_1X_2X_3Y}||P_{X_1}P_{X_2}P_{X_3}W)<\eta,\, D(P_{X_1'X'_2X'_3 Y}||P_{X'_1}\times P_{X'_2}\times P_{X'_3}\times W)< \eta$ and $I(X_1X_2Y;X'_1X'_3|X_3) \geq \eta$,
\begin{align}
&\frac{1}{N_1N_2}\sum_{(r,s)\in D}W^n(\mathbf{y}|\mathbf{x}_{1r}, \mathbf{x}_{2s}, \mathbf{x}_3) \nonumber\\
&\quad \leq \frac{1}{N_1N_2}\sum_{r,s}\sum_{P_{X_1,X'_1,X_2,X_3X'_3Y}\in \mathcal{Q}_3}\sum_{\substack{u,t:(\mathbf{x}_{1r},\mathbf{x}_{1u},\mathbf{x}_{2s},\mathbf{x}_{3}, \mathbf{x}_{3t})\\ \in T^{n}_{X_1 X_1'X_2 X_3 X'_3}, \text{ for some }u\neq r}}\sum_{\mathbf{y}\in T^{n}_{Y|X_1X'_1X_2X_3X'_3(\mathbf{x}_{1r},\mathbf{x}_{1u},\mathbf{x}_{2s},\mathbf{x}_{3}, \mathbf{x}_{3t})}}W^n(\mathbf{y}|\mathbf{x}_{1r}, \mathbf{x}_{2s}, \mathbf{x}_3). \label{eq_error_3}
\end{align}
From \eqref{code_eq6}, we observe that it is sufficient to evaluate the right hand term only when
\begin{equation}\label{eq:analysis3}
I(X_1X_2;X'_1 X'_3 X_3)-|R_1+R_3-I(X'_1X'_3;X_3)|^{+}\leq\epsilon.
\end{equation}
For a distribution $P_{X_1X'_1X_2X_3X'_3Y}\in \mathcal{Q}_3$ satisfying the above condition, let
\begin{align*}
P_{e,X_1X'_1X_2X_3X'_3Y} &= \sum_{\substack{u,t:(\mathbf{x}_{1r},\mathbf{x}_{1u},\mathbf{x}_{2s},\mathbf{x}_{3}, \mathbf{x}_{3t})\\ \in T^{n}_{X_1 X_1'X_2 X_3 X'_3}, \text{ for some }u\neq r}}\sum_{\mathbf{y}\in T^{n}_{Y|X_1X'_1X_2X_3X'_3(\mathbf{x}_{1r},\mathbf{x}_{1u},\mathbf{x}_{2s},\mathbf{x}_{3}, \mathbf{x}_{3t})}}W^n(\mathbf{y}|\mathbf{x}_{1r}, \mathbf{x}_{2s}, \mathbf{x}_3)\\
& \leq \exp{\{n(|R_1+R_3-I(X_1'X_3';X_1X_2X_3)|^+ +\epsilon) +nH(Y|X_1X'_1X_2X_3X'_3) - n(H(Y|X_1X_2X_3)-\epsilon)\}}\\
&=\exp{\{n\left(|R_1+R_3-I(X_1'X_3';X_1X_2X_3)|^+ - I(Y:X'_1X'_3|X_1X_2X_3)+2\epsilon\right)\}}.
\end{align*} 
Let $R_1+R_3<I(X'_1X'_3;X_3)$, then
\begin{align*}
I(Y;X'_1X'_3|X_1X_2X_3) &=  I(X_1X_2Y;X'_1X'_3|X_3)-I(X_1X_2;X'_1X'_3|X_3) \\
&\geq I(X_1X_2Y;X'_1X'_3|X_3) -I(X_1X_2;X'_1X'_3X_3)\\
& \geq \eta-\epsilon,
\end{align*}
where the last inequality follows from the definition of $\mathcal{Q}_3$ and \eqref{eq:analysis3}. Also, $R_1+R_3<I(X'_1X'_3;X_3)\leq I(X'_1X'_3;X_1X_2X_3)$. Thus, 
\begin{align*}
P_{e,X_1X'_1X_2X_3X'_3Y} &\leq \exp\{-n(\eta - 3\epsilon)\}\\
&\rightarrow 0 \text{ as } n\rightarrow0 \text{ if }\eta>3\epsilon.
\end{align*}  
When $R_1+R_3\geq I(X'_1X'_3;X_3)$, \eqref{eq:analysis3} implies 
\begin{align*}
R_1+R_3 &> I(X'_1X'_3;X_3) + I(X_1X_2;X'_1 X'_3 X_3) - \epsilon\\
&= I(X'_1X'_3;X_3) + I(X_1X_2;X_3) + I(X'_1X'_3;X_1X_2|X_3) - \epsilon\\
&\geq I(X'_1X'_3;X_1X_2X_3) - \epsilon.
\end{align*}
This implies that $|R_1+R_3-I(X_1'X_3';X_1X_2X_3)|^+ \leq R_1+R_3-I(X_1'X_3';X_1X_2X_3) + \epsilon$. Thus, 
\begin{align*}
P_{e,X_1X'_1X_2X_3X'_3Y} &\leq \exp\{n(R_1+R_3-I(X_1'X_3';X_1X_2X_3) + \epsilon -I(Y;X'_1X'_3|X_1X_2X_3)+2\epsilon\}\\
&= \exp\{n(R_1+R_3-I(X'_1X'_3;X_1X_2X_3Y)+3\epsilon)\}\\
&\leq \exp\{n(R_1+R_3-I(X'_1X'_3;Y)+3\epsilon)\}.
\end{align*}
Since $P_{X'_1X'_2X'_3Y}$ is such that $D(P_{X_1'X'_2X'_3 Y}||P_{X'_1}\times P_{X'_2}\times P_{X'_3}\times W)< \eta$ where $\eta$ can be chosen arbitrarily small. Thus, $P_{X_1'X'_2X'_3 Y}$ is arbitrarily close to $P_{\tilde{X}_1\tilde{X}_2\tilde{X}_3\tilde{Y}} \defineqq P_{X'_1}\times P_{X'_2}\times P_{X'_3}\times W$. So, for small positive number $\gamma_3$,  $I(X_1'X_3';Y)\geq I(\tilde{X}_1\tilde{X}_3;\tilde{Y})-\gamma_3 \geq \min_{P_{X'_2}} I(\tilde{X}_1\tilde{X}_3;\tilde{Y})-\gamma_3$. Thus, if 
\begin{align}
R_1+R_3&<\min_{P_{X'_2}}I(\tilde{X}_1\tilde{X}_3;Y)-3\epsilon-\gamma_3, \label{eq:cond3}\\
\text{then, }R_1+R_3&\leq \min_{P_{X'_2}}I(X'_1X'_3;Y)-3\epsilon, \nonumber
\end{align}
and therefore, $P_{e,X_1X'_1X_2X_3X'_3Y} \rightarrow 0$ as $n\rightarrow 0$. Since, there are only polynomially many types, this implies that \eqref{eq_error_3} tends to zero as $n$ tends to infinity. By combining \eqref{eq:cond1}, \eqref{eq:cond2} and \eqref{eq:cond3}, and analyzing the second term of~\eqref{eq:second_term} in a similar manner, we get the following conditions on the set of achievable rates.
\begin{align*}
R_1+R_2&<\min_{P_{X'_3}}I(\tilde{X}_1\tilde{X}_2;Y)-\delta_1\\
R_1&<\min_{P_{X'_3}}I(\tilde{X}_1;\tilde{Y}|\tilde{X}_2)-\delta_2\\
R_1+R_3&<\min_{P_{X'_2}}I(\tilde{X}_1\tilde{X}_3;Y)-\delta_3\\
R_2&<\min_{P_{X'_3}}I(\tilde{X}_2;\tilde{Y}|\tilde{X}_1)-\delta_4\\
R_2+R_3&<\min_{P_{X'_1}}I(\tilde{X}_2\tilde{X}_3;Y)-\delta_5,
\end{align*} 
where $\delta_1,\delta_2,\delta_3,\delta_4$ and $\delta_5$ are positive constants which can be made arbitrarily small by choosing small enough $\epsilon$ and $\eta$. The joint distribution of the random variables $\tilde{X}_1,\,\tilde{X}_2\,, \tilde{X}_3$ and $Y$ is given by $p(x_1)p(x_2)p(x_3)W_{Y|X_1,X_2,X_3}(y|x_1,x_2,x_3)$. Similarly, $P_{e,1}$ and $P_{e,2}$ tends to zero as $n$ tends to infinity if, for all permutations $(i,j,k)$ of $(1,2,3)$ with $k = 1$ and $k = 2$ respectively, the following holds for a small positive constant $\delta$ which can be made arbitrarily small by choosing small enough $\epsilon$ and $\eta$:
\begin{align*}
R_i&<\min_{P_{X'_k}}I(\tilde{X}_i;\tilde{Y}|\tilde{X}_j)-\delta,\\
R_i+R_j&<\min_{P_{X'_k}}I(\tilde{X}_i\tilde{X}_j;Y)-\delta.
\end{align*}
Now, via a time sharing argument, we can conclude that all points in $\textup{int}(\mathcal{R})$ are achievable if the channel is non-symmetrizable.
\end{proof}
}
\end{appendices}
\newpage
\vspace*{\fill}
\begin{center}
\huge\bfseries{ Part II\\ Authenicated Communication in the Presence of an Adversarial User\par}
\end{center}
\vspace*{\fill}
\newpage

\part{Authenicated Communication in the Presence of an Adversarial User}

\section{Introduction} 
Consider a two-user memoryless multiple access channel (MAC) where one of the
users may behave adversarially (Fig.~\ref{fig:authcomMAC}).  If both users are
non-adversarial, we require that their messages be reliably decoded.  However,
if one of the users is adversarial, we would like to ensure that either the
decoder correctly recovers the message of the other user or detects the
presence of an adversary and aborts\footnote{\label{ftnote:AVC-MAC}A more
stringent requirement could be that even in the presence of adversarial
behavior, the honest user's message be reliably decoded. It is easy to see that
this setting amounts to treating the MAC as two arbitrarily varying channels
(AVC)~\cite{BlackwellBTAMS60} where each user treats the other as the
adversary. Thus, the capacity region is the rectangular region defined by the
AVC capacities of the two channels, \emph{i.e.}, there is no trade-off between
the rates. This observation holds true under deterministic coding, stochastic
encoding, and randomized coding settings under both maximum and average
probabilities of error. The AVC capacity is zero whenever the channel is {\em
symmetrizable}~\cite{CsiszarN88} under maximum (resp., average) error
criterion for codes with stochastic encoders (resp., deterministic codes). As
we will see in Remark~\ref{rem:AVC-MAC}, under our formulation, the capacity
region is non-trivial when a condition weaker than non-symmetrizability is
satisfied.}. In other words, an adversarial user may not cause an undetected
decoding error for the other (honest) user's message. We call this the problem
of {\em authenticated communication} over MAC. We characterize the
authenticated communication capacity region of the Discrete Memoryless MAC
(DM-MAC) and the Gaussian MAC.
\begin{figure}[h]\centering
\resizebox{0.8\columnwidth}{!}{\begin{tikzpicture}[scale=0.5]
	\draw (2,0) rectangle ++(3,2) node[pos=.5]{$F_2$};
	\draw (2,2.4) rectangle ++(3,2) node[pos=.5]{$F_1$};
	\draw (10,0) rectangle ++(3,4.4) node[pos=.5]{\mch};
	\draw (16,1.2) rectangle ++(3,2) node[pos=.5]{$\phi$};
	\draw[->] (1,1) node[anchor=east]{$\msg_2$} -- ++ (1,0) ;
	\draw[->] (5,1) -- node[above] {$Y^n=F_2(\msg_2)$} ++ (5,0);
	
	\draw[->] (1,3.4) node[anchor=east]{$\msg_1$} -- ++ (1,0) ;
	\draw[->] (5,3.4) -- node[above] {$X^n=F_1(\msg_1)$} ++ (5,0);

	\draw[->] (13,2.2) -- node[above] {$Z^n$} ++ (3,0);
	\draw[->] (19,2.2) -- ++ (1,0) node[anchor=west]{\begin{array}{c}\msgh_1,\msgh_2\\\text{or}\\\bot,\bot\end{array}};
\end{tikzpicture}}
\caption{Adversarial MAC: Reliable decoding of both the messages is required
when the users are non-adversarial. An adversarial user must not cause an
undetected decoding error for an honest user. Clearly, no decoding guarantees
are given for adversarial user(s).}\label{fig:authcomMAC}
\end{figure}
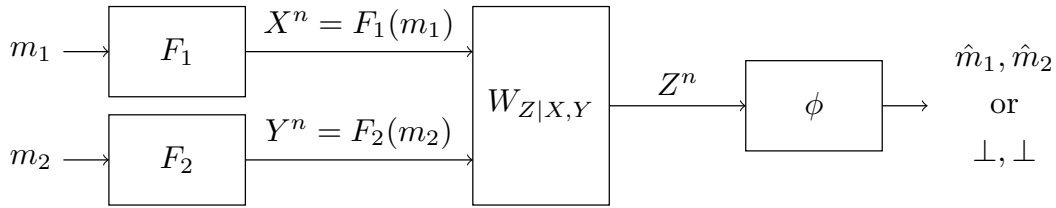
\begin{figure}[h]\centering
\resizebox{0.8\columnwidth}{!}{\begin{tikzpicture}[scale=0.5]
    \draw [dashed] (2,-2.35) rectangle ++(3.5,1.425) node[pos=.5]{Adversary};
    \draw (2,0) rectangle ++(3,2) node[pos=.5]{$\tagenc$};
    \draw (10,0) rectangle ++(3,2) node[pos=.5]{\tagchannel};
    \draw (16,0) rectangle ++(3,2) node[pos=.5]{\tagdec($\msgh,\cdot)$};
    \draw[->] (1,1) node[anchor=east]{$\msg$} -- ++ (1,0) ;
    \draw[dashed] (0.5,0.5)  -- ++ (0,-2.1) ;
    \draw[dashed,->] (0.5,-1.6) -- ++ (1.5,0);
    \draw[dashed,->] (5.5,-1.3) -- ++ (4.5,0);
    \draw[dashed,->] (5.5,-2.05) -- ++ (11.3,0);
    \draw[->] (5,1) -- node[above] {$\tU^\tl=\tagenc(\msg)$} ++ (5,0);
    \draw[->] (13,1) -- node[above] {$\rV^\tl$} ++ (3,0);
    \draw[->] (19,1) -- ++ (1,0) node[anchor=west]{\accept/\reject};
    \draw[->] (11.5,-0.5)node[anchor=north]{$\state^\tl$ or $\state_0^\tl$} -- ++ (0,0.5);
    \draw[->] (17.5,-1.3)node[anchor=north]{$\hat{m}$} -- ++ (0,1.3);
\end{tikzpicture}}
\caption{Point-to-point authentication over an AVC with ``\na''
state $\state_0$: The receiver wants to authenticate whether a message
$\hat{m}\in {\mathcal M}$ it has received (via some unauthenticated
means possibly controlled by an adversary) is indeed the
actual message $m\in {\mathcal M}$ the sender intended. The receiver must
reject ($\reject$) a $\hat{m}\neq m$ irrespective of whether there is an
adversarial attack or not, but must accept ($\accept$) when $\hat{m}=m$ if
there is no adversarial attack (\emph{i.e.}, when the state sequence of the AVC is
$s_0^t$).}\label{fig:authAVC}
\end{figure}
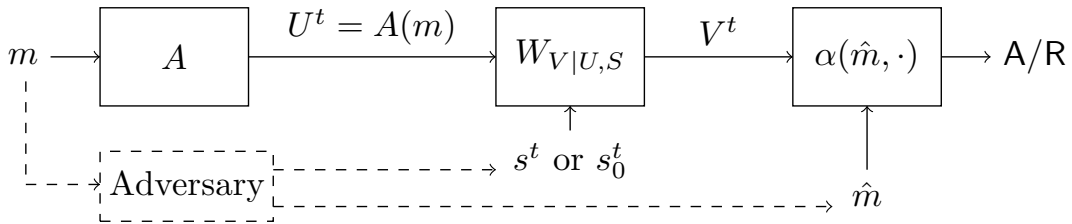

We consider the possibility of using a standard MAC code to achieve
unauthenticated communication followed by authenticating the decoded messages
using an {authentication tag} of vanishing rate.  To this end, we study the
following point-to-point authentication problem which maybe of independent
interest (Fig.~\ref{fig:authAVC}): a receiver wants to authenticate whether a
message $\hat{m}\in {\mathcal M}$ it has received (via some unauthenticated
communication channel possibly controlled by an adversary) is indeed the actual
message $m\in {\mathcal M}$ the sender intended. The authentication is
performed over a memoryless arbitrarily varying channel (AVC) from the sender
to the receiver that has a special ``no-attack'' state symbol which represents
the default behavior of the channel when there is no adversarial attack. The
adversary, when present, may choose the state sequence with the knowledge of
the actual message $m$ (and $\hat{m}$), but not the channel input
produced by the stochastic encoder using private randomness. The goal of the
authentication scheme is to ensure that the receiver does not accept a
$\hat{m}\neq m$ irrespective of whether there is an adversarial attack or not,
but does accept when $\hat{m}=m$ if there is no adversarial attack
(\emph{i.e.}, when the state of the AVC is set as the ``no-attack'' symbol).
We show that when the AVC can be used for authentication at all, the number of
messages grows double-exponentially in the tag blocklength~$t$, \emph{i.e.},
$\log\log|{\mathcal M}|=tR$ for $R>0$. We characterize the rate of
double-exponential growth of the size of messages which can be authenticated
and show that this is \longonly{in fact} equal to the Shannon capacity of the
nominal channel for non-overwritable Discrete Memoryless AVCs (DM-AVCs) and for Gaussian memoryless AVCs. Our results are reminiscent of identification capacity of
Ahlswede and Dueck~\cite{AhlswedeDIT89}\longonly{ and, as will become clear,
the connection is not coincidental} (see
Remark~\ref{rem:IDconnection}\longonly{ and proofs of Theorem~\ref{thm: auth
capacity} and Lemma~\ref{lem:auth gaussian positive}}). 

Returning to the MAC problem with a possibly adversarial user, we show that the
following 3-phase scheme is optimal (Fig.~\ref{fig:MACscheme}). In the first
phase we use a standard MAC code to communicate the messages, albeit without
any guarantee that the decoded messages are indeed what the senders intended.
In the second phase, user~1 authenticates its message to the decoder while
treating user~2 as a potential adversary; simultaneously, user~2 is required to
send a certain ``no-attack'' symbol\footnote{As will be seen, to support any
positive rate of communication to user~1 in the MAC, there must exist such a
``no-attack'' symbol in user~2's alphabet such that authentication is possible
over the AVC from user~1 to the receiver resulting from the MAC by treating the
alphabet of user~2 as the set of states.} from its channel input alphabet
throughout this phase with deviation from this amounting to adversarial
behavior.  In phase~3, the roles of the users are reversed.  The
double-exponential growth of the number of  messages that can be authenticated
means that the authentication phase  does not affect the rate of communication;
specifically, we may choose the authentication tag blocklength $t= {\mathcal O}(\log n)$,
where $n$ is the blocklength of the entire transmission. Thus, whenever
positive rates of communication can be supported to both users (either of whom
may be adversarial), we can, in fact, support all points in the non-adversarial
MAC capacity-region itself. When a positive rate can be supported to only one
of the users, we can support the largest rate of that user in the
non-adversarial MAC capacity-region.  

\shortonly{The work that is closest to ours is by Kosut and
Kliewer~\cite{KosutKITW18} which studies the same AVC model with a
``no-attack'' symbol as in the point-to-point authentication problem above.
However, they study the problem of authenticated {\em communication} capacity,
\emph{i.e.}, unlike the AVC setup above, where the receiver's goal is to use
the channel only to authenticate a message $\hat{m}$ it received through other
means, in~\cite{KosutKITW18}, the channel is used for the entire authenticated
communication. The goal is for the receiver to reliably decode the correct
message when there is no attack and, when under attack (\emph{i.e.}, the state
sequence is not all ``no-attack''), to either declare the presence of an
adversary or output the message transmitter intended. Our authentication
results for the finite-alphabet DM-AVC case build on the authenticated
communication capacity result from~\cite{KosutKITW18}. As we discuss in
Remark~\ref{rem:Kosut-Kliewer comparison}, our result also, arguably, explains
the nature of the result there.}

\shortonly{We also derive analogous results for the additive Gaussian memoryless AVCs
using an {\em independent} approach to first show that the (non-adversarial)
Gaussian channel capacity is achievable even in the authenticated
communication setup. This allows us to show doubly-exponential growth of
authenticable messages and thereby show that the authenticated communication
capacity region of the Gaussian MAC equals the unauthenticated capacity
region.}

\longonly{The work that is closest to ours is by Kosut and
Kliewer~\cite{KosutKITW18} which studies the same AVC model with a
``no-attack'' symbol as in the point-to-point authentication problem above.
However, they study the problem of authenticated {\em communication} capacity,
\emph{i.e.}, unlike the AVC setup above, where the receiver's goal is to use
the channel only to authenticate a message $\hat{m}$ it received through other
means, in~\cite{KosutKITW18}, the \longonly{channel is used for the entire
authenticated communication. The }goal is for the receiver to reliably decode
the correct message when there is no attack and, when under attack
(\emph{i.e.}, the state sequence is not all ``no-attack''), to either declare
the presence of an adversary or output the message transmitter intended.
Clearly, the authenticated communication capacity cannot exceed the Shannon
capacity of the nominal channel (the channel where state is set to the
``no-attack'' symbol). The main result of~\cite{KosutKITW18} is that, for
finite alphabet-finite state channels, the Shannon capacity can be achieved
unless the AVC meets a condition called {\em overwritability} in which case
the authenticated communication capacity is zero.}

\longonly{Our results for finite alphabet-finite state AVCs build
on~\cite{KosutKITW18}; using capacity achieving authenticated communication
codes from~\cite{KosutKITW18}, we first conclude that when the AVC is not
overwritable, authentication of a double-exponentially growing number of
messages is feasible. We also derive analogous results for additive  Gaussian
AVCs using an independent approach to first show that the (non-adversarial)
Gaussian channel capacity is achievable even in the authenticated
communication setup. \longonly{As in the finite alphabet-finite state setting,
this allows us to show a double-exponential growth of the number of messages
that can be authenticated. }These results show that, along the lines of our
MAC scheme, one way to achieve authenticated communication at rates
approaching Shannon capacity of the nominal channel is to first perform
unauthenticated communication using a good code for the nominal channel, and
follow this up with a short authentication phase which does not affect the
overall rate\footnote{Note that the result in~\cite{KosutKITW18} is for a
deterministic encoder/decoder under average error criterion. The description
here is for stochastic encoder under maximum error as our authentication
encoder is necessarily stochastic.  However, standard techniques show that
capacities of authenticated communication under these two settings must be the
same; also see Remark~\ref{rem:det+avg=stoch+max}.}. This gives an alternative
architecture (but not an alternative proof) for achieving the result
in~\cite{KosutKITW18} and, arguably, explains the nature of that result --
whenever authenticated communication is possible at all, since authentication
can be performed with a double-exponentially growing number of messages, the
rate of communication is simply governed by the largest rate of
unauthenticated communication possible, \emph{i.e.}, the Shannon capacity of
the nominal channel.
}

There are several strands of other related work. Communication over MACs with
strategic users modeled as a cooperative game has been studied
in~\cite{LaADIMACS04}. As noted in~\cite{BassalygoBPPI96} in a different
context, the authentication problem is closely related to the identification
problem. The latter has been the subject of extensive study  for
non-AVCs~\cite{AhlswedeDIT89,BurnashevIT00} as well as AVCs~\cite{BocheD:19}.
Message authentication codes where the users have pre-shared keys and
communicate over noiseless channels have been extensively
studied~\cite{SimmonsCRYPTO84,MaurerIT00}. Message authentication over noisy
channels has also been considered~\cite{LaiEPIT09,TuLIT18}. These works differ
from our point-to-point authentication problem: (a) the channel model is
different - the decoder either hears from the adversary or from the sender,
but not from both as in an AVC, however, the adversary may wiretap the
transmitted signal, and (b) they also consider impersonation attacks where the
adversary aims to make the receiver accept a spurious message when the sender
may not have a legitimate message to send, something which does not arise in
our  problem. Byzantine attack on the nodes and edges of a network has been
studied under omniscient and weaker adversarial models in~\cite{KosutTTIT14}
and~\cite{JaggiLKHKM07}, respectively. Unlike the authenticated communication
model in the current paper, the decoders there are required to decode the
source message under byzantine attack. A Gaussian two-hop network with an
eavesdropping and byzantine adversarial relay has been considered
in~\cite{HeYIT13}, where the requirement is decoding with message secrecy and
byzantine attack detection.
 
 \noindent\ \;The main contributions of our work are:\vspace{-0.3em}
\begin{enumerate}
\item Authenticated communication capacity region of DM-MAC under adversarial users;
\item Authentication capacity of point-to-point DM-AVCs with a ``no-attack'' symbol;
\item Explains (arguably) the authenticated communication capacity result of Kosut \& Kliewer~\cite{KosutKITW18} for DM-AVCs with a ``no-attack'' symbol;
\item Analogous results for Gaussian channels.
\end{enumerate}

\section{Problem Setup and Main Results}
\subsection{Adversarial MAC}
\subsubsection{System model}\label{section:MAC_model}
Consider the two-user MAC setup shown in Fig.~\ref{fig:authcomMAC}. The memoryless channel \mch has input alphabets $\mathcal{X,Y}$, and output alphabet ${\mathcal Z}$. We consider stochastic encoders which use independent private randomness.

\shortonly{An $(\nummsg_1,\nummsg_2,n)$  {\em authenticated communication code} for the multiple access channel \mch consists of: \vspace{-0.25em}
\begin{enumerate}[label=(\roman*)]
\item Two message sets, $\mathcal{M}_i = \{1,\ldots,\nummsg_i\}$, $i=1,2$,
\item Two stochastic encoders, $F_{1}:\mathcal{M}_1\rightarrow \mathcal{X}^n$ and $F_{2}:\mathcal{M}_2\rightarrow \mathcal{Y}^n$, and
\item A deterministic decoder, $\phi:\mathcal{Z}^n\rightarrow(\mathcal{M}_1\times\mathcal{M}_2)\cup\{(\bot,\bot)\}.$ 
\end{enumerate}}

\longonly{\begin{defn}[Authenticated communication code]
An $(\nummsg_1,\nummsg_2,n)$  {\em code} for the \mac \mch consists of the following: 
\begin{enumerate}[label=(\roman*)]
\item Two message sets, $\mathcal{M}_i = \{1,\ldots,\nummsg_i\}$, $i=1,2$,
\item Two stochastic encoders, $F_{1}:\mathcal{M}_1\rightarrow \mathcal{X}^n$ and $F_{2}:\mathcal{M}_2\rightarrow\mathcal{Y}^n$, and
\item A deterministic decoder, $\phi:\mathcal{Z}^n\rightarrow(\mathcal{M}_1\times\mathcal{M}_2)\cup\{(\bot,\bot)\}.$ 
\end{enumerate}
\end{defn}}
We define the probability of error in terms of the maximum error probabilities under no-attack and attack conditions as follows. Let $(\hat{M}_1,\hat{M}_2) = \phi(Z^n).$ 
\begin{align*}
&P_{e,\na}\hspace{-0.25em} \defineqq\ \  \smashoperator[lr]{\max_{\vecm\in\cM_1\times\cM_2}}\ \ \, \Prob_{F_1,F_2}\left((\hat{M}_1,\hat{M}_2)\neq\vecm|(M_1,M_2) = \vecm\right)\\
&P_{e,\malone} \defineqq  {\max_{m_2,x^n}}\Prob_{F_2}\left(\hat{M}_2\notin\left\{m_2,\bot\right\}\left|M_2 = m_2,X^n =x^n\right.\right),\\
&P_{e,\maltwo} \defineqq \max_{m_1,y^n}\Prob_{F_1}\left(\hat{M}_1\notin\left\{m_1,\bot\right\}\left|M_1 = m_1,Y^n =y^n\right.\right),
\end{align*}
where the probabilities are over the channel and the independent, private randomness of the stochastic encoders indicated by the subscripts.
The probability of error is defined as
\begin{align*}
P_{e}(F_1,F_2,\phi)=\max{\{P_{e,\na},P_{e,\malone},P_{e,\maltwo} \}}.	
\end{align*}
\longonly{\begin{defn}[Achievable rate pair, Authenticated communication capacity region]}
We say that $(R_1,R_2)$ is an {\em achievable authenticated communication rate pair} if there exists a sequence of $(\lfloor2^{nR_1}\rfloor,\lfloor2^{nR_2}\rfloor,n)$ codes $\{F_1^{(n)},F_2^{(n)},\phi^{(n)}\}_{n=1}^\infty$ such that $\lim_{n\rightarrow\infty}P_{e}(F_1^{(n)},F_2^{(n)},\phi^{(n)})\rightarrow0.$ The {\em authenticated communication capacity region} $\mathcal{R}_{\authcomm}$ is the closure of the set of all achievable rate pairs.
\longonly{\end{defn}}
\subsubsection{Main result}
Suppose $\mathcal{X,Y,Z}$ are finite. We will use \mac to refer to this channel. \longonly{We need the following definition to state our main result.
\begin{defn}[Overwriting in \mac]\label{non_overwrite}}
 We say that user~1 can {\em overwrite} user~2 if there is a distribution $P_{X'|X,Y}$ such that 
\begin{equation}\label{eq1}
\sum_{x'\in \mathcal{X}}P_{X'|X,Y}(x'|x,y)W(z|x',y') = W(z|x,y) 
\end{equation}
for all $y,y'\in \mathcal{Y}, x\in \mathcal{X}$ and $z\in \mathcal{Z}$.
\longonly{\end{defn}} In other words, user~1 can overwrite user~2 if it (user~1) can replace any symbol $y'$ sent by user~2 with any symbol $y$ of its choosing while pretending to send any symbol $x$, again, of its choosing. In doing this, user~1 does not have knowledge of the symbol $y'$ transmitted by user~2. Overwritability of user~1 by user~2 can be defined in a similar manner. 

Let $\mathcal{R}_{\nonadv}$ denote the non-adversarial capacity region of the DM-MAC with stochastic encoders under the maximum probability of error criterion. By non-adversarial, we mean that both users are guaranteed to act honestly. Note that this region is same as the usual non-adversarial capacity region of the DM-MAC under the average probability of error criterion using deterministic codes~\cite{CaiENTROPY14}. The following theorem characterizes the authenticated communication capacity of the \mac.\vspace{-0.25em}
\begin{thm}\label{thm:authMAC}
If the \mac is such that neither user can overwrite the other, then $\mathcal{R}_{\authcomm} = \mathcal{R}_{\nonadv}$. Otherwise, if user 1 can overwrite user 2, but user 2 cannot overwrite user 1, $\mathcal{R}_{\authcomm} = \{(R,0)|(R,0)\in \mathcal{R}_{\nonadv}\}$. If user 2 can overwrite user 1, but user 1 cannot overwrite user 2, $\mathcal{R}_{\authcomm} = \{(0,R)|(0,R)\in \mathcal{R}_{\nonadv}\}$. When both users can overwrite each other, $\mathcal{R}_{\authcomm} = \{(0,0)\}.$
\end{thm} \vspace{-0.25em}\vspace{-0.25em}
\begin{remark}\label{rem:det+avg=stoch+max}
We show in Appendix~\ref{App:stocDet} that the capacity region of the \mac for deterministic codes under average probability of error criterion is the same as the capacity region we defined above (i.e., for codes with stochastic encoders under maximum probability of error criterion).\end{remark}\vspace{-0.25em}\vspace{-0.25em}
\begin{remark}\label{rem:AVC-MAC}
As mentioned in Footnote~\ref{ftnote:AVC-MAC}, there are channels where the authenticated communication capacity region of the \mac is non-trivial, but the capacity region of the DM-MAC when the decoder is required to decode the messages correctly with high probability even when under adversarial attack is zero. For example, consider the so-called binary erasure MAC, the inputs are binary and the output $Z=X+Y$ is ternary.  It is easy to see that the point-to-point channel from one user to the receiver treating the other user as an adversary is symmetrizable~\cite{CsiszarN88} and thus the capacity of the point-to-point AVC is zero (for codes with stochastic encoders under maximum error criterion). However, neither user can overwrite the other since, if one user sends 0, the other user cannot make it look like 1 while pretending to send 1.
Thus, by Theorem~\ref{thm:authMAC}, the authenticated communication capacity region of this channel is its non-adversarial MAC capacity region which is non-trivial, e.g., $(3/4,3/4)\in \mathcal{R}_{\authcomm}$.
\end{remark}\vspace{-0.5em}

\subsection{Authentication over an AVC}\label{subsec:authAVC}
\subsubsection{System model}
Consider the setup shown in Fig.~\ref{fig:authAVC}. The arbitrarily varying channel \tagchannel\ has input alphabet $\cU$, output alphabet $\cV$ and state stace $\cS$. The channel has a nominal or \emph{\na state},  $\state_0$, which is the default channel state when there is no adversarial attack.  \longonly{The transmitter has access to the true message $\msg$ and transmits a tag $\tU^\tl$ of length $\tl$. The receiver has a candidate message $\msgh$ and observes the received tag $\rV^\tl$. Here $\msgh$ may be different from $\msg$. The adversary knows the true message $\msg$ and the candidate message $\msgh$, but not the actual transmitted tag $\tU^\tl$. We say that the system is under attack if $\state^\tl$ is different from $\state_0^\tl$.} The goal for the transmitter and the receiver is to be able to \emph{authenticate} the message $\msgh$ by accepting it when the system is not under attack provided that $\msgh$ equals $\msg$,  and by rejecting it if $\msgh$ is different from $\msg$ (regardless of whether there is an attack). 

\longonly{
We formalize an authentication code for this setup as following.
\begin{defn}[Authentication Code]\label{def:auth}
An $\left(\nummsg,\tl\right)$-{\em authentication code} $(\tagenc,\tagdec)$ for an AVC \tagchannel\  with state space $\cS$ and the \na state $\state_0$ consists of
\begin{enumerate}
\item[(i)] a message set $\msgset = \left\{1,2, \ldots, \nummsg\right\}$,
\item[(ii)] a stochastic encoder $\tagenc:\msgset \rightarrow \cU^{t}$, and
\item[(iii)] a deterministic decoder $\tagdec:\msgset \times \cV^t \rightarrow \left\{\accept, \reject\right\}$.

\end{enumerate}
\end{defn}}
\shortonly{
An $\left(\nummsg,\tl\right)$-{\em authentication code} $(\tagenc,\tagdec)$ for an AVC \tagchannel\  with state space $\cS$ and the \na state $\state_0$ consists of:\vspace{-0.25em}\begin{enumerate}
\item[(i)] a message set $\msgset = \left\{1,2, \ldots, \nummsg\right\}$,
\item[(ii)] a stochastic encoder $\tagenc:\msgset \rightarrow \cU^{t}$, and
\item[(iii)] a deterministic decoder $\tagdec:\msgset \times \cV^t \rightarrow \left\{\accept, \reject\right\}$.

\end{enumerate}\vspace{-0.5em}
}The authentication error probability $\peauth(\tagenc,\tagdec)$ for an authentication code $(\tagenc,\tagdec)$ is the maximum of  $\pfa$, the probability of false rejection under no attack, and $\pwa$, the probability of accepting a wrong message, \emph{i.e.},
\begin{align}
&\pfa\defineqq \max_{\msg\in\msgset}\Prob\left(\tagdec\left(m,\vecrV\right) = \reject| \State^\tl=\state_0^\tl, M = m\right),\label{eq:pfa}\\
&\pwa\defineqq \max_{\substack{\msg,\msgh\in\msgset:\msgh\neq\msg\\\state^\tl\in\cS^\tl}}
\Prob\left(\tagdec\left(\msgh,\vecrV\right) = \accept| \State^\tl=\state^\tl, M = m\right),\label{eq:pwa}\\
&\peauth(\tagenc,\tagdec)\defineqq \max\set{\pfa,\pwa},\label{eq:peauth}
\end{align}
where the probabilities are over the channel and the randomness of the stochastic encoder. Also, $s^t_0$ is a $t$-length vector consisting only of the \na symbol $s_0\in \mathcal{S}$.
\longonly{\begin{defn}[Achievable rates, Authentication Capacity]\label{def:autherr}
We say that $R$ is an {\em achievable authentication rate} for the AVC $\tagchannel$ with the \na state $\state_0$ if there exists a sequence of  codes $\set{(\tagenc^{(\tl)},\tagdec^{(\tl)})}_{\tl=1}^\infty$ such that for each $\tl$, $(\tagenc^{(\tl)},\tagdec^{(\tl)})$ is a $(\lfloor 2^{2^{\tl R}}\rfloor,\tl)$-authentication code with the authentication error probabilities satisfying $\lim_{\tl\to\infty} \peauth(\tagenc^{(\tl)},\tagdec^{(\tl)})=0$. The {\em authentication capacity} $C_{\auth}$ of the AVC \tagchannel\ with the \na state $\state_0$ is the supremum over all achievable authentication rates.
\end{defn}}
\shortonly{
We say that $R$ is an {\em achievable authentication rate} for the AVC $\tagchannel$ with the \na state $\state_0$ if there exists a sequence of  codes $\set{(\tagenc^{(\tl)},\tagdec^{(\tl)})}_{\tl=1}^\infty$ such that for each $\tl$, $(\tagenc^{(\tl)},\tagdec^{(\tl)})$ is a $(\lfloor 2^{2^{\tl R}}\rfloor,\tl)$-authentication code with the authentication error probabilities satisfying $\lim_{\tl\to\infty} \peauth(\tagenc^{(\tl)},\tagdec^{(\tl)})=0$. The {\em authentication capacity} $C_{\auth}$ of the AVC \tagchannel\ with the \na state $\state_0$ is the supremum over all achievable authentication rates.
}
\subsubsection{Main result} We give the exact authentication capacity for AVCs
with finite input, output, and state alphabets in the following theorem. In
particular, similar  to the identification problem~\cite{AhlswedeDIT89}, we
show that the number of messages that can be authenticated with vanishingly
small error probability scales \emph{doubly exponentially} with the tag
blocklength $t$ provided that the AVC is not overwritable.  Let
$C(\state_0)\defineqq\max_{P_\tU}I(U;V|S=\state_0)$ denote the Shannon
capacity of the nominal channel (i.e., the AVC when the state is fixed to be
$\state_0$). Following~\cite{KosutKITW18}, we say that $\tagchannel$ is {\em overwritable} if there exists a conditional distribution $P_{\State|\tU'}$ such that 
\begin{equation}\sum_{\state\in\cS}P_{\State|\tU'}(\state|\tu')\tagchannel(\rv|\tu,\state)=\tagchannel(\rv|\tu',\state_0) \label{eq:overwritable}
\end{equation}
for all $(\tu,\tu',\rv)\in\cU\times\cU\times\cV$. If no such $P_{\State|\tU'}$ exists, we say that \tagchannel\ is not overwritable.\vspace{-0.25em}
\begin{thm}\label{thm: auth capacity}
For DM-AVC $\tagchannel$ with a \na state $s_0\in\cS$, the authentication capacity \shortonly{$C_{\auth}$ equals $C(\state_0)$ if the channel is not overwritable and it is $0$ otherwise.} \longonly{is given by 
\begin{equation}
C_{\auth}=\begin{cases} C(\state_0) & \mbox{if \tagchannel\ is not overwritable},\\ 0 &\mbox{otherwise.}\end{cases}	
\end{equation}}\vspace{-0.25em}
\end{thm}\vspace{-0.25em}\vspace{-0.25em}
\begin{remark} \label{rem:Kosut-Kliewer comparison}On the one hand, our proof
of the above theorem builds on the result of Kosut and
Kliewer~\cite{KosutKITW18} on authenticated communication capacity of
adversarial channels\footnote{Note that we use the term ``authenticated
communication'' to refer to the problem in~\cite{KosutKITW18} which is
different from the ``authentication'' problem of this section.}. On the other
hand, the above theorem gives us insights into the nature of the result
of~\cite{KosutKITW18}. There, the goal is to communicate (and not just
authenticate) a message $m$ reliably when there is no attack and, when under
attack, either abort (i.e., output $\bot$) or output the intended message $m$,
i.e., the decoder does not already possess a message $\hat{m}$ which it simply
wants to authenticate. Clearly, for this authenticated {\em communication}
problem, the rate (of single exponential growth of the message set) cannot
exceed the capacity $C(s_0)$ of the nominal channel. The result
of~\cite{KosutKITW18} is that whenever the DM-AVC is not overwritable, the
authenticated communication capacity is indeed $C(s_0)$, and it is zero
otherwise. The theorem above suggests, along the lines of our MAC scheme
(Fig.~\ref{fig:MACscheme}), an alternative architecture (but not an alternative
proof) to achieve this result -- first perform unauthenticated communication
using a good code for the nominal channel, and follow this up with a short
authentication phase which does not affect the overall rate\footnote{Note that
the result of~\cite{KosutKITW18} is for a deterministic encoder/decoder under
average error criterion. The description here is for stochastic encoder under
maximum error as our authentication encoder is necessarily stochastic.
However, the observation in Remark~\ref{rem:det+avg=stoch+max} applies in this
case too.}.  Arguably, this explains the nature of the result
of~\cite{KosutKITW18} -- whenever authenticated communication is possible at
all, since authentication can be performed for a set of doubly-exponentially
growing number of messages, the rate of communication is simply governed by the
largest rate of unauthenticated communication possible, i.e., the Shannon
capacity of the nominal channel.
\end{remark}\vspace{-0.25em}\vspace{-0.25em}
\begin{remark}\label{rem:IDconnection} The authentication problem studied here
borrows features from the identification setup over
non-AVCs~\cite{AhlswedeDIT89} as well as that over AVCs~\cite{BocheD:19}. On
the one hand, as also in the identification setup over the non-AVC
$W_{V|U,S}(.|.,s_0)$, we desire that the true message be accepted under no
attack. On the other, as in the identification setup over the AVC
$\tagchannel$, we desire that a wrong candidate message be rejected regardless
of the channel state vector. However, unlike the identification setup, when
the system is under attack, we do not have any requirement on the decoder
behavior if the candidate message is the same as the true message. Our proofs
for authentication further emphasize this connection as we build on ideas
from~\cite{AhlswedeDIT89} in our achievabilities as well as converses.
\end{remark}\vspace{-0.25em}

\subsection{The Gaussian case}

We derive analogues to Theorem~\ref{thm:authMAC}, Theorem~\ref{thm: auth capacity} and~\cite[Theorem~2]{KosutKITW18} for additive Gaussian noise channels. We first consider the adversarial MAC setting over an additive Gaussian MAC modeled as $Z=X+Y+W$, with $X,Y,Z\in\bbR$ and $W\sim\cN(0,\sigma^2)$. We assume that each legitimate user is subject to their power constraints, while an \emph{adversarial} user has no such power constraints and can output {\em any} real value. \vspace{-0.25em}
\begin{thm} \label{thm:auth comm gaussian mac}  For an additive Gaussian adversarial MAC with average input power constraints $\rho_1$ and $\rho_2$ on the first and second users and noise variance $\sigma^2$, $\cR_{\mathrm{auth-comm}}$ equals the capacity region of the corresponding non-adversarial MAC.
\end{thm}\vspace{-0.25em}

Our result builds on an authentication code for an additive Gaussian AVC, modeled as $V=U+S+W$, with $\tU,\rV,\State\in\bbR$ and $W\sim\cN(0,\sigma^2)$. \longonly{In the following theorem, we consider the authentication problem over an additive Gaussian AVC whose channel input and output are real valued and are related as   $V=U+S+W$, where $S\in\bbR$ is adversarially chosen (and equals zero for the nominal channel), and $W$ is zero mean Gaussian noise with variance $\sigma^2$. }As in Theorem~\ref{thm: auth capacity}, we see that the maximum number of messages that can be authenticated grows doubly exponentially with the blocklength, the exponent being equal to the capacity of the nominal channel.\vspace{-0.25em}
\begin{thm}\label{thm:auth gaussian}
For an additive Gaussian AVC with average input power constraint $\rho$ and noise variance $\sigma^2$, the authentication capacity equals the channel capacity of the nominal channel, \emph{i.e.}, $C_{\mathrm{auth}}=	\frac{1}{2}\log\left(1+\frac{\rho}{\sigma^2}\right)$.
\end{thm}\vspace{-0.25em}
As in the case of Theorem~\ref{thm: auth capacity}, we prove Theorem~\ref{thm:auth gaussian} using an authenticated communication code for the additive Gaussian AVC. The following theorem is an analogue of~\cite[Theorem~2]{KosutKITW18} for the additive Gaussian AVC setting. Our proof technique is significantly different from that of~\cite[Theorem~2]{KosutKITW18}, which relies on the finiteness of the alphabets and the state-set.\vspace{-0.25em}
\begin{thm} \label{thm:auth comm gaussian}  For an additive Gaussian AVC with average input power constraint $\rho$ and noise variance $\sigma^2$, the authenticated communication capacity equals the channel capacity of the nominal channel, \emph{i.e.}, 
$C_{\mathrm{auth-comm}}=	\frac{1}{2}\log\left(1+\frac{\rho}{\sigma^2}\right)$.\end{thm}\vspace{-0.25em}

\section{Proof Sketches}
\subsection{Proof sketch of Theorem~\ref{thm: auth capacity}}
The detailed proof is given in Appendix~\ref{App:AuthAVC}.\\
The achievability scheme involves transmitting an identification code for the identity channel~\cite[Proposition~1]{AhlswedeDIT89} using an authenticated communication code. We use the observation that given a deterministic authenticated communication code from~\cite{KosutKITW18} that has a small average probability of error, we can construct a stochastic authenticated communication code with a slightly lower rate and with a small maximal probability of error. When the channel is non-overwritable, the converse follows from the identification capacity converse for the channel $W_{\rV|\tU,\State}(.|.,\state_0)$. When the channel is overwritable, for any message $m'$, the adversary can mount an attack under which, irrespective of the actual message $m$, the output distribution is the same as for $m'$ under the nominal channel. Using this we argue that $\peauth \geq 1/2$.\vspace{-0.5em}
\subsection{Proof sketch of Theorem~\ref{thm:authMAC}}
The detailed proof is given is Appendix~\ref{App:authMAC}.\\
Suppose neither user can overwrite the other. Our achievability scheme consists of three phases: a DM-MAC code (of blocklength $n$, say) with {\em stochastic} encoders for non-adversarial users with maximum probability of error vanishing as $n\rightarrow\infty$~\cite{CaiENTROPY14}; this is followed by two authentication phases, one for each user.
While authenticating user~2, we consider the point-to-point channel from user~2 to the receiver and treat user~1 as possibly adversarial. For this AVC, user~1's alphabet corresponds to the set of adversarial states. The following lemma shows that if user~1 cannot overwrite user~2 (i.e., \eqref{eq1} does not hold), this AVC has a ``\na'' state ($x_0$ in the lemma) such that it is not overwritable (i.e., \eqref{eq:overwritable} does not hold). This will allow us to use the authentication codes from Section~\ref{subsec:authAVC} with tag blocklength $t={\mathcal O}(\log{n})$.\vspace{-0.25em}
\begin{lemma}\label{lemma:1}
When user~1 cannot overwrite user~2, there exists $x_0\in \mathcal{X}$ such that for all $P_{X'|Y}$, there exists $(y,y',z)\in \mathcal{Y}\times\mathcal{Y}\times\mathcal{Z}$, that satisfies 
\begin{equation}\label{eq:overwrite lemma}
\sum_{x'\in\mathcal{X}}P_{X'|Y}(x'|y)W(z|x',y')\neq W(z|x_0,y).
\end{equation}
\end{lemma}\vspace{-0.25em}
The proof follows by noting that if \eqref{eq:overwrite lemma} does not hold, then, for every $x\in\mathcal{X}$, we can find an attack distribution $P_{X'|X=x,Y}$ which satisfies \eqref{eq1}.  
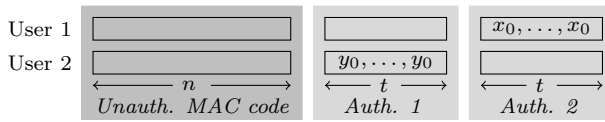
\begin{figure}\centering
\begin{tikzpicture}[xscale=0.5,yscale=0.3]
    \node[] at (-1.8,5.5) {\footnotesize User 2} ;
    \node[] at (-1.8,7) {\footnotesize User 1} ;
    \draw [fill = black!25!white, black!25!white] (-0.7,3) rectangle ++(5.8,5);
    \draw [fill = black!15!white, black!15!white] (5.4,3) rectangle ++(3.8,5);
    \draw [fill = black!15!white, black!15!white] (9.5,3) rectangle ++(3.8,5);
    \draw (-0.4,5) rectangle ++(5.2,1) node[pos=.5]{$$};
    \draw[<-] (-0.4,4.5) -- ++ (2.1,0) node[anchor=west] {\footnotesize$n$};
    \draw[<-] (4.8,4.5) -- ++ (-2.1,0) ;

    \draw (-0.4,6.5) rectangle ++(5.2,1) node[pos=.5]{$$};
    \draw (5.7,5) rectangle ++(3.2,1) node[pos=.5]{\footnotesize$y_0,\ldots,y_0$};
    \draw (5.7,6.5) rectangle ++(3.2,1) node[pos=.5]{};
    \draw[<-] (5.7,4.5) -- ++ (1.2,0) node[anchor=west] {\footnotesize$t$};
    \draw[<-] (8.9,4.5) -- ++ (-1.2,0) ;

    \draw (9.8,5) rectangle ++(3.2,1) node[pos=.5]{};
    \draw (9.8,6.5) rectangle ++(3.2,1) node[pos=.5]{\footnotesize$x_0,\ldots,x_0$};
    \draw[<-] (9.8,4.5) -- ++ (1.2,0) node[anchor=west] {\footnotesize$t$};
    \draw[<-] (13,4.5) -- ++ (-1.2,0) ;
    \node[align=center] at (2.2,3.6) {{\em \footnotesize Unauth. MAC code}};
    \node[align=center] at (7.2,3.6) {{\em \footnotesize Auth. 1}};
    \node[align=center] at (11.3,3.6) {{\em \footnotesize Auth. 2}};

\end{tikzpicture}
\caption{A 3-phase scheme for authenticated communication over MAC. The tag blocklengths $t={\mathcal O}(\log n)$ of the authentication phases are much smaller than the blocklength $n$ of the preceding unauthenticated communication phase.}\label{fig:MACscheme}
\end{figure}
Also, note that the resulting nominal channel $W(.|x_0,.)$ has non-zero capacity. If not, consider a distribution $P_{X'|Y}$ in \eqref{eq:overwrite lemma} such that $P_{X'|Y}(x_0|y) = 1$, for all $y\in{\mathcal Y}$. This implies that $W(z|x_0,y')\neq W(z|x_0,y)$ for some $z\in \mathcal{Z}$ and $y,y'\in \mathcal{Y}$.
Similarly, for the AVC from user~1 to the receiver with ${\mathcal Y}$ as the state alphabet, we can find a no-attack state $y_0\in \mathcal{Y}$ for which it is non-overwritable. Our coding scheme is depicted in Fig.~\ref{fig:MACscheme}. The decoder outputs a message pair estimate only if both of the authentication decoders accept, otherwise it outputs $(\bot,\bot)$ to declare the presence of an adversarial user. We show vanishing error probabilities for all rates pairs in $\mathcal{R}_{\nonadv}$. The converse follows from arguing that a user (say, user~1) who can overwrite the other can induce a channel output distribution corresponding to any message pair $(m_1',m_2')$ of its choosing irrespective of the message ($m_2$) the overwritten user sends.\vspace{-0.5em}

\subsection{Proof sketches for the Gaussian case}
The detailed proofs are given in Appendix~\ref{App:Gaussian}.\\
We first prove Theorem~\ref{thm:auth comm gaussian}.   The achievability proof departs significantly from the finite alphabet setting of~\cite[Theorem~2]{KosutKITW18}. We use a two stage code consisting of a  channel code for (non-authenticated) communication over blocklength $n$ on the nominal Gaussian channel followed by followed by an authentication tag of blocklength $\tl=\cO(\log{n})$. The existence of such a code is guaranteed by Lemma~\ref{lem:auth gaussian positive}, which shows that the number of authenticable messages scales doubly exponentially with the tag blocklength.\vspace{-0.25em}
\begin{lemma}\label{lem:auth gaussian positive}
For an additive Gaussian AVC with input power constraint $25\sigma^2$ and noise variance $\sigma^2$, $C_{\mathrm{auth}}>0$.\vspace{-0.25em}\vspace{-0.25em}
\end{lemma}
The main ideas behind the proof of Lemma~\ref{lem:auth gaussian positive} are illustrated in Fig.~\ref{fig:auth gaussian}. Using Theorem~\ref{thm:auth comm gaussian}, the proofs of Theorems~\ref{thm:auth gaussian} and~\ref{thm:auth comm gaussian mac}  follow on similar lines as the finite alphabet setting. The converse in Theorem~\ref{thm:auth gaussian} follows from the identification capacity converse for the nominal Gaussian channel~\cite{BurnashevIT00}, while the achievability uses an authenticated communication code for the adversarial Gaussian AVC to authenticably transmit an identification code for the noiseless channel. Finally, we show the achievability in Theorem~\ref{thm:auth comm gaussian mac} by using a three-phase communication scheme using the authentication code from Theorem~\ref{thm:auth gaussian} (similar to the proof of Theorem~\ref{thm:authMAC}).\vspace{-0.25em}
\begin{figure}[t]\centering
\scalebox{0.63}{\begin{tikzpicture}[scale=0.35]
	\draw[dotted] (-1,1) rectangle ++(4,12) node[anchor=west]{$\tagenc$};
	\draw (0,2) rectangle ++(2,2) node[pos=.5]{$g$};
	\draw (0,6) rectangle ++(2,2) node[pos=.5]{$\permchoice$};
	\draw (0,10) rectangle ++(2,2) node[pos=.5]{$\idenc$};
	\draw[<-] (1,0) node[below]{$\tagenc(\msg)\in\cA_{\msg}$} -- ++ (0,2) ;
	\draw[<-] (1,4) -- node[above] {} ++ (0,2);
	\draw[<-] (1,8) -- node[above] {} ++ (0,2);
	\draw[<-] (1,12) -- ++ (0,2) node[above]{$\msg$};
\end{tikzpicture}}\ \ \qquad\scalebox{1.13}{\begin{tikzpicture}[scale=0.4]

\coordinate (m1) at ({3*cos(70)},{3*sin(70)});
\coordinate (m2) at ({3*cos(300)},{3*sin(300)});
\coordinate (m3) at ({3*cos(210)},{3*sin(210)});

\coordinate (mh1) at ({3*cos(100)},{3*sin(100)});
\coordinate (mh2) at ({3*cos(370)},{3*sin(370)});
\coordinate (mh3) at ({3*cos(135)},{3*sin(135)});

\coordinate (mmh) at ({3*cos(175)},{3*sin(175)});

\coordinate (cw1) at ({3*cos(270)},{3*sin(270)});
\coordinate (cw2) at ({3*cos(330)},{3*sin(330)});
\coordinate (cw3) at ({3*cos(40)},{3*sin(40)});
\coordinate (cw4) at ({3*cos(240)},{3*sin(240)});

\coordinate (s) at (.3,1.6);

\coordinate (icw) at ({6.5*cos(150)-1.25},{6.5*sin(150)+1});
\coordinate (im) at ({6.5*cos(150)-1.25},{6.5*sin(150)+0.5});
\coordinate (imh) at ({6.5*cos(150)-1.25},{6.5*sin(150)});
\coordinate (immh) at ({6.5*cos(150)-1.25},{6.5*sin(150)-0.5});

\draw[line width=0.01,gray,opacity=0.3] (0,0) ++ (3,0) arc  (0:360:3) -- cycle;

\fill[fill opacity=1,fill=green!60!olive] (m1) ++ (.1,0) arc (0:360:.1) -- cycle;
\fill[fill opacity=1,fill=green!60!olive] (m2) ++ (.1,0) arc (0:360:.1) -- cycle;
\fill[fill opacity=1,fill=green!60!olive] (m3) ++ (.1,0) arc (0:360:.1) -- cycle;

\fill[fill opacity=0.3,fill=red!60!white] (mh1) ++ (.8,0) arc (0:360:.8) -- cycle;
\fill[fill opacity=0.3,fill=red!60!white] (mh2) ++ (.8,0) arc (0:360:.8) -- cycle;
\fill[fill opacity=0.3,fill=red!60!white] (mh3) ++ (.8,0) arc (0:360:.8) -- cycle;

\fill[fill opacity=1,fill=red!60!blue] (mh1) ++ (.1,0) arc (0:360:.1) -- cycle;
\fill[fill opacity=1,fill=red!60!blue] (mh2) ++ (.1,0) arc (0:360:.1) -- cycle;
\fill[fill opacity=1,fill=red!60!blue] (mh3) ++ (.1,0) arc (0:360:.1) -- cycle;

\fill[fill opacity=0.3,fill=red!60!white] (mmh) ++ (.8,0) arc (0:360:.8) -- cycle;

\fill[fill opacity=1,fill=red!60!blue] (mmh) ++ (.1,0) arc (0:180:.1) -- cycle;
\fill[fill opacity=1,fill=green!60!olive] (mmh) ++ (-.1,0) arc (180:360:.1) -- cycle;

\fill[-,white,opacity=1] (cw1) ++ (.1,0) arc  (0:360:.1) -- cycle; 
\fill[-,white,opacity=1] (cw2) ++ (.1,0) arc  (0:360:.1) -- cycle; 
\fill[-,white,opacity=1] (cw3) ++ (.1,0) arc  (0:360:.1) -- cycle; 
\fill[-,white,opacity=1] (cw4) ++ (.1,0) arc  (0:360:.1) -- cycle; 

\draw[-,black,opacity=1] (m1) ++ (.1,0) arc  (0:360:.1) -- cycle; 
\draw[-,black,opacity=1] (m2) ++ (.1,0) arc  (0:360:.1) -- cycle; 
\draw[-,black,opacity=1] (m3) ++ (.1,0) arc  (0:360:.1) -- cycle; 
\draw[-,black,opacity=1] (mh1) ++ (.1,0) arc  (0:360:.1) -- cycle; 
\draw[-,black,opacity=1] (mh2) ++ (.1,0) arc  (0:360:.1) -- cycle; 
\draw[-,black,opacity=1] (mh3) ++ (.1,0) arc  (0:360:.1) -- cycle; 
\draw[-,black,opacity=1] (mmh) ++ (.1,0) arc  (0:360:.1) -- cycle; 

\draw[-,black,opacity=1] (cw1) ++ (.1,0) arc  (0:360:.1) -- cycle; 
\draw[-,black,opacity=1] (cw2) ++ (.1,0) arc  (0:360:.1) -- cycle; 
\draw[-,black,opacity=1] (cw3) ++ (.1,0) arc  (0:360:.1) -- cycle; 
\draw[-,black,opacity=1] (cw4) ++ (.1,0) arc  (0:360:.1) -- cycle; 

\fill[green,opacity=0.2] (m3) ++ (s) ++ (.1,0) arc  (0:360:.1) -- cycle; 
\draw[line width=0.3,black,opacity=0.1] (m3) ++ (s) ++ (.1,0) arc  (0:360:.1) -- cycle; 
\fill[green,opacity=0.2] (m1) ++ (s) ++ (.1,0) arc  (0:360:.1) -- cycle; 
\draw[line width=0.3,black,opacity=0.1] (m1) ++ (s) ++ (.1,0) arc  (0:360:.1) -- cycle; 
\fill[green,opacity=0.2] (m2) ++ (s) ++ (.1,0) arc  (0:360:.1) -- cycle; 
\draw[line width=0.3,black,opacity=0.1] (m2) ++ (s) ++ (.1,0) arc  (0:360:.1) -- cycle; 

\draw (m1) node[below] {\tiny $\vectu(1)$};
\draw (m2) node[right] {\tiny $\vectu(2)$};
\draw (m3) node[left] {\tiny $\vectu(3)$};

\draw (mh1) node[above] {\tiny $\hat{u}^\tl(4)$};
\draw (mh2) node[left] {\tiny $\hat{u}^\tl(1)$};
\draw (mh3) node[above ] {\tiny $\hat{u}^\tl(2)$};

\draw (mmh) node[left] {\tiny $\vectu(4)=\hat{u}^\tl(3)$};

\draw[green!40!black] (m1)++(s) node[left] {\tiny $\vectu(1)+\state^\tl$};
\draw[green!40!black] (m2)++(s) node[left] {\tiny $\vectu(2)+\state^\tl$};

\draw[red!70!black] (m3)++(s) node[right] {\tiny $\vectu(3)+\state^\tl$};

\draw[brown,dotted,->] (m1) -- ++ (s) node[midway,above=-.2em,sloped] {\tiny $\state^\tl$};
\draw[brown,dotted,->] (m2) -- ++ (s) node[midway,above=-.2em,sloped] {\tiny $\state^\tl$};
\draw[brown,dotted,->] (m3) -- ++ (s) node[midway,above=-.2em,sloped] {\tiny $\state^\tl$} ;

\draw[-,black,opacity=1] (icw) ++ (.1,0) arc  (0:360:.1) -- cycle; 
\draw(icw) node[right] {\tiny all codewords of $g$};

\fill[fill opacity=1,fill=green!60!olive] (im) ++ (.1,0) arc (0:360:.1) -- cycle;
\draw(im) node[right] {\tiny $\cA_{\msg}\setminus\cA_{\msgh}$};

\fill[fill opacity=1,fill=red!60!blue] (imh) ++ (.1,0) arc (0:360:.1) -- cycle;

\draw(imh) node[right] {\tiny $\cA_{\msgh}\setminus\cA_{\msg}$};

\fill[fill opacity=1,fill=red!60!blue] (immh) ++ (.1,0) arc (0:180:.1) -- cycle;
\fill[fill opacity=1,fill=green!60!olive] (immh) ++ (-.1,0) arc (180:360:.1) -- cycle;

\draw(immh) node[right] {\tiny $\cA_{\msgh}\cap\cA_{\msg}$};

\end{tikzpicture}}
\caption{To prove Lemma~\ref{lem:auth gaussian positive}, we consider a stochastic encoder $\tagenc$ for the authentication code that is the composition of an identification encoder $\idenc$ for the identity channel from~\cite[Proposition~1]{AhlswedeDIT89}, an appropriately chosen permutation $\permchoice$, and an encoder $g$ for a channel code with codewords of power $25\sigma^2$ and minimum distance at least $7\sigma\sqrt{\tl}$. The identification code guarantees a small overlap between $\cA_\msg$ (all codewords for $\msg$, denoted by $u^\tl(i)$'s), and $\cA_{\msgh}$ (all codewords for $\msgh$, denoted by $\hat{u}^\tl(i)$'s)  for every pair of messages $\msg\neq\msgh$. Next, we use a concentration argument over the random selection (under a uniform choice of $\permchoice$) of the set $\cA_\msg\setminus\cA_{\msgh}$ from the set of all codewords for $g$ that are not in $\cA_{\msgh}$ to argue that, w.h.p., under a fixed $\state^\tl$, at most a small fraction of codewords from $\cA_\msg\setminus\cA_{\msgh}$ lie within a distance $3\sigma\sqrt{t}$ of any codeword in $\cA_{\msgh}$. This shows a small misauthentication probability for a fixed $\msg,\msgh$, and $\state^\tl$. Finally, we take a union bound to show that this also holds for all $\msg$, $\msgh$, and a sufficiently dense set of adversarial state vectors.}\label{fig:auth gaussian}
\end{figure}
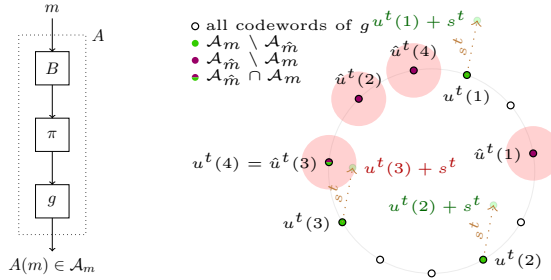

\begin{appendices}
\section{Authentication over DM-AVCs: capacity characterization (Proof of Theorem~\ref{thm: auth capacity})}\label{App:AuthAVC}
	Before proving this result, we restate the definition of an \emph{identification code} over an AVC from~\cite{BocheD:19}.
\begin{defn}[Identification Code]\label{def:ident} An $(\nummsg,\tl)$-identification code $(\idenc,\iddec)$ for an arbitrarily varying channel \tagchannel\ with state space $\cS$ consists of
\begin{enumerate}
\item an index set $\msgset=\set{1,2,\ldots,\nummsg}$,
\item a stochastic encoder $\idenc:\msgset\to\cU^\tl$, and
\item a decoder $\iddec:\msgset\times \cV^\tl\to\set{\accept,\reject}$.
\end{enumerate}
\end{defn}

The {\em misidentification probability} $P_{e,id}(\idenc,\iddec)$ for an identification code $(\idenc,\iddec)$ is the maximum, over all channel state sequences and messages, of $P_{\reject}$, the probability of incorrectly rejecting the true message and  $P_{\accept}$, the probability of incorrectly accepting any given message other than the true message, \emph{i.e.},
\begin{align*}
&P_{\reject} \defineqq \max_{\msg\in\msgset,\state^\tl\in\cS^\tl}\Prob\left(\iddec(\msg,\vecrV)=\reject\vert \State^\tl=\state^\tl,\Msg=\msg\right), \label{eq:pr}\\
&P_{\accept} \defineqq \max_{\msg\in\msgset,\state^\tl\in\cS^\tl,\msgh\neq\msg} \Prob\left(\iddec(\msgh,\vecrV)=\accept\vert \State^\tl=\state^\tl,\Msg=\msg\right),\mbox{ and}\\
&P_{e,id}\defineqq \max\set{P_{\reject},P_{\accept}}.
\end{align*}
\begin{defn}[Achievable Rates, Identification Capacity]\label{def:idcap} We say that $R$ is an achievable identification rate for the AVC $\tagchannel$ if there exists a sequence of  codes $\set{(\idenc^{(\tl)},\iddec^{(\tl)})}_{\tl=1}^\infty$ such that for each $\tl$, $(\idenc^{(\tl)},\iddec^{(\tl)})$ is a $(\lfloor 2^{2^{\tl R}}\rfloor,\tl)$-identification code with the identification error probabilities satisfying $\lim_{\tl\to\infty} P_{e,id}(\idenc^{(\tl)},\iddec^{(\tl)})=0$. Finally, we define the {\em identification capacity} of the AVC \tagchannel\ as the supremum over all achievable identification rates.	
\end{defn}

\begin{proof}[Proof of Theorem~\ref{thm: auth capacity}]For the converse, we note that any $(\nummsg,\tl)$-authentication code $(\tagenc,\tagdec)$ for the AVC \tagchannel\ is also an $(\nummsg,\tl)$-identification code for the nominal channel $W_{\rV|\tU,\State}(.|.,\state_0)$ such that $P_{e,\mathrm{id}}(\tagenc,\tagdec)\leq P_{e,\mathrm{auth}}(\tagenc,\tagdec)$. Thus, the authentication capacity for the the AVC \tagchannel\ is bounded from above by the identification capacity of the the nominal channel $W_{\rV|\tU,\State}(.|.,\state_0)$. By~\cite{AhlswedeDIT89}, the latter equals $C(\state_0)$. 
For $u^t\in \mathcal{U}^t$ and $m\in \mathcal{M}$, define the distribution $Q_{U^t|M}(u^t|m)\defineqq\Prob(\tagenc(m)=u^t)$. When the channel is overwritable, suppose some $R>0$ is achievable. Let $m,m'\in \mathcal{M}, m\neq m'$ be two distinct messages. Consider the adversarial strategy  $P_{S^t}(s^t) = \sum_{\tilde{u}^t\in\mathcal{U}^t}Q_{U^t|M}(\tilde{u}^t|m')\prod_{i=1}^{t}P_{S|U'}(s_i|\tilde{u}_i)$. Here, $s_i$ and $\tilde{u}_i$ are $i^{\text{th}}$ elements of the $s^n$ and $\tilde{u}^n$ sequences respectively and $P_{S|U'}$ satisfies \eqref{eq:overwritable}.

Let $P_{\accept,m\rightarrow m'}$ be the probability of accepting $m'$ under the adversarial attack described above when $m$ was actually sent  and $P_{\reject,m}$ be the probability of rejecting $m$, when $m$ was sent and there was no attack. For $m\in \mathcal{M}$, let $\alpha^{-1}(\accept,m) \defineqq \{v^t\,:\,\alpha(m,v^t)=\accept\}$. Then, 
\begin{align*}
P_{\accept,m\rightarrow m'} &= \sum_{u^t\in\mathcal{U}^t}Q(u^t|m)\sum_{s^t\in \mathcal{S}^t}P_{S^t}(s^t)\sum_{v^t\in\alpha^{-1}(\accept,m')}\prod_{i=1}^t W(v_i|u_i,s_i)\\
&=\sum_{u^t\in\mathcal{U}^t,v^t\in\alpha^{-1}(\accept,m')}Q(u^t|m)\sum_{\tilde{u}^t\in\mathcal{U}^t}Q(\tilde{u}^t|m')\sum_{s^t\in \mathcal{S}^t}\prod_{i=1}^{t}P_{S|U'}(s_i|\tilde{u}_i) W(v_i|u_i,s_i)\\
&\stackrel{\text{(a)}}{=}\sum_{u^t\in\mathcal{U}^t,v^t\in\alpha^{-1}(\accept,m')}Q(u^t|m)\sum_{\tilde{u}^t\in\mathcal{U}^t}Q(\tilde{u}^t|m')\prod_{i=1}^{t}W(v_i|\tilde{u}_i,s_0)\\
&=\sum_{\tilde{u}^t\in\mathcal{U}^t,v^t\in\alpha^{-1}(\accept,m')}Q(\tilde{u}^t|m')\prod_{i=1}^{t}W(v_i|\tilde{u}_i,s_0)\\
&=1-P_{\reject,m'},
\end{align*}
where (a) follows from \eqref{eq:overwritable}.
This implies that $P_{\reject,m'} + P_{\accept,m\rightarrow m'}=1$. Hence, $P_{\reject} + P_{\accept}\geq1.$ This  implies that $\nummsg\leq 1$ for any $(\nummsg,\tl)$-authentication code with $P_{e,\auth}(A,\alpha)<\frac{1}{2}$. Therefore, the authentication capacity of \tagchannel\ is zero.

Next, we show the achievability when \tagchannel\ is not overwritable by constructing a sequence of authentication codes for large enough blocklengths. Our construction consists of an identification code for an identity channel that is sent using an authenticated communication code. 

Let $R<C(\state_0)$, $\epsilon\in[0,1]$, and ${\tl}^*$ be large enough. Lemma~\ref{lemma:authCOMavgTOstoc} and \cite[Proposition~1]{AhlswedeDIT89} guarantee that the following hold for every $\tl>\tl^*$. Let $\tilde{\tl}=\lceil R\tl\rceil$ and let $\delta>0$.
\begin{itemize}
\item There exists a (stochastic) $(2^{2^{\tilde{\tl}-\delta}},\tilde{\tl})$-identification code $(\idenc,\iddec)$ for the identity channel with $P_{e,id}\leq \epsilon/2$. 
\item There exits a (stochastic) $(2^{\tilde{\tl}},\tl)$-authenticated communication code  $(F,\phi)$ for the channel \tagchannel\ with the nominal state $\state_0$ and $e(F,\phi)\leq\epsilon/2$ (see Appendix~\ref{App:AVC_stoc_det} for the definition of an authenticated communication code and its error probability). 
\end{itemize}
Consider the stochastic encoder $\tagenc:\set{0,1,\ldots,2^{2^{\tilde{\tl}-\delta}}}\to\cU^{\tl}$ and the decoder $\tagdec:\set{0,1,\ldots,2^{2^{\tilde{\tl}-\delta}}}\times\cV^{\tl}\to\set{\accept,\reject}$, where, for every $\msg,\msgh\in\set{0,1,\ldots,2^{2^{\tilde{\tl}-\delta}}}$, $u^t\in\mathcal{U}^t$, and $\rv^\tl\in\cV^\tl$, 
 \begin{align*}
 	&\Prob_{\tagenc}\left(\tagenc(\msg)=u^t\right)=\Prob_{F,\idenc}\left(F(\idenc(\msg))=\tu^{\tl}\right),\mbox{ and}\\
 	&\tagdec(\msgh,\rv^\tl)=\begin{cases}\accept &\mbox{if } \phi(\rv^\tl)\neq\bot\mbox{ and }\iddec(\msgh,\phi(\rv^\tl))=\accept\\ \reject & \mbox{otherwise.}\end{cases}
 \end{align*}
We claim that $(\tagenc,\tagdec)$ is a $(2^{2^{\tilde{\tl}-\delta}},\tl)$-authentication code for  the channel \tagchannel\ with the nominal state $\state_0$ and $P_{e,\auth}(\tagenc,\tagdec)\leq\epsilon$. Note that as $\epsilon$ and $\delta$ can be made arbitrarily small, showing the above amounts to proving the achievability all all rates below $C(\state_0)$.

To this end, let $\Msg$, $\tilde{U}^{\ttl}$ and $\vectU$, be the random variables denoting the message, the stochastic output of the identification encoder $B$ upon input $\Msg$, and the stochastic output of the encoder $F$ upon input $\tilde{U}^{\ttl}$, respectively. For any fixed $\state^\tl\in\State^\tl$, let $\vecrV$ denote the random vector corresponding to the channel output when the channel input is $\vectU$. 

We first analyze the probability that $(\tagenc,\tagdec)$ rejects the true message under no attack. Let $\cE_1\defineqq\set{\phi(\vecrV)\neq\tilde{U}^{\ttl}}$ denote the event that the authenticated communication code  makes a decoding error. Note that when there is no attack, the probability of the event $\cE_1$ is no larger than $\epsilon/2$ by our choice of $(F,\phi)$. Further when under $\cE_1^c$, $\tagdec(\msgh,\vecrV)=\iddec(\msgh,\tilde{U}^{\ttl})$ for every $\msgh\in\set{0,1,\ldots,2^{2^{\tilde{\tl}-\delta}}}$. Using these, we note that, for any $\msgh\in\set{0,1,\ldots,2^{2^{\tilde{\tl}-\delta}}}$, 
\begin{align}
&\Prob_{\tagenc,W}\left(\tagdec(\msgh,\vecrV)=\reject\vert \State^\tl=\state_0^\tl,\Msg=\msg\right)	\\
& = \Prob_{\tagenc,W}\left(\set{\tagdec(\msgh,\vecrV)=\reject},\cE_1\vert \State^\tl=\state_0^\tl,\Msg=\msg\right)+\Prob_{\tagenc,W}\left(\set{\tagdec(\msgh,\vecrV)=\reject},\cE_1^c\vert \State^\tl=\state_0^\tl,\Msg=\msg\right)\\
&\leq \Prob_{F,W}\left(\cE_1\vert \State^\tl=\state_0^\tl,\Msg=\msg\right)+\Prob_{\idenc,F,W}\left(\tagdec(\msgh,\vecrV)=\reject\vert\cE_1^c, \State^\tl=\state_0^\tl,\Msg=\msg\right)\\
&= \Prob_{F,W}\left(\cE_1\vert \State^\tl=\state_0^\tl,\Msg=\msg\right) + \Prob_{\idenc}\left(\iddec(\msgh,\tilde{U}^{\ttl})=\reject\vert\Msg=\msg\right)\\
& \leq \epsilon/2 + \Prob_{\idenc}\left(\iddec(\msgh,\tilde{U}^{\ttl})=\reject\vert\Msg=\msg\right).\label{eq:auth reject upper bd}\\
\intertext{Recall that, when $\msgh$ equals $\msg$, by our choice of $(\idenc,\iddec)$, $\Prob_{\idenc}\left(\iddec(\msg,\tilde{U}^{\ttl})=\reject\vert\Msg=\msg\right)\leq\epsilon/2$. Thus, we have}
&\Prob_{\tagenc,W}\left(\tagdec(\msg,\vecrV)=\reject\vert \State^\tl=\state_0^\tl,\Msg=\msg\right)\leq \epsilon.
\end{align}
Next, we analyze the probability that $(\tagenc,\tagdec)$ accepts an incorrect message under any state vector. Let $\cE_\bot\defineqq \set{\phi(\vecrV)=\bot}$. Note that $\cE_\bot\subseteq\cE_1$ and by our choice of $(F,\phi)$, the probability of the event $\cE_1\setminus\cE_\bot$ is no larger than $\epsilon/2$ for any state vector. Further, when under $\cE_\bot$, $\tagdec(\msgh,\vecrV)$ equals $\reject$ with probability one. Thus,  for any $\msg,\msgh\in\set{0,1}^{\tilde{\tl}}$ and $\state^\tl\in\cS^\tl$, we have
\begin{align}
&\Prob_{\tagenc,W}\left(\tagdec(\msgh,\vecrV)=\accept\vert \State^\tl=\state^\tl,\Msg=\msg\right)\\
&=\Prob_{\tagenc,W}\left(\set{\tagdec(\msgh,\vecrV)=\accept},\cE_\bot\vert \State^\tl=\state^\tl,\Msg=\msg\right)+ \Prob_{\tagenc,W}\left(\set{\tagdec(\msgh,\vecrV)=\accept},\cE_1\setminus\cE_\bot\vert \State^\tl=\state^\tl,\Msg=\msg\right)\\
&\qquad+ \Prob_{\tagenc,W}\left(\set{\tagdec(\msgh,\vecrV)=\accept},\cE_1^c\vert \State^\tl=\state^\tl,\Msg=\msg\right)\\
& = 0 + \Prob_{\tagenc,W}\left(\set{\tagdec(\msgh,\vecrV)=\accept},\cE_1\setminus\cE_\bot\vert \State^\tl=\state^\tl,\Msg=\msg\right) + \Prob_{\tagenc,W}\left(\set{\tagdec(\msgh,\vecrV)=\accept},\cE_1^c\vert \State^\tl=\state^\tl,\Msg=\msg\right)\\
&\leq  \Prob_{F,W}\left(\cE_1\setminus\cE_\bot\vert \State^\tl=\state^\tl,\Msg=\msg\right) + \Prob_{\idenc,F,W}\left(\tagdec(\msgh,\vecrV)=\accept\vert\cE_1^c, \State^\tl=\state^\tl,\Msg=\msg\right)\\
& = \Prob_{F,W}\left(\cE_1\setminus\cE_\bot\vert \State^\tl=\state^\tl,\Msg=\msg\right) + \Prob_{\idenc}\left(\iddec(\msgh,\tilde{U}^{\ttl})=\accept\vert\Msg=\msg\right)\\
&\leq \epsilon/2 + \Prob_{\idenc}\left(\iddec(\msgh,\tilde{U}^{\ttl})=\accept\vert\Msg=\msg\right)\\
\intertext{When $\msgh\neq\msg$, by our choice of $(\idenc,\iddec)$, $\Prob_{\idenc}\left(\iddec(\msgh,\tilde{U}^{\ttl})=\accept\vert\Msg=\msg\right)\leq\epsilon/2$. Thus, we have}
&\Prob_{\tagenc,W}\left(\tagdec(\msgh,\vecrV)=\accept\vert \State^\tl=\state^\tl,\Msg=\msg\right)\leq \epsilon.
\end{align}
This shows that $\peauth(\tagenc,\tagdec)\leq\epsilon$.

\end{proof}

\section{Authenticated Communication over Multiple Access Channels (Proof of Theorem~\ref{thm:authMAC})}\label{App:authMAC}

\begin{proof}[Proof of Theorem~\ref{thm:authMAC}]
{\em Converse: }Suppose user~1 can overwrite user~2. Let, if possible, for some $N_2>1$, there exists an $(N_1,N_2,n)$ authenticated communication code with vanishing error probabilities. Define $Q_1(x^n|m_1) \defineqq \Prob(F_1(m_1) = x^n),\,m_2\in \mathcal{M}_2$ and $Q_2(y^n|m_2) \defineqq \Prob(F_2(m_2) = y^n),\,m_2\in\mathcal{M}_2$.
 Let $m_1\in \mathcal{M}_1,\, m_2, m'_2\in \mathcal{M}_2$ such that $m_2\neq m'_2$. Consider an adversarial strategy $\tilde{P}_{X^n}$ where user~1 overwrites any message sent by user 2 by an encoding of $m'_2$ while pretending to send message $m_1$, i.e., for $\tilde{x}^n \in \mathcal{X}^n,$
\begin{align}\label{eq:attack_MAC}
\tilde{P}_{X^n}(\tilde{x}^n) = \sum_{x^n\in\mathcal{X}^n, y^n\in\mathcal{Y}^n}Q_1(x^n|m_1)Q_2(y^n|m'_2)\left(\prod_{i=1}^{n}P_{X'|X,Y}\left(\tilde{x}_i|x_i, y_i\right)\right).
\end{align}
where $\tilde{x}_i, x_i$ and  $y_i$ are $i^{\text{th}}$ elements of the $\tilde{x}^n, x^n$ and $y^n$ sequences respectively and $P_{X'|X,Y}$ satisfies \eqref{eq1}. Let  $P_{e,\malone,m_1,m_2\rightarrow m'_2}$ be the probability of accepting $(m_1,m'_2)$ when user 2 sends $m_2$ and user 1 sends an $\tilde{x}^n$ sequence according to the distribution given by~\eqref{eq:attack_MAC}. Let $P_{e,\na,m_1,m_2}$ denote the probability of rejecting $(m_1,m_2)$ when $(m_1,m_2)$ was sent and there was no attack. Define $\phi^{-1}(m_1, m'_2) \defineqq \{z^n\,:\, \phi(z^n) = (m_1,m'_2)\}$ and $W^n(z^n|\tilde{x}^n, \tilde{y}^n)\defineqq \prod_{i=1}^{n}W(z_i|x_i,y_i)$ where $z_i$ is the $i^{\text{th}}$ element of $z^n$. Then, 
\begin{align*}
&P_{e,\malone,m_1,m_2\rightarrow m'_2} = \sum_{\tilde{x}^n\in \mathcal{X}^n, \tilde{y}^n\in \mathcal{Y}^n}\tilde{P}_{X^n}(\tilde{x}^n)Q_2(\tilde{y}^n|m_2)\Prob\left(\left(\hat{M}_1,\hat{M}_2\right) = \left(m_1,m'_2\right)\left|X^n = \tilde{x}^n, Y^n = \tilde{y}^n\right.\right)\\
&\qquad=  \mathop{\mathop{\sum}_{\tilde{x}^n\in \mathcal{X}^n, \tilde{y}^n\in \mathcal{Y}^n}}_{z^n \in \phi^{-1}(m_1, m'_2)}\tilde{P}_{X^n}(\tilde{x}^n)Q_2(\tilde{y}^n|m_2)W^n(z^n|\tilde{x}^n, \tilde{y}^n)\\
&\qquad= \mathop{\mathop{\sum}_{\tilde{x}^n\in \mathcal{X}^n, \tilde{y}^n\in \mathcal{Y}^n}}_{z^n \in \phi^{-1}(m_1, m'_2)}Q_2(\tilde{y}^n|m_2)\sum_{x^n\in\mathcal{X}^n, y^n\in\mathcal{Y}^n}Q_1(x^n|m_1)Q_2(y^n|m'_2)\left(\prod_{i=1}^{n}P_{X'|X,Y}\left(\tilde{x}_i|x_i, y_i\right)W(z_i|\tilde{x}_i, \tilde{y}_i)\right)\\
&\qquad=\mathop{\mathop{\sum}_{z^n \in \phi^{-1}(m_1, m'_2)}}_{\tilde{y}^n\in \mathcal{Y}^n}Q_2(\tilde{y}^n|m_2)\sum_{x^n\in\mathcal{X}^n, y^n\in\mathcal{Y}^n}Q_1(x^n|m_1)Q_2(y^n|m'_2)\left(\prod_{i=1}^{n}\left(\sum_{\tilde{x}_i\in \mathcal{X}}P_{X'|X,Y}\left(\tilde{x}_i|x_i, y_i\right)W(z_i|\tilde{x}_i, \tilde{y}_i)\right)\right)\\
&\qquad\stackrel{\text{(a)}}{=}\mathop{\mathop{\sum}_{z^n \in \phi^{-1}(m_1, m'_2)}}_{\tilde{y}^n\in \mathcal{Y}^n}Q_2(\tilde{y}^n|m_2)\sum_{x^n\in\mathcal{X}^n, y^n\in\mathcal{Y}^n}Q_1(x^n|m_1)Q_2(y^n|m'_2)W^n(z^n|x^n,y^n)\\
&\qquad=\sum_{x^n\in\mathcal{X}^n, y^n\in\mathcal{Y}^n,z^n \in \phi^{-1}(m_1, m'_2)}Q_1(x^n|m_1)Q_2(y^n|m'_2)W^n(z^n|x^n,y^n)\\
&\qquad= 1- P_{e,\na,m_1,m'_2},
\end{align*}
where (a) follows from \eqref{eq1}. This implies that $P_{e,\na,m_1,m'_2} + P_{e,\malone,m_1,m_2\rightarrow m'_2}=1$. Hence, we get $P_{e,\na}+P_{e,\malone}\geq 1$ whenever $N_2>1$. Thus, $N_2 = 1$ and $\mathcal{R}_{\authcomm} \subseteq  \{(R,0)\,|\,(R,0)\in \mathcal{R}_{\nonadv}\}$. Similarly, when user~2 can overwrite user~1, $\mathcal{R}_{\authcomm} \subseteq \{(0,R)\,|\,(0,R)\in \mathcal{R}_{\nonadv}\}$. This also implies that when both users can overwrite each other, $\mathcal{R}_{\authcomm} = \{(0,0)\}.$ \\

{\em Achievability: }
Suppose neither user can overwrite the other. The scheme when one user can overwrite the other, but not vice versa, follows similarly. Our achievability scheme consists of three phases: a DM-MAC code (of blocklength $n$, say) with {\em stochastic} encoders for non-adversarial users with maximum probability of error vanishing as $n\rightarrow\infty$; this is followed by two authentication phases, one for each user.
While authenticating user~2, we consider the point-to-point channel from user~2 to the receiver and treat user~1 as possibly adversarial. For this AVC, user~1's alphabet corresponds to the set of adversarial states. The following lemma shows that if user~1 cannot overwrite user~2 (i.e., equation~\eqref{eq1} does not hold), this AVC has a ``\na'' state ($x_0$ in the lemma) such that it is not overwritable (i.e., equation~\eqref{eq:overwritable} does not hold). This will allow us to use the authentication codes  from Section~\ref{subsec:authAVC} with blocklength $t={\mathcal O}(\log{n})$.
\begin{duplicatelemma}[\ref{lemma:1}]
When user~1 cannot overwrite user~2, there exists $x_0\in \mathcal{X}$ such that for all $P_{X'|Y}$, there exists $(y,y',z)\in \mathcal{Y}\times\mathcal{Y}\times\mathcal{Z}$, that satisfies 
\begin{align*}
\sum_{x'\in\mathcal{X}}P_{X'|Y}(x'|y)W(z|x',y')\neq W(z|x_0,y). \tag*{\eqref{eq:overwrite lemma}}
\end{align*}
\end{duplicatelemma}

\begin{proof}
If user~1 cannot overwrite user~2, then from \eqref{eq1} it follows that for all $P_{X'|X,Y}$, there exist $y,y'\in \mathcal{Y},\, x\in \mathcal{X} \text{ and } z\in \mathcal{Z}$ such that
\begin{equation*}
\sum_{x'}P_{X'|X,Y}(x'|x,y)W(z|x',y') \neq W(z|x,y).
\end{equation*}
This implies that there is a symbol $x_0\in \mathcal{X}$ such that for all $P_{X'|Y}$, there are $y,y'\in \mathcal{Y}, z\in \mathcal{Z}$ such that 
\begin{equation}\label{appendix:eq3}
\sum_{x'}P_{X'|Y}(x'|y)W(z|x',y')\neq W(z|x_0,y).
\end{equation}
If not, then, for every $x\in\mathcal{X}$, we can find an attack distribution $P_{X'|X=x,Y}$ which satisfies \eqref{eq1}. 
\end{proof}
Also, note that the resulting nominal channel $W(.|x_0,.)$ has non-zero capacity. If not, consider a distribution $P_{X'|Y}$ in \eqref{eq:overwrite lemma} such that $P_{X'|Y}(x_0|y) = 1$, for all $y\in{\mathcal Y}$. This implies that $W(z|x_0,y')\neq W(z|x_0,y)$ for some $z\in \mathcal{Z}$ and $y,y'\in \mathcal{Y}$.
Similarly, for the AVC from user~1 to the receiver with ${\mathcal Y}$ as the state alphabet, we can find a no-attack state $y_0\in \mathcal{Y}$ for which the AVC is non-overwritable. Our coding scheme is depicted in Fig.~\ref{fig:authcomMAC}. We will use $C(x_0)$ and $C(y_0)$ to refer to the capacities of $W_{Z|X,Y}(.|x_0,.)$ and $W_{Z|X,Y}(.|.,y_0)$ respectively. 

Consider a rate pair $(R_1,R_2)\in \mathcal{R}_{\nonadv}$. For any $\delta, \epsilon>0$, we will show that there exists an $(N_1,N_2,n)$ authenticated communication code $(F_1,F_2,\phi) $ with $\lim_{n\rightarrow\infty}\frac{\log{N_1}}{n}\geq\frac{R_1-\delta}{1+2\delta}$, $\lim_{n\rightarrow\infty}\frac{\log{N_2}}{n}\geq\frac{R_1-\delta}{1+2\delta}$ and $P_e(F_1,F_2,\phi)\leq \epsilon$. We start by noting that for all large enough numbers, say $n\geq n'$, there exist $(2^{n(R_1-\delta)},2^{n(R_2-\delta)},n)$ stochastic codes $(F_{1,\text{MAC}},F_{2,\text{MAC}},\phi_{\text{MAC}})$ for the non-adversarial DM-MAC such that the maximum probability of error is bounded above by $\frac{\epsilon}{2}$.
 Fix $0<r< \min\{C(x_0),C(y_0)\}$ and $a>\frac{1}{r}$. Let $t\defineqq a \log{n}$. Theorem~\ref{thm: auth capacity} implies that for all large numbers, say $n\geq n''$, there exist $\left(2^{2^{tr}}, t\right)$ authentication codes $(A_{1},\alpha_1)$ for the adversarial AVC $W_{Z|X,Y}$ with $\mathcal{Y}$ as the set of states and $y_0$ as the \na state such that $\Prob_{e,\text{auth}}(A_1,\alpha_1) \leq \frac{\epsilon}{4}$. Similarly, for all large numbers, say $n\geq n'''$, there exist $\left(2^{2^{tr}}, t\right)$ authentication codes $(A_{2},\alpha_2)$ for the adversarial AVC $W_{Z|X,Y}$ with $\mathcal{X}$ as the set of states and $x_0$ as the \na state such that $\Prob_{e,\text{auth}}(A_2,\alpha_2)\leq \frac{\epsilon}{4}$. 
Choose $n>\max{\{n',n'', n'''\}}$ which also satisfies $2^{tr}>\max{\{nR_1,nR_2\}}$ and $t<\delta n$.\\
{\em Encoding.} For $(m_1,m_2)\in \mathcal{M}_1\times\mathcal{M}_2$, the encoders are defined as 
\begin{align*}
F_1\left(m_1\right) &\defineqq \left(F_{1,\text{MAC}}\left(m_1\right),F_{1,\text{auth}}\left(m_1\right),x_0^{t}\right),\\
F_2\left(m_2\right) &\defineqq \left(F_{2,\text{MAC}}\left(m_2\right),y_0^{t},F_{2,\text{auth}}\left(m_1\right)\right),
\end{align*}
where $x_0^{t}$ and $y_0^{t}$ are $t$ length sequences of $x_0$ and $y_0$ respectively. \\ 
{\em Decoding.} 
Let $(\tilde{m}_1,\tilde{m}_2) = \phi_{\text{MAC}}(z^n)$. The decoder is defined as 
 \begin{displaymath}
\phi\left(z^{n+2t}\right)=
\left\{\begin{array}{ll}
		\phi_{\text{MAC}}\left(z^n\right),&\alpha_1\left(\tilde{m}_1,z_{n+1}^{n+t}\right) = \accept \text{ and }\\& \alpha_2\left(\tilde{m}_2,z_{n+t+1}^{n+2t}\right) = \accept, \\
		\left(\bot,\bot\right),&\text{otherwise.}
	\end{array}
\right.
\end{displaymath} 
{\em Analysis of probability of error.} 
Let $(\tilde{M}_1,\tilde{M}_2) = \phi_{\text{MAC}}(Z^n)$. Under {\em no attack}, for all $(m_1,m_2) \in \mathcal{M}_2\times\mathcal{M}_2$,
\begin{align*}
&\Prob\left(\phi\left(Z^{n+2t}\right)\neq\left(m_1,m_2\right)\left|M_1 = m_1,M_2 =m_2\right.\right) =\Prob\left(\phi\left(Z^{n+2t}\right)\neq\left(m_1,m_2\right),\left(\tilde{M}_1,\tilde{M}_2 \right)\neq \left(m_1,m_2\right)\left|M_1 = m_1,M_2 =m_2\right.\right)\\
&\qquad\quad\quad+\Prob\left(\phi\left(Z^{n+2t}\right)\neq\left(m_1,m_2\right),\left(\tilde{M}_1,\tilde{M}_2 \right)= \left(m_1,m_2\right)\left|M_1 = m_1,M_2 =m_2\right.\right) \\
&\qquad\leq\Prob\left(\left(\tilde{M}_1,\tilde{M}_2 \right)\neq \left(m_1,m_2\right)\left|M_1 = m_1,M_2 =m_2\right.\right)+ \Prob \left(\left\{\alpha_1\left(m_1,Z_{n+1}^{n+t}\right)=\reject\right\}\left|M_1 = m_1\right.\right)\\
&\qquad\qquad+\Prob \left(\left\{\alpha_2\left(m_2,Z_{n+t +1}^{n+2t}\right)=\reject\right\}\left|M_2 = m_2\right.\right) \\
&\qquad\leq \frac{\epsilon}{2}+\frac{\epsilon}{4}+\frac{\epsilon}{4}\\
&\qquad\leq \epsilon
\end{align*} 
Under {\em attack by user~1}, for all $m_2\in \mathcal{M}_2$ and $x^{n+2t}\in \mathcal{X}^{n+2t}$,
\begin{align*}
&\Prob\left(\hat{M}_2\notin\left\{m_2,\bot\right\}\left|M_2 = m_2,X^{n+2t} =x^{n+2t}\right.\right)=\Prob\left(\tilde{M}_2=m_2,\hat{M}_2\notin\left\{m_2,\bot\right\}\left|M_2 = m_2,X^{n+2t} =x^{n+2t}\right.\right)\\
&\qquad\qquad+\Prob\left(\tilde{M}_2\neq m_2,\hat{M}_2\notin\left\{m_2,\bot\right\}\left|M_2 = m_2,X^{n+2t} =x^{n+2t}\right.\right)\\
&\qquad\leq0+\Prob\left(\alpha_2\left(\tilde{M}_2,Z_{n+t+1}^{n+2t}\right)=\accept\left|M_2 = m_2,X^{n+2t} =x^{n+2t},\tilde{M}_2\neq m_2\right.\right)\\
&\qquad\leq\frac{\epsilon}{4}.
\end{align*}
Similarly, under {\em attack by user~2}, for all $m_1\in \mathcal{M}_1$ and $y^{n+2t}\in \mathcal{Y}^{n+2t}$,
\begin{align*}
\Prob\left(\hat{M}_1\notin\left\{m_1,\bot\right\}\left|M_1 = m_1,Y^{n+2t} =y^{n+2t}\right.\right)
\leq\frac{\epsilon}{4}.
\end{align*}

Thus, $P_{e}(F_1,F_2,\phi) \leq \epsilon$. Let $\{\epsilon_k \}_{k=1}^{\infty}$ be a decreasing sequence of positive numbers 
with $\lim_{k\rightarrow \infty}\epsilon_k\rightarrow 0$. For each $\epsilon_k,\,k \in\{1,2,\ldots\}$, we can find an $n$ such
 that there exists a $(2^{n(R_1-\delta)},2^{n(R_2-\delta)},n+2t)$ authenticated communication code $(F_1^{n+2t},F_2^{n+2t},\phi^{n+2t})$ for the \mac \mch such that $P_{e}(F_1^{n+2t},F_2^{n+2t},\phi^{n+2t}) \leq \epsilon_k$. Thus, we can find a sequence of $(2^{n(R_1-\delta)},2^{n(R_2-\delta),n+2t})$ codes with $\lim_{n\rightarrow \infty}P_{e}(F_1^{n+2t},F_2^{n+2t},\phi^{n+2t})\rightarrow 0$. In fact, it is not difficult to see that we can find a sequence of $(2^{n\frac{R_1-\delta}{1+2 \delta}},2^{n\frac{R_2-\delta}{1+2 \delta}},n)$ authenticated communication codes $(F_1^{(n)},F_2^{(n)},\phi^{(n)})$ with $P_{e}(F_1^{(n)},F_2^{(n)},\phi^{(n)})\rightarrow 0$ as $n$ goes to infinity\footnote{First note that to achieve the given rates and error probabilities, we may use MAC codes of blocklength $n, n+1, n+2, \ldots$ and authentication codes of tag length $t, t+1, t+2, \ldots$ (since if $r$ is such that $2^{2^{tr}}>2^{nR}, \, 2^{2^{(t+1)r}}>2^{2^{(n+1)R}}$ for large enough $n$).Thus, we get codes for every third integer after some $n+2t$. If we choose $n$ such that $2^{2^{tr}}> \max{\{ 2^{(n+2)R_1},2^{(n+2)R_2}\}}$, we can increase the blocklength of just the MAC code to $n+1$ and $n+2$ which fills in the gaps while ensuring the same error probability.}. This implies that the rate pair $\left(\frac{R_1-\delta}{1+2\delta},\frac{R_2-\delta}{1+2\delta}\right)$ is achievable. Since, $\delta$ is any positive number and $\mathcal{R}_{\authcomm}$ is the closure of the set of all achievable rates, we have $(R_1,R_2)\in\mathcal{R}_{\authcomm}$. Hence, $\mathcal{R}_{\authcomm} = \mathcal{R}_{\nonadv}$ when the channel is non-overwritable.  

\end{proof}

\section{Gaussian channels} 
\label{App:Gaussian}
{
\subsection{Positivity of authentication capacity for additive Gaussian AVCs}
In the following, we prove Lemma~\ref{lem:auth gaussian positive}. We show that when the maximum average input power allowed is $25\sigma^2$, positive authentication rates are achievable. Note that this also implies the positivity of authentication capacity when the input power constraint is smaller than $25\sigma^2$ -- one may start with a larger block so that we have enough energy budget to obtain at least $25\sigma^2$ average power on part of the block.
\begin{proof}[Proof of Lemma~\ref{lem:auth gaussian positive}]
In the following, we argue that, for any desired upper bound on the misauthentication probability, for every large enough blocklength $\tl$, there exist authentication codes with positive rate for the additive Gaussian AVC. To this end, we let $\epsilon\in(0,1/2)$,  $0<R<1/8$, $0<\tilde{R}<1$, $0<\delta<\min\set{\epsilon/10,1/\left(2^{1/2\epsilon}+1\right)}$, $\ttl=\lfloor R\tl\rfloor$,  $\nummsg=\lfloor 2^{2^{\tilde{R}\ttl}}\rfloor$, and $\msgset=[1:\nummsg]$. Consider a $(\nummsg,\tl)$-authentication code chosen according the following random procedure.
	
	First, let $(\idenc,\iddec)$ be a $(\nummsg,\ttl)$-identification code for the identity channel~\cite[Proposition~1]{AhlswedeDIT89} satisfying the following properties:
	\begin{itemize}
		\item[(A)] For each $\msg\in\msgset$, the output of the stochastic encoder $\idenc(\msg)$ is uniformly distributed over a set $\mathcal{B}_{\msg}\subset \set{0,1}^{\ttl}$ with $|\mathcal{B}_{\msg}|=\lfloor\delta2^{\ttl}\rfloor$.
		\item[(B)] The $P_{e,\mathrm{id}}(\idenc,\iddec)=\max_{\msg,\msgh\in\msgset:\msg\neq\msgh} |\mathcal{B}_{\msgh}\cap\mathcal{B}_{\msg}|/|\mathcal{B}_{\msg}|<\epsilon/2$.
	\end{itemize}
 Next, let $(g,\gamma)$ be a $(2^{\ttl},\tl)$-channel code for the additive Gaussian AVC satisfying the following:
	\begin{itemize}
		\item[(C)] For each $\tilde{u}\in\set{0,1}^{\ttl}$, $\norm{g(\tilde{u})}=5\sigma\sqrt{ \tl}$.
		\item[(D)] For each pair $\tilde{u},\tilde{u}'\in\set{0,1}^{\ttl}$ with $\tilde{u}\neq\tilde{u}'$, $\norm{g(\tilde{u})-g(\tilde{u}')}\geq 7\sigma\sqrt{\tl}$.
	\end{itemize}
 The existence of such a code is guaranteed by~\cite{Blachman62}. Note that property~(D) implies that, for large enough $\tl$, the maximal probability of decoding error for $(g,\gamma)$ is no larger than $\epsilon/2$.
Finally, let $\permchoice:\set{0,1}^{\ttl}\to\set{0,1}^{\ttl}$ be an appropriately chosen permutation. We comment on the choice of $\permchoice$ in the proof below. Once $\permchoice$ is chosen, it is fixed and is assumed to be known to all parties. 

Using the above, define the $(\nummsg,\tl)$-authentication code $(\tagenc,\tagdec)$ as follows.
\begin{itemize}
\item Let $\tagenc:\msgset\to\bbR^\tl$ be a stochastic encoder with $\tagenc=g\circ\permchoice\circ\idenc$, \emph{i.e.}, for every $\msg\in\msgset$ and $\vectu\in\bbR^t$, $\Prob_\tagenc(\tagenc(\msg)=\vectu)=\Prob_\idenc(g(\permchoice(\idenc(\msg)))=\vectu)$. For each $\msg\in\msgset$, let $\cA_\msg$ be the set of codewords that are output with non-zero probability when the message is $\msg$. From property~(A), note that $\tagenc(\msg)$ is uniformly distributed over $\cA_\msg$ and $|\cA_\msg|=\lfloor\delta 2^{\tl}\rfloor$. Let $\cA=\cup_{\msg\in\msgset}\cA_\msg$
\item Let $\tagdec:\msgset\times\bbR^\tl\to\set{\accept,\reject}$ be a deterministic decoder such that for every $\msgh\in\msgset$, and $\vecrv\in\bbR^\tl$, 
\begin{equation}
\tagdec(\msgh,\vecrv)=\begin{cases}\accept & \mbox{if } \exists\ \hat{u}^\tl\in\cA_{\msgh}\mbox{ s.t. }\norm{\vecrv-\hat{u}^\tl}<\frac{3}{2}\sigma\sqrt{\tl}\\ \reject&\mbox{otherwise.} \end{cases}	
\end{equation}

\end{itemize}

We first note that for the code $(\tagenc,\tagdec)$ defined above, 
\begin{align}
P_{\reject|s_0}(\tagenc,\tagdec)&\leq P_{\reject}(\idenc,\iddec)+P_e(g,\gamma)\\
&\leq 0+\frac{\epsilon}{2}\\
&=\frac{\epsilon}{2}.
\end{align}
Next, that there exists a $\permchoice$ such that we show that $\pwa(\tagenc,\tagdec)\leq \epsilon$. To this end, we use the probabilistic method by first replacing $\permchoice$ by a permutation $\perm$ drawn uniformly from the set of all permutations over $\set{0,1}^{\ttl}$. As with $\permchoice$, we assume that drawn permutation $\perm$ is known to all the parties, including an adversary. We argue that, with a high probability over the choice of $\perm$, $\pwa(\tagenc,\tagdec)\leq \epsilon$. To this end, for each $\msgh\in\msgset$ and $\state^\tl\in\cS^\tl$, let $\cD(\msgh,\state^\tl)\defineqq\set{\vectu\in\cA: \norm{\vectu+\state^\tl-\hat{\tu}^\tl}<3\sigma\sqrt{\tl}\mbox{ for some }\hat{\tu}^\tl\in\cA_{\msgh}}$.

Let $\msg,\msgh\in\msgset$ with $\msg\neq\msgh$ and let $\state^\tl\in\bbR^\tl$. We have

\begin{align}
&\Prob_{\Noisevec,\tagenc}\left(\tagdec\left(\msgh,\vecrV\right)=\accept|\Msg=\msg,\State^\tl=\state^\tl\right)\\
&= \Prob_{\Noisevec,\tagenc}\left(\tagdec\left(\msgh,\tagenc(\msg)+\state^\tl+\Noisevec\right)=\accept\right)\\
&= \frac{1}{|\cA_\msg|}\sum_{\vectu\in\cA_\msg} \Prob_{\Noisevec}\left(\tagdec\left(\msgh,\vectu+\state^\tl+\Noisevec\right)=\accept\right)\\
&=\frac{1}{|\cA_\msg|}\sum_{\vectu\in\cA_\msg\cap\cD(\msgh,\state^\tl)} \Prob_{\Noisevec}\left(\tagdec\left(\msgh,\vectu+\state^\tl+\Noisevec\right)=\accept\right)\\
&\  + \frac{1}{|\cA_\msg|}\sum_{\vectu\in\cA_\msg\setminus\cD(\msgh,\state^\tl)} \Prob_{\Noisevec}\left(\tagdec\left(\msgh,\vectu+\state^\tl+\Noisevec\right)=\accept\right)\\
&\leq \frac{\lvert\cA_\msg\cap\cD(\msgh,\state^\tl)\rvert}{\lvert\cA_\msg\rvert}\\
&\  + \frac{1}{|\cA_\msg|}\sum_{\substack{\vectu\in\cA_\msg\setminus\cD(\msgh,\state^\tl)\\ \hat{u}^\tl\in\cA_{\msgh}}} \Prob_{\Noisevec}\left(\norm{\vectu+\state^\tl+\Noisevec-\hat{u}^\tl}<\frac{3}{2}\sigma\sqrt{\tl}\right)\\
\intertext{Applying Proposition~\ref{prop:gaussian confusion set} to the first term, we obtain that, with probability at least $1- 2^{-\delta^3 2^{R\tl}/4}$ over the choice of $\perm$,}
&\Prob_{\Noisevec,\tagenc}\left(\tagdec\left(\msgh,\vecrV\right)=\accept|\Msg=\msg,\State^\tl=\state^\tl\right)\\
&\leq \epsilon/2+2\delta.\\
&\ +\frac{1}{|\cA_\msg|}\smashoperator[r]{\sum_{\substack{\vectu\in\cA_\msg\setminus\cD(\msgh,\state^\tl)\\ \hat{u}^\tl\in\cA_{\msgh}}}} \Prob_{\Noisevec}\left(\norm{\vectu+\state^\tl-\hat{u}^\tl}-\norm{W^\tl}<\frac{3}{2}\sigma\sqrt{\tl}\right)\\
&\leq \epsilon/2+2\delta+\frac{1}{|\cA_\msg|}\sum_{\substack{\vectu\in\cA_\msg\setminus\cD(\msgh,\state^\tl)\\ \hat{u}^\tl\in\cA_{\msgh}}} \Prob_{\Noisevec}\left(\norm{W^\tl}>\frac{3}{2}\sigma\sqrt{\tl}\right)\label{eq:chi-square break}.
\end{align}

To bound the summand above, we note that 
\begin{align}
\shortonly{&}\Prob_{\Noisevec}\left(\norm{W^\tl}>\frac{3}{2}\sigma\sqrt{\tl}\right) \shortonly {\\} &\shortonly{\ } = \Prob_{\Noisevec}\left(\sum_{\tau=1}^\tl W_\tau^2>\frac{9}{4}\sigma^2\tl\right)\\
&\shortonly{\ }= 	\Prob_{\Noisevec}\left(\frac{1}{\tl}\sum_{\tau=1}^\tl \frac{W_\tau^2-\sigma^2}{\sigma^2}>\frac{5}{4}\right)\\
&\shortonly{\ }\leq \Prob_{\Noisevec}\left(\frac{1}{\tl}\sum_{\tau=1}^\tl \frac{W_\tau^2-\sigma^2}{\sigma^2}>1\right)\shortonly{\\ &\ }\leq 2 e^{-\tl/8}.\label{eq:chi squared}
\end{align}
The last inequality above follows from standard bounds on the tail probability of chi-squared distribution. Continuing the chain of inequalities from~\eqref{eq:chi-square break}, we bound the number of $(\vectu,\hat{u}^\tl)$ pairs in the sum by $|\cA_\msg|^2$ and use~\eqref{eq:chi squared} and the assumption that $R<1/8$ to obtain that, with probability at least $1- 2^{-\delta^3 2^{R\tl}/4}$ over the choice of $\perm$, for large enough $\tl$,
\begin{align}
\shortonly{&}\Prob_{\Noisevec,\tagenc}\left(\tagdec\left(\msgh,\vecrV\right)=\accept|\Msg=\msg,\State^\tl=\state^\tl\right)\shortonly{\\}&\shortonly{\ }\leq 	\epsilon/2+2\delta+2\delta 2^{R\tl}e^{-\tl/8}\\
&\shortonly{\ }\leq \epsilon/2+4\delta\\
&\leq 9\epsilon/10.
\end{align}

Finally, we take a union bound over all $\msg,\msgh\in\msgset$ and over a sufficiently dense set of state vectors $\state^\tl$ to conclude that, with high probability over the choice of $\perm$,  $\pwa(\tagenc,\tagdec)<\epsilon$. 
\end{proof}
In Propositions~\ref{prop:gaussian confusable} and~\ref{prop:gaussian confusion set}, we rely on the notation introduced in the proof of Lemma~\ref{lem:auth gaussian positive}.

\begin{prop}\label{prop:gaussian confusable} 
	 For every $\msgh\in\msgset$ and $\state^\tl\in\bbR^\tl$, $\lvert\cD(\msgh,\state^\tl)\rvert\leq \delta 2^{\ttl}$.
\end{prop}
\begin{proof}
Let $\vectu\in\cD(\msgh,\state^\tl)$. Thus, there exists $\hat{u}^\tl\in\cA_{\msgh}$ such that $\norm{\vectu+\state^{\tl}-\hat{u}^\tl}<3\sigma\sqrt{\tl}$. Consider any $\tilde{u}^\tl\in\cA\setminus\set{\vectu}$. By property~(D), $\norm{\tilde{u}^\tl-\vectu}\geq 7\sigma\sqrt{\tl}$. Thus,  $\norm{\tilde{u}^\tl+\state^{\tl}-\hat{u}^\tl}=\norm{\tilde{u}^\tl-\vectu+\vectu+\state^{\tl}-\hat{u}^\tl}\geq \norm{\tilde{u}^\tl-\vectu}-\norm{ \vectu+\state^{\tl}-\hat{u}^\tl}\geq 4\sigma\sqrt{\tl}$. Thus, for each $\hat{u}^\tl\in\cA_{\msg}$, there is at most one $\vectu\in\cA$ such that $\norm{\vectu-\hat{u}^\tl-\state^{\tl}}<3\sigma\sqrt{\tl}$, which implies the claim of this proposition.
\end{proof}
\begin{prop}  \label{prop:gaussian confusion set}Let $\msg,\msgh\in\msgset$ with $\msg\neq\msgh$ and $\state^\tl\in\bbR^\tl$. With probability at least $1- 2^{-\delta^3 2^{R\tl}/4}$ over the random choice of $\perm$, $\left|\cD(\msgh,\state^\tl)\cap \cA_{\msg}\right|\leq(\epsilon/2+2\delta)\delta 2^{\ttl}$.
\end{prop}
\begin{proof}
We first fix the set $\cA_{\msgh}$. Note that this also fixes the set $\cA_{\msg}\cap\cA_{\msgh}$ as the image of the set $\cB_{\msg}\cap\cB_{\msgh}$ under the encoder mapping. Thus, 
\begin{equation}\left|\cD(\msgh,\state^\tl)\cap \cA_{\msg}\cap\cA_{\msgh}\right|\leq \left|\cA_{\msg}\cap\cA_{\msgh}\right|\leq (\epsilon\delta/2) 2^{\ttl}.\label{eq:intersection}\end{equation} In the above, the last inequality follows from properties~(A) and~(B) in the proof of Lemma~\ref{lem:auth gaussian positive}.

Next, we note that, conditioned on fixed choice of  the set $\cA_{\msgh}$ and under the uniformly random choice of $\perm$, the process of drawing the set $\cA_{\msg}\setminus\cA_{\msgh}$ is equivalent to successively sampling $K\defineqq\lvert\cA_{\msg}\setminus\cA_{\msgh}\rvert$ elements $\vectU(1), \vectU(2),\ldots, \vectU(k)$ uniformly without replacement from  the set $\cA\setminus\cA_{\msgh}$. Note that, for large enough $\ttl$, 
\begin{equation}\delta 2^{\ttl}/2\overset{(a)}\leq\delta (1-\epsilon/2)2^{\ttl}-1\overset{(b)}\leq K\overset{(c)}\leq \delta 2^{\ttl}.\label{eq:bound on K}\end{equation} 
In the above, $(a)$ is due to $\ttl$ being large enough, and $(b)$ and $(c)$ follow from properties~(A) and~(B) of the code $(\idenc,\iddec)$ chosen in the proof of Lemma~\ref{lem:auth gaussian positive}. Next, note that, for any $k\in[1:K]$ and large enough $\ttl$, we have

\begin{align} 
&p\defineqq\Prob_{\perm}\left(\vectU(k)\in \cD(\msgh,\state^\tl)\right) \shortonly{\\&\quad} = \frac{\lvert\cD(\msgh,\state^\tl)\rvert}{\lvert\cA\setminus\cA_{\msgh}\rvert}\\
	&\quad \leq \frac{\delta 2^{\ttl}}{2^{\ttl}-\lfloor\delta 2^{\ttl}\rfloor}\shortonly{\\&\quad}\leq \frac{\delta}{1-\delta-2^{-\ttl}}\leq \frac{3}{2}\delta.
\end{align}

Thus, 
\begin{align}
& \Prob_\perm\left(\left|\cD(\msgh,\state^\tl)\cap \left(\cA_{\msg}\setminus\cA_{\msgh}\right)\right|>2\delta^2 2^{\ttl}\right)\shortonly{\\ & \quad} \leq    \Prob_\perm\left(\left|\cD(\msgh,\state^\tl)\cap \left(\cA_{\msg}\setminus\cA_{\msgh}\right)\right|>pK+\frac{1}{2}\delta^2 2^{\ttl}\right)	\\
& \quad =\Prob_\perm\left(\frac{1}{K}\sum_{k=1}^K\mathbf{1}_{\set{\vectU(k)\in \cD(\msgh,\state^\tl)}}-p>\frac{\delta^2 2^{\ttl}}{2K}\right)\\
& \quad \leq \Prob_\perm\left(\frac{1}{K}\sum_{k=1}^K\mathbf{1}_{\set{\vectU(k)\in \cD(\msgh,\state^\tl)}}-p>\frac{\delta}{2}\right).\label{eq:concentration without replacement}
\end{align}
Note that the probability term in~\eqref{eq:concentration without replacement} is the probability that the empirical average of the random variables $\mathbf{1}_{\set{\vectU(k)\in \cD(\msgh,\state^\tl)}}$, drawn without replacement from a multi-set of size $K$ containing $pK$ ones and $(1-p)K$ zeroes, deviates from its mean. By~\cite[Theorem~4]{Hoeffding63}, this probability is bounded from above by the corresponding probability when each $\mathbf{1}_{\set{\vectU(k)\in \cD(\msgh,\state^\tl)}}$ is drawn with replacement from the same multi-set. By~\cite[Theorem~1]{Hoeffding63} and~\eqref{eq:bound on K}, this is bounded from above by $2^{-K\delta^2/2}\leq 2^{-\delta^3 2^{\ttl}/4}$ for large enough values of $\ttl$. Combining the above with the bound obtained in~\eqref{eq:intersection}, we prove the claim of the proposition. \end{proof}

\subsection{Authenticated communication capacity for Gaussian AVCs}
Using Lemma~\ref{lem:auth gaussian positive}, we argue that the authenticated communication capacity of Gaussian AVCs is same as the unauthenticated communication capacity.
\begin{proof}[Proof sketch for Theorem~\ref{thm:auth comm gaussian}]
The converse follows from the channel capacity converse for the nominal channel, which is an additive Gaussian noise channel with input power constraint $\rho$ and noise variance $\sigma^2$. The achievability follows by using a two phase code formed by concatenating a reliable communication code for the nominal channel with an authentication code of positive rate for the additive Gaussian AVC from Lemma~\ref{lem:auth gaussian positive}.
\end{proof}

\subsection{Authentication capacity for Gaussian AVCs}
Next, using Theorem~\ref{thm:auth comm gaussian}, we argue that the authentication capacity of Gaussian AVCs is same as the unauthenticated communication capacity.
\begin{proof}[Proof sketch for Theorem~\ref{thm:auth gaussian}] The proof is similar lines as Theorem~\ref{thm: auth capacity}. The converse follows by observing that the authentication capacity for the given AVC is bounded from above by the identification capacity of the nominal channel, \emph{i.e.}, the additive Gaussian noise channel with input power constraint $\rho$ and noise variance $\sigma^2$. By~\cite{BurnashevIT00}, the latter equals $({1}/{2})\log\left(1+{\rho}/{\sigma^2}\right)$.
Similarly, the achievability follows by composing an identification code for the noiseless channel (from~\cite{AhlswedeDIT89}) with an authenticated communication  code for the Gaussian AVC from Theorem~\ref{thm:auth comm gaussian}.
\end{proof}

\subsection{Authenticated communication capacity region for Gaussian MACs}
Finally, using Theorem~\ref{thm:auth gaussian}, we show that the authenticated capacity region of Gaussian MACs is same as the unauthenticated communication capacity.

\begin{proof}[Proof sketch for Theorem~\ref{thm:auth comm gaussian mac}] The proof is similar lines as Theorem~\ref{thm:authMAC}. The converse follows by observing that the authenticated communication capacity region cannot be a superset of the (unauthenticated communication) capacity region. The achievability follows from the same 3 phase idea as the the proof of Theorem~\ref{thm:authMAC}, \emph{i.e.}, first transmit the messages using a (unauthenticated) code for the Gaussian MAC. Next, in the second and third phases, the users take turns transmitting authentication codes for their respective messages while a non-adversarial second user stays silent (\emph{i.e.}, sends all $0$'s).
\end{proof}

}
\section{Equivalence of capacity regions of adversarial DM-MAC for codes with stochastic encoders under maximum error probability and for deterministic codes under average error probability}\label{App:stocDet}
In this section, we will argue that the capacity region of the \mac using codes with stochastic encoder (stochastic codes) under maximum probability of error is the same as the capacity region using deterministic codes under average probability of error criterion. Similar to Section~\ref{section:MAC_model}, we define the average probability of error for a deterministic code $(f_1,f_2,\phi)$ in terms of error probabilities under \na and attack conditions as follows. Let $x^n\in \mathcal{X}^n$ be the sequence sent by user 1 when it (user 1) is adversarial and $y^n\in \mathcal{Y}^n$ be the sequence sent by user 2 when it (user 2) is adversarial. Let $(\hat{M}_1,\hat{M}_2) = \phi(Z^n)$. For $m_1\in \mathcal{M}_1, m_2\in\mathcal{M}_2$,
\begin{align*}
&P_{e,\na,m_1,m_2}(f_1,f_2,\phi) \defineqq \Prob((\hat{M}_1,\hat{M}_2)\neq\left(m_1,m_2\right)|(M_1,M_2) = (m_1,m_2))\\
&P_{e,\malone, m_2, x^n}(f_2,\phi) \defineqq \Prob\left(\hat{M}_2\notin\left\{m_2,\bot\right\}\left|M_2 = m_2,X^n =x^n\right.\right),\\
&P_{e,\maltwo, m_1, y^n}(f_1,\phi) \defineqq \Prob\left(\hat{M}_1\notin\left\{m_1,\bot\right\}\left|M_1 = m_1,Y^n =y^n\right.\right),
\end{align*}
where the probabilities are over the channel.
We define the average probability of error under attack and no attack conditions as follow.
\begin{align*}
&\bar{P}_{e,\na}(f_1,f_2,\phi) \defineqq \frac{\sum_{m_1\in\mathcal{M}_1, m_2\in\mathcal{M}_2}P_{e,\na,m_1,m_2}(f_1,f_2,\phi)}{N_1 N_2},\\
&\bar{P}_{e,\malone}(f_2,\phi) \defineqq \max_{x^n\in\mathcal{X}^n}\frac{\sum_{m_2\in\mathcal{M}_2}P_{e,\malone, m_2, x^n}(f_2,\phi)}{N_2},\\
&\bar{P}_{e,\maltwo}(f_1,\phi) \defineqq \max_{y^n\in \mathcal{Y}^n}\frac{\sum_{m_1\in\mathcal{M}_1}P_{e,\maltwo, m_1, y^n}(f_1,\phi)}{N_1}. 
\end{align*}
The average probability of error of the deterministic code $(f_1,f_2,\phi)$ is defined as
\begin{align*}
\bar{P}_{e}(f_1,f_2,\phi)&=\max\,\left\{\bar{P}_{e,\na}(f_1,f_2,\phi),\bar{P}_{e,\malone}(f_2,\phi),\bar{P}_{e,\maltwo}(f_1,\phi)\right\}.	
\end{align*}
The average probability of error for a stochastic code $(F_1,F_2,\phi)$ can also be defined in an analogous fashion. We will use $\bar{P}_{e}(F_1,F_2,\phi)$ to refer to this average probability of error.
\begin{lemma}\label{lemma:Achiev}Any rate pair achievable using stochastic codes under maximum probability of error criterion can also be achieved using deterministic codes under average probability of error criterion. 
\end{lemma}
\begin{proof}
The proof is along the lines of~\cite{Ahlswede78}~\cite[Theorem~12.13]{CK11}.\\
For any rate pair $(R_1,R_2)\geq0,$ let $\left(F_1,F_1,\phi\right)$ be a $\left(2^{nR_1},2^{nR_2},n\right)$ stochastic code with $P_{e}\left(F_1,F_1,\phi\right)\leq \epsilon$. This also implies that $\bar{P}_{e}\left(F_1,F_1,\phi\right)\leq \epsilon$. Define random variables $\mathbf{X}_{m_1},  m_1\in \mathcal{M}_1$ over $\mathcal{X}^n$ and $\mathbf{Y}_{m_2}, m_2\in \mathcal{M}_2$ over $\mathcal{Y}^n$, all independent of each other, as $\Prob(\mathbf{X}_{m_1} = x^n) = \Prob(F_1(m_1) = x^n)$ and $\Prob(\mathbf{Y}_{m_2} = y^n) = \Prob(F_2(m_2) = y^n)$.
For $(m_1,m_2)\in\mathcal{M}_1\times\mathcal{M}_2$, define a random variable $V_{m_1,m_2}$ as a function of the random variables $\mathbf{X}_{m_1}$ and $\mathbf{Y}_{m_2}$ as follows. 
\begin{align*}
V_{m_1,m_2} \defineqq \psi_{m_1,m_2}(\mathbf{X}_{m_1},\mathbf{Y}_{m_2}),
\end{align*}
where for $\tilde{x}^n\in\mathcal{X}^n, \tilde{y}^n\in\mathcal{Y}^n,$
\begin{align*}
\psi_{m_1,m_2}(\tilde{x}^n,\tilde{y}^n) \defineqq \mathbb{P}\left(\phi\left(Z^n\right)\neq\left(m_1,m_2\right)\left|X^n = \tilde{x}^n,Y^n =\tilde{y}^n\right.\right), 
\end{align*}
The probability above is over the channel.
For $(m_1,m_2)\in\mathcal{M}_1\times\mathcal{M}_2$ and  $ x^n\in\mathcal{X}^n$ and $y^n \in \mathcal{Y}^n$, we also define the random variables $V_{m_2,x^n}$ and $V_{m_1,y^n}$ as functions of the random variables $\mathbf{X}_{m_1}$ and $\mathbf{Y}_{m_2}$ as follows.
\begin{align*}
V_{m_2,x^n} = \psi_{m_2,x^n}(\mathbf{Y}_{m_2}),\\
V_{m_1,y^n} = \psi_{m_1,y^n}(\mathbf{X}_{m_1}),
\end{align*}
where for $\tilde{x}^n\in\mathcal{X}^n$ and $\tilde{y}^n\in\mathcal{Y}^n$,
\begin{align*}
\psi_{m_2,x^n}(\tilde{y}^n) \defineqq \mathbb{P}\left(\hat{M}_2\notin\left\{m_2,\bot\right\}\left|X^n =x^n,Y^n = \tilde{y}^n\right.\right),\\
\psi_{m_1,y^n}(\tilde{x}^n) \defineqq \mathbb{P}\left(\hat{M}_1\notin\left\{m_1,\bot\right\}\left|X^n = \tilde{x}^n,Y^n =y^n\right.\right).
\end{align*}
The probability above is over the channel.
Let $a>0$, and consider the probability of the event $\left\{\frac{1}{ N_1  N_2  } \sum_{m_1,m_2} V_{m_1,m_2}\geq a \epsilon\right\}$ where the probability is over the randomness of $\mathbf{X}_{m_1},\,m_1\in\mathcal{M}_1$ and $\mathbf{Y}_{m_2},\,m_2\in \mathcal{M}_2$.
\begin{align*}
\mathbb{P}\left\{\frac{1}{ N_1  N_2  } \sum_{m_1,m_2} V_{m_1,m_2}\geq a \epsilon\right\} &= \mathbb{P}\left\{2^{\sum_{m_1,m_2}V_{m_1,m_2}}\geq 2^{a \epsilon N_1  N_2  }\right\}\\
&\leq 2^{- a \epsilon  N_1  N_2  }\prod_{m_1,m_2}E\left[2^{V_{m_1,m_2}}\right]\\
&\stackrel{\text{(a)}}{\leq}2^{-a \epsilon N_1  N_2  }\prod_{m_1,m_2}3^{E\left[V_{m_1,m_2}\right]}\\
&=2^{-a \epsilon N_1  N_2  +\log_{2}3\left(\sum_{m_1,m_2}E\left[V_{m_1,m_2}\right]\right)}\\
\end{align*}
Here, (a) follows by noting that for $0\leq t\leq 1, 2^t\leq 1+t \leq 3^t$. Since $\bar{P}_{e}\left(F_1,F_1,\phi\right)\leq \epsilon$, we have
$\frac{1}{ N_1  N_2  }\sum_{m_1,m_2} E\left[V_{m_1,m_2}\right]\leq\epsilon.$ Thus,
\begin{align*}
\mathbb{P}\left\{\frac{1}{ N_1  N_2  }\sum_{m_1,m_2}V_{m_1,m_2}\geq a \epsilon\right\}&\leq 2^{- N_1  N_2  \epsilon\left(a-\log_{2}3\right)}\\
&=2^{-2^{n(R_1+R_2)}\epsilon(a-\log_{2}3)}
\end{align*}
Following the similar line of arguments, we can upper bound the following probability terms.
\begin{align*}
\mathbb{P}\left\{\frac{1}{N_2  }\sum_{m_2}V_{m_2,x^n}\geq a \epsilon\right\}\leq2^{-2^{nR_2}\epsilon(a-\log_{2}3)}\\
\mathbb{P}\left\{\frac{1}{ N_1 }\sum_{m_1}V_{m_1,y^n}\geq a \epsilon\right\}\leq2^{-2^{nR_1}\epsilon(a-\log_{2}3)}.
\end{align*} 
As before the probability in the above terms is over the randomness of $\mathbf{X}_{m_1},\,m_1\in\mathcal{M}_1$ and $\mathbf{Y}_{m_2},\,m_2\in \mathcal{M}_2$. Thus,
\begin{align*}
\mathbb{P}\left(\left\{\frac{1}{N_2 }\sum_{m_2}V_{m_2,x^n}\geq a \epsilon \text{ for some $x^n \in \mathcal{X}^n$}\right\} \right.&\left. \cup \left\{\frac{1}{ N_1 }\sum_{m_1}V_{m_1,y^n}\geq a \epsilon \text{ for some $y^n \in \mathcal{Y}^n$}\right\}\cup \left\{\frac{1}{ N_1  N_2  }\sum_{m_1,m_2}V_{m_1,m_2}\geq a \epsilon\right\}\right)  \\
&\stackrel{\text{(a)}}{\leq} |\mathcal{X}|^{n}2^{-2^{nR_2}\epsilon(a-\log_{2}3)}+|\mathcal{Y}|^{n} 2^{-2^{nR_1}\epsilon(a-\log_{2}3)} + 2^{-2^{n(R_1+R_2)}\epsilon(a-\log_{2}3)}\\
&\leq \max{\left\{(2|\mathcal{X}|+1)^n,(2|\mathcal{Y}|+1)^n\right\}}2^{-2^{n\min{\{R_1,R_2\}}}\epsilon(a-\log_{2}3)}\\
&< 1, \text{ for $a=2$ and large enough }n.
\end{align*}
where (a) follows from union bound. Thus, there is some realization of the random variables $\mathbf{X}_{m_1},\mathbf{Y}_{m_2},m_1\in \mathcal{M}_1,m_2\in \mathcal{M}_2$ such that the average error probability is bounded by $2\epsilon$.
\end{proof}
Let $R^{\det}_{\authcomm}$ be used to refer to the capacity region of the \mac using deterministic codes under average probability of error criterion.
\begin{lemma}\label{lemma:converse}
If the \mac is not overwritable by both users, then $\mathcal{R}^{\det}_{\authcomm} \subseteq \mathcal{R}_{\nonadv}$. Otherwise, if user 1 can overwrite user 2 but user 2 cannot overwrite user 1, $\mathcal{R}^{\det}_{\authcomm} \subseteq \{(R,0),(R,0)\in \mathcal{R}_{\nonadv}\}$. When user 2 can overwrite user 1 but user 1 cannot overwrite user 2, $\mathcal{R}^{\det}_{\authcomm} \subseteq \{(0,R),(0,R)\in \mathcal{R}_{\nonadv}\}$. When both users can overwrite each other, $\mathcal{R}^{\det}_{\authcomm} = \{(0,0)\}.$
\end{lemma}
\begin{proof}
We first consider the case when one of the users, say user 1, can overwrite the other user while user 2 is honest. Suppose some rate pair $(R_1,R_2),\, R_1>0, R_2>0$ is achievable. Consider the adversarial strategy $Q_{X^n}(x^n)$ where user 1 overwrites every sequence sent out by user 2 by an encoding of a uniformly chosen message $M_2$ from $\mathcal{M}_2$ while pretending to send an encoding of a uniformly chosen message $M_1$ from $m_1\in\mathcal{M}_1$ i.e., the output distribution of user 2, $Q_{X^n}(x^n) = \frac{1}{N_1 N_2}\sum_{m'_1\in \mathcal{M}_1,m'_2\in\mathcal{M}_2}\prod_{i=1}^{n}\left(\sum_{x\in\mathcal{X}}P_{X'|X,Y}(x|f_1(m'_1)_i,f_2(m'_2)_i)\right)$ where $f_1(m'_1)_i$ and $f_2(m'_2)_i$ are $i^{\text{th}}$ elements of $f_1(m'_1)$ and $f_2(m'_2)$ sequences respectively and $P_{X'|X,Y}$ satisfies \eqref{eq1}. For $m_1\in \mathcal{M}_1$ and $m_2\in \mathcal{M}_2$, let $\phi^{-1}(m_1,m_2) = \{z^n\,:\,\phi(z^n) = (m_1,m_2)\}$ and $\phi_2^{-1}(m_2) = \{z^n\,:\,\phi(z^n) = (\tilde{m}_1,m_2), \tilde{m}_1\in \mathcal{M}_1\}$ and $\phi_2^{-1}(\bot) = \{z^n\,:\,\phi(z^n) = (\bot, \bot)\}$. Then,
\begin{align*}
\bar{P_{e}}&\left(f_1,f_1,\phi\right)\geq {\frac{1}{N_2}\sum_{x^n\in\mathcal{X}^n,m_2\in\mathcal{M}_2}Q_{X^n}(x^n)P_{e,\malone, m_2, x^n}(f_2,\phi)}\\
&= \frac{1}{N_2}\sum_{m_2}\frac{1}{N_1 N_2}\sum_{m'_1\in \mathcal{M}_1,m'_2\in\mathcal{M}_2}\prod_{i=1}^{n}\left(\sum_{x\in\mathcal{X}}P_{X'|X,Y}(x|f_1(m'_1)_i,f_2(m'_2)_i)\right)\Prob\left(\hat{M}_2 \notin \{m_2,\bot\}|X^n = x^n, Y^n = f_2(m_2)\right)\\
&=  \frac{1}{N_2}\sum_{m_2}\frac{1}{N_1 N_2}\sum_{m'_1\in \mathcal{M}_1,m'_2\in\mathcal{M}_2}\sum_{z^n\notin\phi_2^{-1}(m_2)\cup\phi_2^{-1}(\bot)}\prod_{i=1}^{n}\left(\sum_{x\in\mathcal{X}}P_{X'|X,Y}(x|f_1(m'_1)_i,f_2(m'_2)_i)W\left(z_i|x_i,f_2(m_2)_i\right)\right)\\
&\stackrel{\text{(a)}}{=}  \frac{1}{N_2}\sum_{m_2}\frac{1}{N_1 N_2}\sum_{m'_1\in \mathcal{M}_1,m'_2\in\mathcal{M}_2}\sum_{z^n\notin\phi_2^{-1}(m_2)\cup\phi_2^{-1}(\bot)}\prod_{i=1}^{n}W\left(z_i|f_1(m'_1)_i,f_2(m'_2)_i\right)\\
&\geq \frac{1}{N_2 }\sum_{m_2}\frac{1}{N_1 N_2}\sum_{m'_1, m'_2\neq m_2}\sum_{z^n\notin\phi_2^{-1}(m_2)\cup\phi_2^{-1}(\bot)}\prod_{i=1}^{n}W\left(z_i|f_1(m'_1)_i,f_2(m'_2)_i\right)\\
&\stackrel{\text{(b)}}{\geq} \frac{1}{N_2 }\sum_{m_2}\frac{1}{N_1 N_2}\sum_{m'_1, m'_2\neq m_2}\sum_{z^n\in\phi^{-1}(m'_1,m'_2)}\prod_{i=1}^{n}W\left(z_i|f_1(m'_1)_i,f_2(m'_2)_i\right)\\
&= \frac{1}{N_1 N_2}\sum_{m'_1, m'_2}\sum_{z^n\in\phi^{-1}(m'_1,m'_2)}\prod_{i=1}^{n}W\left(z_i|f_1(m'_1)_i,f_2(m'_2)_i\right)\sum_{m_2\neq m'_2}\frac{1}{N_2}\\
&=(1-\bar{P}_{e,\na}(f_1,f_2,\phi))\frac{N_2-1}{N_2},
\end{align*}

where (a) follows from \eqref{eq1} and (b) follows by observing that for $m'_2\neq m_2, \,\phi^{-1}(m'_1,m'_2) \subseteq \phi_2^{-1}(m'_2)\subseteq  \mathcal{Z}^n\setminus(\phi_2^{-1}(m_2)\cup\phi_2^{-1}(\bot))$. The above lower bound implies that $\bar{P_e}\left(f_1,f_1,\phi\right)N_2 + (N_2-1)\bar{P}_{e,\na}(f_1,f_2,\phi) \geq N_2-1$. We further observe that 
\begin{align*}
\bar{P_e}\left(f_1,f_1,\phi\right)(2N_2 - 1) \geq \bar{P_e}\left(f_1,f_1,\phi\right)N_2 + (N_2-1)\bar{P}_{e,\na}(f_1,f_2,\phi) \geq N_2-1.
\end{align*}
and thus, $\bar{P_e}\left(f_1,f_1,\phi\right) \geq \frac{N_2-1}{2N_2-1}$ which is bounded away from zero for all $N_2>1$.
Similarly, we can argue that when sender 2 can overwrite sender 1, the average error is bounded away from 0 for all $N_1>1$. This also implies that when both users can overwrite each other, $\mathcal{R}^{\det}_{\authcomm} = \{(0,0)\}.$ 
\end{proof}
\begin{lemma}The capacity region of the adversarial multiple access channel \mac using deterministic codes under average probability of error criterion is same as the capacity region using codes with stochastic encoder using maximum probability of error criterion.
\end{lemma}
\begin{proof}The achievability follows from Lemma~\ref{lemma:Achiev} and the converse follows from Lemma~\ref{lemma:converse}.
\end{proof}

\section{Authenticated Communication over AVCs}\label{App:AVC_stoc_det}
In this section, we revisit the authenticated communication setup of~\cite{KosutKITW18} and show that any rate achievable using deterministic codes under average probability of error criterion is also achievable using codes with stochastic encoder under maximum probability of error criterion. Consider the system model given in Fig.~\ref{fig:authcomAVC}.

\begin{figure}\centering
\begin{tikzpicture}[scale=0.5]
	\draw (2,0) rectangle ++(3,2) node[pos=.5]{$f^{(n)}$};
	\draw (10,0) rectangle ++(3,2) node[pos=.5]{$W_{V|U,S}$};
	\draw (16,0) rectangle ++(3,2) node[pos=.5]{$\phi^{(n)}$};
	\draw[->] (1,1) node[anchor=east]{$\msg$} -- ++ (1,0) ;
	\draw[->] (5,1) -- node[above] {$U^n$} ++ (5,0);
	\draw[->] (13,1) -- node[above] {$V^n$} ++ (3,0);
	\draw[->] (19,1) -- ++ (1,0) node[anchor=west]{$\msgh/\bot$};
	\draw[->] (11.5,-1)node[anchor=north]{$S^n$} -- ++ (0,1);
\end{tikzpicture}
\caption{Authenticated communication over an AVC $W_{V|U,S}$}\label{fig:authcomAVC}
\end{figure}
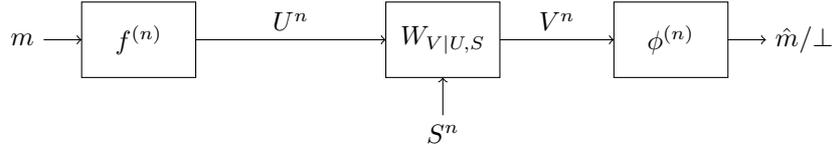

 \begin{defn}
An $(N,n)$ {\em deterministic authenticated communication code} consists of the following:
\begin{enumerate}
\item A message set $\mathcal{M} = \{1,2,\ldots,N\}$,
\item A deterministic encoder $f^{(n)}:\mathcal{M}\rightarrow\mathcal{U}^n$,
\item A deterministic decoder $\phi^{(n)}:\mathcal{V}^n\rightarrow \mathcal{M}\cup\{\bot\}$.
\end{enumerate}
\end{defn}
For $m\in\mathcal{M}$ and $s^n\in\mathcal{S}^n$, define probability of error as 
\begin{displaymath}
e_m^{\det}\left(f^{(n)},\phi^{(n)},s^n\right) = 
\left\{\begin{array}{ll}
        \mathbb{P}\left(\phi^{(n)}\left(V^n\right)\neq m\left|U^n=f^{(n)}(m), S^n = s_0^n\right.\right), &\, s^n= s_0^n,\\
        \mathbb{P}\left(\phi^{(n)}\left(V^n\right)\notin \left\{m,\bot\right\}\left|U^n=f^{(n)}(m), S^n = s^n\right.\right),&s^n\neq s_0^n,
    \end{array}
\right.
\end{displaymath} 
and let
\begin{align*}
e_m^{\det}\left(f^{(n)},\phi^{(n)}\right) = \max_{s^n\in\mathcal{S}^n} e^{\det}_m\left(f^{(n)},\phi^{(n)},s^n\right).   
\end{align*}
The average probability of error $\bar{e}^{\det}\left(f^{(n)},\phi^{(n)}\right)$ is defined as 
\begin{align*}
\bar{e}^{\det}\left(f^{(n)},\phi^{(n)}\right) = \max_{s^n\in\mathcal{S}^n} \frac{1}{ N }\sum_{m\in\mathcal{M}}e^{\det}_m\left(f^{(n)},\phi^{(n)},s^{n}\right).
\end{align*}
The maximum probability of error $e^{\det}\left(f^{(n)},\phi^{(n)}\right)$ is defined as 
\begin{align*}
e^{\det}\left(f^{(n)},\phi^{(n)}\right) = \max_{m\in \mathcal{M}}e^{\det}_m\left(f^{(n)},\phi^{(n)}\right).
\end{align*}
An $(N,n)$ random code $\left(F^{(n)},\Phi^{(n)}\right)$ and an $(N,n)$ stochastic code $\left(F^{(n)},\phi^{(n)}\right)$ can be defined in a similar manner. The average error and the maximum error for a random code will be denoted by $\bar{e}^{\rand}\left(F^{(n)},\Phi^{(n)}\right)$ and $e^{\rand}\left(F^{(n)},\Phi^{(n)}\right)$ respectively. The average error and the maximum error for a code with stochastic encoder will be denoted by $\bar{e}\left(F^{(n)},\phi^{(n)}\right)$ and $e\left(F^{(n)},\phi^{(n)}\right)$ respectively. An authenticated communication rate $R$ is achievable using deterministic codes under average probability of error criterion if there exists a sequence of codes $\left\{f^{(n)},\phi^{(n)}\right\}_{n=1}^{\infty}$ such that $\lim_{n\rightarrow\infty}\bar{e}^{\det}\left(f^{(n)},\phi^{(n)}\right)\rightarrow 0$. We define $C^{\det}_{\authcomm}$ as the supremum of all achievable rates. We can give an analogous definition of authentication communication capacity (denoted by $C_{\authcomm}$) using stochastic codes under maximum probability of error criterion. Let $C(s_0)$ denote the no-adversary capacity of the adversarial channel $W_{V|U,S}$ defined as $C(\state_0)\defineqq\max_{P_\tU}I(U;V|S=\state_0)$.
The main result of~\cite{KosutKITW18} is 
\begin{thm}\label{Thm:KK}~\cite[Theorem~2]{KosutKITW18}
If the channel is non-overwritable (i.e., \eqref{eq:overwritable} does not hold), $C^{\det}_{\authcomm}=C(s_0)$; if it is overwritable, then $C^{\det}_{\authcomm}=0$.
\end{thm}
\noindent In this section, we will show an analogous result for stochsatic codes under maximum probability of error criterion.

\begin{prop} \label{prop:maxauthcom}
If the channel is non-overwritable (i.e., \eqref{eq:overwritable} does not hold), $C_{\authcomm}=C(s_0)$; if it is overwritable, then $C_{\authcomm}=0$.
\end{prop} 
\begin{proof}
{\em Converse:}Any good stochastic code (under maximum probability of error criterion) for the adversarial channel $W_{V|U,S}$ is also a good code for the nominal channel $W_{V|U,S}(.|.,s_0)$. Thus, $C_{\authcomm}\leq C(s_0)$. When the channel is overwritable, suppose some $R>0$ is achievable. For $u^t\in \mathcal{U}^t$ and $m\in \mathcal{M}$, define the distribution $Q_{U^t|M}(u^t|m)\defineqq\Prob(F^{(n)}(m)=u^t)$.  Let $m,m'\in \mathcal{M}, m\neq m'$ be two distinct messages. Consider the adversarial strategy  $P_{S^t}(s^t) = \sum_{\tilde{u}^t\in\mathcal{U}^t}Q_{U^t|M}(\tilde{u}^t|m')\prod_{i=1}^{t}P_{S|U'}(s_i|\tilde{u}_i)$. Here, $s_i$ and $\tilde{u}_i$ are the $i^{\text{th}}$ elements of the $s^n$ and $\tilde{u}^n$ sequences respectively and $P_{S|U'}$ satisfies \eqref{eq:overwritable}.

Let $e_{m\rightarrow m'}\left(F^{(n)},\phi^{(n)}\right)$ be the probability that the decoder outputs $m'$ under the adversarial attack described above when $m$ was actually sent  and $e_{m,s_0}\left(F^{(n)},\phi^{(n)}\right)$ be the probability with which the decoder outputs some message $\tilde{m}\in \mathcal{M}\setminus\{m\}$ or $\bot$, when $m$ was actually sent and there was no attack. For $m\in \mathcal{M}$, let ${\phi^{(n)}}^{-1}(m) \defineqq \{v^t\,:\,\phi^{(n)}(v^t)=m\}$. Then, 
\begin{align*}
e_{m\rightarrow m'}\left(F^{(n)},\phi^{(n)}\right) &= \sum_{u^t\in\mathcal{U}^t}Q_{U^t|M}(u^t|m)\sum_{s^t\in \mathcal{S}^t}P_{S^t}(s^t)\sum_{v^t\in{\phi^{(n)}}^{-1}(m')}\prod_{i=1}^t W(v_i|u_i,s_i)\\
&=\sum_{u^t\in\mathcal{U}^t,v^t\in{\phi^{(n)}}^{-1}(m')}Q_{U_t|M}(u^t|m)\sum_{\tilde{u}^t\in\mathcal{U}^t}Q_{U^t|M}(\tilde{u}^t|m')\sum_{s^t\in \mathcal{S}^t}\prod_{i=1}^{t}P_{S|U'}(s_i|\tilde{u}_i) W(v_i|u_i,s_i)\\
&\stackrel{\text{(a)}}{=}\sum_{u^t\in\mathcal{U}^t,v^t\in{\phi^{(n)}}^{-1}(m')}Q_{U^t|M}(u^t|m)\sum_{\tilde{u}^t\in\mathcal{U}^t}Q(\tilde{u}^t|m')\prod_{i=1}^{t}W(v_i|\tilde{u}_i,s_0)\\
&=\sum_{\tilde{u}^t\in\mathcal{U}^t,v^t\in{\phi^{(n)}}^{-1}(m')}Q_{U^t|M}(\tilde{u}^t|m')\prod_{i=1}^{t}W(v_i|\tilde{u}_i,s_0)\\
&=1-e_{m',s_0}\left(F^{(n)},\phi^{(n)}\right),
\end{align*}
where (a) follows from \eqref{eq:overwritable}.
This implies that $e_{m\rightarrow m'}\left(F^{(n)},\phi^{(n)}\right) + e_{m',s_0}\left(F^{(n)},\phi^{(n)}\right)=1$. Hence, $2e\left(F^{(n)},\phi^{(n)}\right)\geq1.$ This  implies that $\nummsg\leq 1$ for any $(N,n)$-stochastic authenticated communication code with $e\left(F^{(n)},\phi^{(n)}\right)<\frac{1}{2}$. Therefore, $C_{\authcomm}$ is zero.

The other part of this proposition follows from Theorem~\ref{Thm:KK} and Lemma~\ref{lemma:authCOMavgTOstoc} (stated below).
\end{proof} 
\begin{lemma}\label{lemma:authCOMavgTOstoc}
Any rate $R>0$ of authenticated communication achievable using deterministic codes under average probability of error criterion can also be achieved using codes with stochastic encoder under maximum probability of error criterion. 
\end{lemma}
\begin{proof}
This proof borrows ideas from~\cite{Ahlswede78}~\cite[Theorem~12.13]{CK11}. We will prove the lemma in three steps. Firstly, given a positive rate deterministic code under average error probability criterion, we will construct a random code (i.e., a code with shared randomness at encoder and decoder unknown to the adversary) that achieves the same rate under maximum error probability criterion. Then, we will use a random code reduction argument along the lines of~\cite{AW69}~\cite[Lemma 12.8]{CK11} to reduce the given random code to another random code with support only over polynomially many deterministic codes. Lastly, we will use a positive rate deterministic code under average error probability criterion to share this small amount of randomness required with the decoder and use this to construct a stochastic code. 

Suppose $R>0$ is achievable using deterministic codes under average probability of error criterion. Then, let $(f^{(n)},\phi^{(n)})$ be a $(2^{nR},n)$ deterministic code with
\begin{align*} 
\bar{e}^{\det}\left(f^{(n)},\phi^{(n)}\right) \leq \epsilon.
\end{align*}
Define a $(2^{nR},n)$ random code $(F^{(n)}, \Phi^{(n)})$ as follows. Let $\pi_n$ be the set of all permutations of $\mathcal{M}$. Let $\sigma \sim \text{Unif}(\pi_n)$. The encoder is given by
\begin{align*}
F^{(n)}(m) &= f^{(n)}\left(\sigma\left(m\right)\right),\, m\in \mathcal{M}.
\end{align*}
For $ v^n\in \mathcal{V}^n$, the decoder is defined as
\begin{displaymath}
\Phi^{(n)}\left(v^n\right) = 
\left\{\begin{array}{ll}
        \sigma^{-1}\left(\phi^{(n)}\left(v^n\right)\right),&\phi^{(n)}\left(v^n\right)\in \mathcal{M},\\
        \bot, &\phi^{(n)}\left(v^n\right)= \bot.  \end{array}
\right.
\end{displaymath} 
For $m \in \mathcal{M},$ probability of error for this random code is
\begin{align*}
e^{\rand}_m\left(F^{(n)},\Phi^{(n)}\right) = \frac{1}{ N }\sum_{m'\in\mathcal{M}}e^{\det}_{m'}\left(f^{(n)},\phi^{(n)}\right) \leq \epsilon.
\end{align*}
Thus, 
\begin{align*}
e^{\rand}\left(F^{(n)},\Phi^{(n)}\right) &= \max_{m}{e^{\rand}_m\left(F^{(n)},\Phi^{(n)}\right)}\\
&\leq \epsilon.
\end{align*}
We will use the following {\em random code reduction} lemma which will be proved later. 
\begin{lemma}\label{lemma:random_code_reduction}
    Let $(F^{(n)},\Phi^{(n)}) \sim Q$ be a random code. Assume $|\mathcal{S}| < \infty$. Then, for every
    \begin{align*}
        \epsilon &> 2 \log{\left(1+ e^{\rand}\left(F^{(n)},\Phi^{(n)}\right)\right)}\\
        \text{and } k&> \frac{2}{\epsilon}(\mathsf{log} N +n\log{|\mathcal{S}|}),
    \end{align*}
    there exists a family of (deterministic) codes $\{(f^{(n)}_i,\phi^{(n)}_i)\}_{i=1}^{k}$ such that
    \begin{equation} \label{eq:e}
        \frac{1}{k}\sum_{i=1}^{k} e^{\det}_m(f^{(n)}_i,\phi^{(n)}_i,s^n)  < \epsilon \quad \text{ for every }  m \in \mathcal{M} \text{ and } s^n\in\mathcal{S}^n.
    \end{equation}
\end{lemma}

Invoking the above lemma on the random code $\left(F^{(n)},\Phi^{(n)}\right)$ with $k = n^2$, for large enough $n$, we get a sequence of codes $\left\{\left(f^{(n)}_i,\phi^{(n)}_i\right)\right\}_{i=1}^{k}$ satisfying
    \begin{equation}
        \frac{1}{n^2}\sum_{i=1}^{n^2} e^{\det}_m\left(f^{(n)}_i,\phi^{(n)}_i,s^n\right)  < \epsilon \quad \text{ for every }  m \in \mathcal{M} \text{ and } s^n \in\mathcal{S}^n.
    \end{equation}
Since by assumption, positive authenticated communication rates are achievable, we see that for large $n$, there exists a sequence of $\left(n^2,l_n\right)$ codes $\left(\tilde{f}^{(l_n)},\tilde{\phi}^{(l_n)}\right)$ with $\lim_{n\rightarrow \infty}\frac{l_n}{n}\rightarrow 0$ and satisfying 
\begin{align*}
\bar{e}^{\det}\left(\tilde{f}^{(l_n)},\phi^{(l_n)}\right) \leq \epsilon.
\end{align*} 
We define a $\left(2^{nR},n+l_n\right)$ stochastic code $\left(F^{(n+l_n)},\phi^{(n+l_n)}\right)$ as follows: For $J\sim \text{Unif}[1:n^2]$ and $m\in \{1,2,\ldots,2^{nR}\}$, let the stochastic encoder map be
\begin{align*}
F^{(n+l_n)}(m) &= \left(\tilde{f}^{(l_n)}(J), f^{(n)}_{J}(m)\right).
\end{align*}
For $v^{n+l_n}\in \mathcal{V}^{n+l_n}$, the decoder map is defined as
\begin{displaymath}
\phi^{(n+l_n)}(v^{n+l_n}) = 
\left\{\begin{array}{ll}
        m &\text{ if } \phi^{(l_n)}\left(v^{l_n}\right) = i \text{ and }\phi_i^{(n)}\left(v_{l_n+1}^{n+l_n}\right)= m \text{ for some }i \in [1:n^2] \text{ and }m \in \left[1:2^{nR}\right],\\
        \bot &\text{ otherwise.}   
\end{array}
\right.
\end{displaymath}
For $m\in \mathcal{M}$, under no attack, the probability of error is
\begin{align*}
&e_m\left(F^{(n+l_n)},\phi^{(n+l_n)}, s_0^{n+l_n}\right)  = \mathbb{P}\left(\phi^{(n+l_n)}\left(V^{n+l_n}\right)\neq m\left|U^{n+l_n}=F^{(n+l_n)}(m), S^{n+l_n} = s_0^{n+l_n}\right.\right)\\
&\qquad=\frac{1}{n^2}\sum_{i=1}^{n^2}\mathbb{P}\left(\left\{\phi^{(n+l_n)}\left(V^{n+l_n}\right)\neq m\right\}\left|J = i,U^{n+l_n}=\left(\tilde{f}^{(l_n)}(i), f^{(n)}_i(m)\right), S^{n+l_n} = s_0^{n+l_n}\right.\right)\\
&\qquad=\frac{1}{n^2}\sum_{i=1}^{n^2}\left(\mathbb{P}\left(\left\{\phi^{(n+l_n)}\left(V^{n+l_n}\right)\neq m\right\}\cap\left\{\phi^{(l_n)}\left(V^{l_n}\right)=i\right\}\left|J = i,U^{n+l_n}=\left(\tilde{f}^{(l_n)}(i), f^{(n)}_i(m)\right), S^{n+l_n} = s_0^{n+l_n}\right.\right)\right.\\
&\qquad\qquad +\left.\mathbb{P}\left(\left\{\phi^{(n+l_n)}\left(V^{n+l_n}\right)\neq m\right\}\cap\left\{\phi^{(l_n)}\left(V^{l_n}\right)\neq i\right\}\left|J = i,U^{n+l_n}=\left(\tilde{f}^{(l_n)}(i), f^{(n)}_i(m)\right), S^{n+l_n} = s_0^{n+l_n}\right.\right)\right)\\
&\qquad \leq\frac{1}{n^2}\sum_{i=1}^{n^2}\left(\mathbb{P}\left(\left\{\phi_i^{(n)}\left(V_{l_n+1}^{n+l_n}\right)\neq m\right\}\left|J = i,U^{n+l_n}=\left(\tilde{f}^{(l_n)}(i), f^{(n)}_i(m)\right), S^{n+l_n} = s_0^{n+l_n},\phi^{(l_n)}\left(V^{l_n}\right)=i\right.\right)\right.\\
&\qquad\qquad +\left.\mathbb{P}\left(\left\{\phi^{(l_n)}\left(V^{l_n}\right)\neq i\right\}\left|J = i,U^{n+l_n}=\left(\tilde{f}^{(l_n)}(i), f^{(n)}_i(m)\right), S^{n+l_n} = s_0^{n+l_n}\right.\right)\right)\\
&\qquad=\frac{1}{n^2}\sum_{i=1}^{n^2}e^{\det}_m\left(f_i^{(n)},\phi_i^{(n)},s_0^n\right) + \frac{1}{n^2}\sum_{i=1}^{n^2}e^{\det}_i\left(\tilde{f}^{(l_n)},\phi^{(l_n)},s_0^{l_n}\right)\\
&\qquad  \leq \epsilon + \epsilon\\
&\qquad = 2\epsilon. 
\end{align*}
Suppose the adversary uses the state sequence $s^{n+l_n}\in \mathcal{S}^{n+l_n},\, s^{n+l_n} \neq s_0^{n+l_n}$. The probability of error (under attack) is
\begin{align*}
&e_m\left(F^{(n+l_n)},\phi^{(n+l_n)}, s^{n+l_n}\right)  = \mathbb{P}\left(\phi^{(n+l_n)}\left(V^{n+l_n}\right)\notin \{m,\bot\}\left|X^{n+l_n}=F^{(n+l_n)}(m), S^{n+l_n} = s^{n+l_n}\right.\right)\\
&\qquad=\frac{1}{n^2}\sum_{i=1}^{n^2}\mathbb{P}\left(\left\{\phi^{(n+l_n)}\left(V^{n+l_n}\right)\notin \{m,\bot\}\right\}\left|J = i,X^{n+l_n}=\left(\tilde{f}^{(l_n)}(i), f^{(n)}_i(m)\right), S^{n+l_n} = s^{n+l_n}\right.\right)\\
&\qquad=\frac{1}{n^2}\sum_{i=1}^{n^2}\left(\mathbb{P}\left(\left\{\phi^{(n+l_n)}\left(V^{n+l_n}\right)\notin \{m,\bot\}\right\}\cap\left\{\phi^{(l_n)}\left(V^{l_n}\right)=i\right\}\left|J = i,U^{n+l_n}=\left(\tilde{f}^{(l_n)}(i), f^{(n)}_i(m)\right), S^{n+l_n} = s^{n+l_n}\right.\right)\right.\\
&\qquad\qquad +\mathbb{P}\left(\left\{\phi^{(n+l_n)}\left(V^{n+l_n}\right)\notin \{m,\bot\}\right\}\cap\left\{\phi^{(l_n)}\left(V^{l_n}\right)=\bot\right\}\left|J = i,U^{n+l_n}=\left(\tilde{f}^{(l_n)}(i), f^{(n)}_i(m)\right), S^{n+l_n} = s^{n+l_n}\right.\right)\\
&\qquad\qquad +\left.\mathbb{P}\left(\left\{\phi^{(n+l_n)}\left(V^{n+l_n}\right)\notin \{m,\bot\}\right\}\cap\left\{\phi^{(l_n)}\left(V^{l_n}\right)\neq \{i,\bot\}\right\}\left|J = i,U^{n+l_n}=\left(\tilde{f}^{(l_n)}(i), f^{(n)}_i(m)\right), S^{n+l_n} = s^{n+l_n}\right.\right)\right)\\
&\qquad \leq\frac{1}{n^2}\sum_{i=1}^{n^2}\left(\mathbb{P}\left(\left\{\phi_i^{(n)}\left(V_{l_n+1}^{n+l_n}\right)\notin \{m,\bot\}\right\}\left|J = i,U^{n+l_n}=\left(\tilde{f}^{(l_n)}(i), f^{(n)}_i(m)\right), S^{n+l_n} = s^{n+l_n},\phi^{(l_n)}\left(V^{l_n}\right)=i\right.\right)\right.\\
&\qquad\qquad + 0 +\left.\mathbb{P}\left(\left\{\phi^{(l_n)}\left(V^{l_n}\right)\neq \{i,\bot\}\right\}\left|J = i,U^{n+l_n}=\left(\tilde{f}^{(l_n)}(i), f^{(n)}_i(m)\right), S^{n+l_n} = s^{n+l_n}\right.\right)\right)\\
&\qquad\leq\frac{1}{n^2}\sum_{i=1}^{n^2}e^{\det}_m\left(f_i^{(n)},\phi_i^{(n)},s_{l_n+1}^{n+l_n}\right) + \frac{1}{n^2}\sum_{i=1}^{n^2}e^{\det}_i\left(\tilde{f}^{(l_n)},\phi^{(l_n)},s^{l_n}\right)\\
&\qquad \leq \epsilon + \epsilon\\
&\qquad  = 2\epsilon. 
\end{align*}
Thus, 
\begin{align*}
e\left(F^{(n+l_n)},\phi^{(n+l_n)}\right) &= \max_{m}{\max_{s^{n+l_n}\in\mathcal{S}^{n+l_n}}{e_m}\left(F^{(n+l_n)},\phi^{(n+l_n)},s^{n+l_n}\right)}\\
&\leq 2\epsilon.
\end{align*}
Let $\{\epsilon_k \}_{k=1}^{\infty}$ be a decreasing sequence of positive numbers with $\lim_{k\rightarrow \infty}\epsilon_k\rightarrow 0$. For each $\epsilon_k,\,k \in\{1,2,\ldots\}$, we can find an $n$ such that there exists a $(2^{nR},n+l_n)$ stochastic codes for the adversarial channel $W_{V|U,S}$ such that $e\left(F^{(n+l_n)},\phi^{(n+l_n)}\right) \leq 2\epsilon_k$. In fact, it is possible to find a sequence of $(2^{nR},n+l_n)$ codes with $\lim_{n\rightarrow \infty}e\left(F^{(n+l_n)},\phi^{(n+l_n)}\right)\rightarrow 0$ and $\lim_{n\rightarrow \infty}\frac{\log{ N }}{n+l_n}~\rightarrow~R$.
\end{proof}
\noindent It remains to prove Lemma~\ref{lemma:random_code_reduction}.
\begin{proof}[Proof of Lemma~\ref{lemma:random_code_reduction}]
    The proof is along the lines of~\cite{AW69}~\cite[Lemma 12.8]{CK11}. We include it here for completeness.\\ 
    Draw $k$ independent samples $\{(F^{(n)}_i,\Phi^{(n)}_i)\}_{i = 1}^{k}$ with distribution $Q$.
    For $m \in \mathcal{M}$ and $s^n \in \mathcal{S}^n$,
    \begin{equation*}
        \mathbb{P}\left(\frac{1}{k}\sum_{i=1}^{k} e^{\det}_m\left(F^{(n)}_i,\Phi^{(n)}_i,s^n\right) \geq \epsilon \right) = \mathbb{P}\left( 2^{  \sum_{i = 1}^{k}e^{\det}_m\left(F^{(n)}_i,\Phi^{(n)}_i,s^n\right)} \geq 2^{k\epsilon  }\right)
    \end{equation*}
    \begin{align*}
        &\overset{\text{(a)}}{\leq} 2^{-k \epsilon}\,\mathbf{E}\left[ 2^{ \sum_{i = 1}^{k}e^{\det}_m\left(F^{(n)}_i,\Phi^{(n)}_i,s^n\right)}\right]\\    
        &\overset{\text{(b)}}{=}{\leq} 2^{-k \epsilon}\,\left(\mathbf{E}\left[ 2^{e^{\det}_m\left(F^{(n)}_1,\Phi^{(n)}_1,s^n\right)}\right]\right)^k\\
        &\overset{\text{(c)}}{\leq}{\leq} 2^{-k \epsilon}\,\left(\mathbf{E}\left[ 1+e^{\det}_m\left(F^{(n)}_1,\Phi^{(n)}_1,s^n\right)\right]\right)^k\\
        &\overset{\text{(d)}}{=}{\leq} 2^{-k \epsilon}\,\left( 1+\mathbf{E}\left[e^{\det}_m\left(F^{(n)}_1,\Phi^{(n)}_1,s^n\right)\right]\right)^k\\
        &={\leq} 2^{-k \epsilon}\,\left( 1+e^{\rand}_m\left(F^{(n)},\Phi^{(n)},s^n\right)\right)^k\\
        &=2^{-k \left(\epsilon-\log_2\left(1+e^{\rand}_m\left(F^{(n)},\Phi^{(n)},s^n\right)\right)\right)}.
    \end{align*}
    Here, $(a)$ follows by Markov's inequality, $(b)$ follows because $\{(F^{(n)}_i,\Phi^{(n)}_i)\}_{i = 1}^{k}$ is an independent and identically distributed (i.i.d.) sequence and $(c)$ follows because $2^z \leq 1 + z$ for $0\leq z \leq1$. 
    \begin{equation*}
        \mathbb{P}   \left(\frac{1}{k}\sum_{i=1}^{k} e^{\det}_m\left(F^{(n)}_i,\Phi^{(n)}_i,s^n\right) < \epsilon \text{ for every } m \in \mathcal{M} \text{ and }s^n \in \mathcal{S}^n \right) 
    \end{equation*}
    \begin{align*}
    &= 1-\mathbb{P}\left( \mathop{\mathop{\cup}_{m \in \mathcal{M}}}_{s^n \in \mathcal{S}^n} \left\{\frac{1}{k}\sum_{i=1}^{k} e^{\det}_m(F^{(n)}_i,\Phi^{(n)}_i,s^n) \geq \epsilon\right\}\right)\\
    &\geq 1-\sum_{m \in \mathcal{M},\,s^n \in \mathcal{S}^n}\exp\left(-k \left(\epsilon-\log\left(1+e^{\rand}_m\left(F^{(n)},\Phi^{(n)},s^n\right)\right)\right)\right)\\
    &\geq 1-\sum_{m \in \mathcal{M},\,s^n \in \mathcal{S}^n}\exp\left(-k \left(\epsilon-\log\left(1+ e^{\rand}\left(F^{(n)},\Phi^{(n)}\right)\right)\right)\right)\\
    & = 1 -  N  |\mathcal{S}|^{n} \exp\left(-k \left(\epsilon-\log\left(1+e^{\rand}\left(F^{(n)},\Phi^{(n)}\right)\right)\right)\right)\\
    &>0
    \end{align*}
    This guarantees the existence of a sequence of codes $\{(f^{(n)}_i,\phi^{(n)}_i)\}_{i=1}^{k}$ satisfying the error criterion \eqref{eq:e}.

\end{proof}

\end{appendices}

\end{document}